\documentclass{article}
\usepackage[utf8]{inputenc}
\usepackage[T1]{fontenc}
\usepackage{fullpage}

\usepackage[round]{natbib}

\usepackage[margin = 1in]{geometry}

\usepackage{url}            %
\usepackage{booktabs}       %
\usepackage{amsfonts}       %
\usepackage{microtype}      %
\usepackage{wrapfig}

\usepackage{authblk}  
\usepackage{subcaption} 
\usepackage{times}
\usepackage{adjustbox}
\usepackage{dsfont}
\usepackage{algorithm}
\usepackage{algpseudocode}
\usepackage{amsmath, amssymb}
\usepackage{amsthm, thmtools, thm-restate}
\usepackage{bbm}
\usepackage{graphicx}
\usepackage{latexsym}
\usepackage{mathtools}
\usepackage{multirow}
\usepackage{paralist}
\usepackage{minitoc} %
\usepackage{xspace}
\usepackage{xparse}
\usepackage[shortlabels]{enumitem}

\usepackage{tikz}
\usepackage{pgfplots}
\pgfplotsset{compat=1.18}
\usetikzlibrary{shapes}
\usetikzlibrary{positioning}
\usetikzlibrary{plotmarks}
\usetikzlibrary{patterns}
\usetikzlibrary{intersections,shapes.arrows}
\usetikzlibrary{pgfplots.fillbetween}

\usepackage[colorlinks,citecolor=bluegray,linkcolor=blue,urlcolor=blue,breaklinks]{hyperref}
\usepackage[capitalise]{cleveref}

\definecolor{orangeCB}{HTML}{FE6100}
\definecolor{maroonCB}{HTML}{DC267F}
\definecolor{blueCB}{HTML}{648FFF}
\definecolor{purpleCB}{HTML}{785ef0}
\definecolor{bluegray}{rgb}{0.04,0,0.7}
\definecolor{darkbrown}{rgb}{0.40,0.2,0.05}

\algrenewcommand\algorithmicindent{2em}
\algdef{SE}[SUBALG]{Indent}{EndIndent}{}{\algorithmicend\ }%
\algtext*{Indent}
\algtext*{EndIndent}
\algnewcommand\Input{\item[\textbf{Input:}]}
\algnewcommand\Output{\item[\textbf{Output:}]}
\newcommand{\Desc}[2]{\Statex \hspace{\algorithmicindent} $\bullet$ \textit{#1}: #2}

\crefname{section}{Section}{Sections}
\crefname{appendix}{Appendix}{Appendices}

\newtheorem{assumption}{Assumption}
\newtheorem{corollary}{Corollary}
\newtheorem{fact}{Fact}
\newtheorem{lemma}{Lemma}
\newtheorem{theorem}{Theorem}
\newtheorem{proposition}{Proposition}

\theoremstyle{remark}
\newtheorem{remark}{Remark}

\crefname{assumption}{Assumption}{Assumptions}
\crefname{corollary}{Corollary}{Corollaries}
\crefname{fact}{Fact}{Facts}
\crefname{lemma}{Lemma}{Lemmas}
\crefname{theorem}{Theorem}{Theorems}
\crefname{proposition}{Proposition}{Propositions}

\crefname{figure}{Figure}{Figures}
\crefname{table}{Table}{Tables}

\newcommand*{\Autocov}{\Gamma}  
\newcommand*{\freq}{\omega}  

\newcommand*{\SD}{\mathbf{F}}  
\newcommand*{\SDinv}{\mathbf{F}^{-1}}  
\newcommand*{\SDs}{\mathbf{F}^*}  
\newcommand*{\SDsinv}{\mathbf{F}^{*-1}}  
\newcommand*{\SDh}{\hat{\mathbf{F}}}  
\newcommand*{\SDhcomp}[1]{\hat{F}_{#1}}  

\newcommand*{\rSD}{\mathbf{A}} 
\newcommand*{\rSDs}{\mathbf{A}^*} 
\newcommand*{\rSDh}{\hat{\mathbf{A}}} 
\newcommand*{\iSD}{\mathbf{B}} 
\newcommand*{\iSDs}{\mathbf{B}^*} 
\newcommand*{\iSDh}{\hat{\mathbf{B}}} 

\newcommand*{\Ss}{\Sigma^*}  
\newcommand*{\Ssinv}{\Sigma^{*-1}}  
\newcommand*{\Sh}{\hat\Sigma}  
\newcommand*{\Sx}{\Sigma_1}  
\newcommand*{\Sxs}{\Sigma_1^*}  
\newcommand*{\Sxsinv}{\Sigma_1^{*-1}}  
\newcommand*{\Sxh}{\hat\Sigma_1}  
\newcommand*{\Sxhinv}{\hat\Sigma_1^{-1}}  
\newcommand*{\Sy}{\Sigma_2}  
\newcommand*{\Sys}{\Sigma_2^*}  
\newcommand*{\Sysinv}{\Sigma_2^{*-1}}  
\newcommand*{\Syh}{\hat\Sigma_2}  
\newcommand*{\Syhinv}{\hat\Sigma_2^{-1}}  
\newcommand*{\matS}{\mathbf{S}}  

\newcommand*{\D}{\Delta}  
\newcommand*{\Ds}{\Delta^*}  
\newcommand*{\Dh}{\hat\Delta}  
\newcommand*{\Dt}{\tilde{\Delta}}  
\newcommand*{\Dtgen}[1]{\tilde{\Delta}_{\EEcoef, #1}}  

\newcommand*{\dbds}{\mathbf{v}^*}  
\newcommand*{\dbdsT}{\mathbf{v}^{*\top}}  
\newcommand*{\mdbds}{\mathbf{w}^*}  
\newcommand*{\mdbdsT}{\mathbf{w}^{*\top}}  
\newcommand*{\dbdh}{\hat{\mathbf{v}}}  
\newcommand*{\dbdhT}{\hat{\mathbf{v}}^{\top}}  
\newcommand*{\mdbdh}{\hat{\mathbf{w}}}  

\newcommand*{\asympSD}{\sigma}
\newcommand*{\asympSDh}{\hat{\sigma}} 
\newcommand*{\VarSh}{\Psi}
\newcommand*{\VarSxh}{\Psi_1^*}
\newcommand*{\VarSyh}{\Psi_2^*}

\newcommand*{\effN}{N}
\newcommand*{\relN}{\theta}  

\newcommand*{\lossD}{L_D}  
\newcommand*{\gradD}{\nabla L_D}  
\newcommand*{\gradDL}{\nabla L_{0}}  
\newcommand*{\gradDR}{\nabla L_{1}}  

\newcommand*{\EEcoef}{\eta}  
\newcommand*{\lossDgen}{L_\eta}  
\newcommand*{\gradDgen}{\nabla L_\eta}  

\newcommand*{\mineigen}{\nu} 
\newcommand*{\Sxhrate}{\delta_{\textnormal{LD}, 1}}  
\newcommand*{\Syhrate}{\delta_{\textnormal{LD}, 2}}  
\newcommand*{\Hessrate}{\delta_{\textnormal{LD}, H}}  
\newcommand*{\BEprob}{\varepsilon_{\textnormal{BE}}}
\newcommand*{\LDprob}{\varepsilon_{\textnormal{LD}}}
\newcommand*{\LDrate}{\delta_{\textnormal{LD}}}
\newcommand*{\LDlrate}{\delta_{\textnormal{LD},l}} 
\newcommand*{\Mxhrate}{\delta_{\textnormal{M},1}} 
\newcommand*{\Myhrate}{\delta_{\textnormal{M},2}} 
\newcommand*{\Vxhrate}{\delta_{\textnormal{V},1}} 
\newcommand*{\Vyhrate}{\delta_{\textnormal{V},2}} 
\newcommand*{\Gxhrate}{\delta_{\textnormal{G},1}} 
\newcommand*{\Gyhrate}{\delta_{\textnormal{G},2}} 

\newcommand*{\sparsity}{s}  
\newcommand*{\Cone}{\mathcal{C}}
\newcommand*{\Supp}{\mathcal{S}}
\newcommand*{\Conewidth}{\kappa}  

\newcommand*{\regD}{\lambda_1}  
\newcommand*{\regDmax}{\lambda_{1\mathrm{max}}}  
\newcommand*{\regH}{\lambda_2}  

\newcommand*{\dependence}{\delta}  
\newcommand*{\damrho}{\rho}  
\newcommand*{\dammax}{\Phi}
\newcommand*{\filter}{g}  
\newcommand*{\inno}{\epsilon}  
\newcommand*{\daOalpha}{\alpha}  
\newcommand*{\daOpsi}{\psi}  

\newcommand*{\skewness}{\gamma} 

\newcommand*{\ComplexNormal}{\mathcal{CN}}  
\newcommand*{\Normal}{\mathcal{N}}  

\newcommand*{\calA}{\mathcal{A}}  

\newcommand*{\rb}[1]{(#1)}  
\newcommand*{\lrrb}[1]{\left(#1\right)}  
\newcommand*{\bigrb}[1]{\big(#1\big)}
\newcommand*{\Bigrb}[1]{\Big(#1\Big)}

\newcommand*{\cb}[1]{\{#1\}}  
\newcommand*{\lrcb}[1]{\left\{#1\right\}}  

\newcommand*{\biggcb}[1]{\bigg\{#1\bigg\}}

\newcommand*{\sqb}[1]{[#1]}  
\newcommand*{\lrsqb}[1]{\left[#1\right]}  

\newcommand*{\Bigsqb}[1]{\Big[#1\Big]}

\newcommand*{\abs}[1]{\lvert#1\rvert}  
\newcommand*{\lrabs}[1]{\left\lvert#1\right\rvert}  

\newcommand*{\norm}[1]{\lVert#1\rVert}  
\newcommand*{\lrnorm}[1]{\left\lVert#1\right\rVert}  

\newcommand*{\onenorm}[1]{\left\lVert#1\right\rVert_1}  
\newcommand*{\infnorm}[1]{\left\lVert#1\right\rVert_{\infty}}  
\newcommand*{\Lqnorm}[1]{\left\lVert#1\right\rVert_{L_q}}  

\newcommand*{\Basis}{\mathbf{E}}  
\newcommand*{\Commutation}{\mathbf{K}}  
\newcommand*{\Hessian}{\mathbf{H}}  
\newcommand*{\Identity}{\mathbf{I}}  
\newcommand*{\Permutation}{\mathbf{P}}  
\newcommand*{\Rotation}{\mathbf{R}}  

\newcommand*{\matA}{\mathbf{A}}
\newcommand*{\matB}{\mathbf{B}}
\newcommand*{\matGs}{\mathbf{G}^*} 
\newcommand*{\matGh}{\hat{\mathbf{G}}} 

\newcommand*{\matVs}{\mathbf{V}^*} 
\newcommand*{\matVsT}{\mathbf{V}^{* \top}} 
\newcommand*{\matVh}{\hat{\mathbf{V}}} 
\newcommand*{\matVhT}{\hat{\mathbf{V}}^{\top}} 
\newcommand*{\matMs}{\mathbf{M}^*} 
\newcommand*{\matMsT}{\mathbf{M}^{* \top}} 
\newcommand*{\matMh}{\hat{\mathbf{M}}} 
\newcommand*{\matMhT}{\hat{\mathbf{M}}^{\top}} 
\newcommand*{\matX}{\mathbf{X}}
\newcommand*{\matY}{\mathbf{Y}}
\newcommand*{\matZ}{\mathbf{Z}}

\newcommand*{\complex}{\mathbb{C}}
\newcommand*{\reals}{\mathbb{R}}
\newcommand*{\integers}{\mathbb{Z}}

\renewcommand*{\Im}{\operatorname{Im}}  
\renewcommand*{\Re}{\operatorname{Re}}  

\renewcommand*{\vec}{\operatorname{vec}}  

\newcommand*{\EE}{\mathbb{E}}
\newcommand*{\PP}{\mathbb{P}}
\DeclareMathOperator{\Cov}{Cov}
\DeclareMathOperator{\Var}{Var}

\newcommand*{\basis}{\mathbf{e}}  

\newcommand*{\vecv}{\mathbf{v}}
\newcommand*{\vecw}{\mathbf{w}}

\def\measurehat#1{%
    \setbox0=\vbox{$\hat{#1}\hfil\break$\null\par
        \setbox0=\lastbox\unskip\unpenalty\global\setbox1=\lastbox}%
    \setbox0=\hbox{\unhbox1 \unskip\unpenalty\unskip \global\setbox2=\lastbox}%
    \setbox0=\vbox{\unvbox2 \setbox0=\lastbox}%
}
\def\doublehat#1{%
    \measurehat{#1}\dimen0=\wd0 \measurehat{\kern0pt#1}%
    \raise.35ex\rlap{\kern\dimexpr\dimen0-\wd0$\hat{\phantom{#1}}$}{\hat#1}%
}

\def\measuretilde#1{%
    \setbox0=\vbox{$\tilde{#1}\hfil\break$\null\par
        \setbox0=\lastbox\unskip\unpenalty\global\setbox1=\lastbox}%
    \setbox0=\hbox{\unhbox1 \unskip\unpenalty\unskip \global\setbox2=\lastbox}%
    \setbox0=\vbox{\unvbox2 \setbox0=\lastbox}%
}
\def\doubletilde#1{%
    \measuretilde{#1}\dimen0=\wd0 \measuretilde{\kern0pt#1}%
    \raise.35ex\rlap{\kern\dimexpr\dimen0-\wd0$\tilde{\phantom{#1}}$}{\tilde#1}%
}

\doparttoc

\title{Assumption-Lean Inference for Spectral Differential Network  Analysis of High-Dimensional Time Series}
\author[1]{Michael Hellstern}
\author[2,3]{Byol Kim}
\author[1,4]{Ali Shojaie}

\affil[1]{Department of Biostatistics, University of Washington}
\affil[2]{Department of Statistics, Sookmyung Women's University}
\affil[3]{Research Institute of Natural Science, Sookmyung Women's University}
\affil[4]{Department of Statistics, University of Washington}
\date{}

\begin{document}
\maketitle
\faketableofcontents

\begin{abstract}
Network analysis for multivariate time series is popular in many fields, from neuroscience to seismology. The inverse spectral density is a common choice for time series network analysis due to its representation of the frequency domain correlation between two variables after removing the best linear predictor of all other variables. In many applications, the goal is to study how these networks change across different conditions. For example, in neuroscience, one might be interested in how the brain connectivity network changes before and after stimulation. Towards this goal, we develop an inference framework based on a direct estimate of the difference in two high-dimensional inverse spectral densities. We develop a new Gaussian approximation error bound for any de-biased D-trace estimation procedure which is then leveraged to both inform optimal window sizes of Welch's estimators of the spectral density and establish asymptotic normality of our de-biased D-trace estimator. Moreover, we develop an efficient algorithm based on a generalized D-trace estimation procedure to overcome the computational complexity of high-dimensional inference. The method is illustrated on synthetic data experiments and on experiments with electroencephalography data.
\end{abstract}


\section{Introduction}
\label{sec:intro}
Multivariate time series data arise across a wide range of scientific disciplines. A fundamental question in the analysis of multivariate data is the dependence structure among observed variables. This structure is often represented as a probabilistic graphical model---an undirected network whose nodes represent the variables and edges reflect the dependencies between them. In time series analysis, it is common to study this dependence structure by decomposing the signal into contributions from specific frequencies via the spectral decomposition.

The spectral signal decomposition has yielded important insights in areas as diverse as seismology \citep{james2017improved}, oceanography \citep{laurindo2019cross}, and neuroscience \citep{bloch2022network}. Spectral analysis  techniques have been especially useful in neuroscience, where brain activities captured at different frequencies reveal oscillations in sensory-cognitive processes \citep{fries2015rhythms}. 

While there are several ways to define a network based on spectral decompositions \citep{belghazi2018mutual, ursino2020transfer}, an informative  approach is to define the \emph{spectral network} based on the \emph{inverse spectral density} \citep{hellstern2025spectral}. The spectral density of a time series process is the discrete-time Fourier transform of the autocovariance function; as such, it carries information about the pairwise coherences. However, just as a pairwise correlations contain indirect effects of the other nodes in the network, so do pairwise coherences, which reduces their suitability as a measure of dependence. By contrast, a component of the inverse spectral density represents, up to scaling, the pairwise partial coherence in which the linear effects of the other nodes have been removed \citep{dahlhaus2000graphical}. Therefore, the network defined by the inverse spectral density may better reflect effective brain connectivity \citep{friston2011functional}, enabling a preliminary understanding of causal relationships.

In many scientific domains, including neuroscience, scientific questions often involve comparisons across conditions, e.g., treated versus control. As a concrete example, \citet{bloch2022network} studied the changes in the network structure induced by optogenetic stimulation by comparing the neural activity from a micro-electrocorticography ($\mu$ECoG) array under stimulation to that in the resting state. It is natural to formulate such questions as one of the difference between two networks, i.e., their \emph{differential network} \citep{shojaie2021differential}. However, in high-dimensional settings, in which the obtainable sample sizes are expected to be small relative to the number of variables, both estimation and inference can be complicated. In these settings, most classical estimators become ill-posed. The situation is more dire for time series data for which the effective sample sizes can be significantly smaller due to temporal dependence. Regularized estimation approaches overcome these limitations and have been used to obtain consistent estimates from high-dimensional time series \citep{banerjee2008model, bohm2009shrinkage, basu2015regularized, fiecas2019spectral, deb2024regularized}; however, developing valid inference for high-dimensional time series has proved more challenging.  

In recent work, \citet{hellstern2025spectral} proposed the \emph{spectral D-trace difference (SDD)} procedure to directly estimate a high-dimensional spectral differential network from a pair of spectral density estimates without estimating the individual inverse spectral density matrices. They proved that the SDD estimator is consistent provided that the two input spectral density estimators are consistent. Consistency of the high-dimensional spectral density estimators was established using the functional dependence framework  \citep{wu2005nonlinear, fiecas2019spectral} which has been used extensively to study dependent data \citep{wu2007limit, shao2007asymptotic, liu2010asymptotics, wu2018asymptotic}. However, \citet{hellstern2025spectral} did not investigate how to carry out statistical inference using their SDD estimator. In this paper, we close the loop by proposing a novel estimator of an edge of the differential network that is approximately unbiased and Normal.

To develop inference for high-dimensional spectral differential networks, we implement two modifications to the SDD procedure. First, Welch's estimator, rather than the smoothed periodogram, is used to estimate the spectral densities. Second, de-biasing is used after regularized estimation of the differential network to perform inference. It is well known that regularization introduces non-ignorable estimation bias \citep{knight2000asymptotics}. To remedy this, several methods, including \citet{vanDeGeer2014asymptotically, javanmard2014confidence, ning2017general, neykov2018unified, xia2026statistical}, rely on de-biasing to perform statistical inference. We start from the same general principles to build a new inferential framework for our high-dimensional spectral differential network.

Our main technical contributions, presented in Theorems~\ref{thm:BE:projWelch} and \ref{thm:BE:SDD}, are twofold: a Gaussian approximation error bound for any de-biased D-trace estimation procedure and a new optimal window size for Wald-type inference derived from a Berry-Esseen-type inequality for a sum of low-dimensional projections of Welch's estimators. The former can be understood as providing a unified analysis framework for any Wald-type inference procedure built on top of high-dimensional D-trace estimation \citep{zhang2014sparse, yuan2017differential}. The latter is a consequence of a considerable refinement of \citet{zhang2025spectral}'s analysis. In particular, the optimal rate that we find $B \asymp T^{1/2}$ for Gaussian approximation error minimization contrasts with the previous recommendation of $B \asymp T^{1/3}$ from \citet{zhang2025spectral}, which was obtained for a different aim of balancing the trade-off between the stochastic deviation bound and the bias order.

To achieve scalable inference, we propose variants of our method based on asymmetric estimating equations, which are also covered by our general theoretical framework. This is because both the original and the new variants can be regarded as special cases of a generalized D-trace estimation framework. The proposed generalized D-trace estimation allows us to reduce the computational complexity of inference for $p$ variables from $O(p^6)$ to $O(p^3)$ and may be of independent interest. Finally, we provide a novel implementation that deals with many computational challenges of high-dimensional spectral differential inference through efficient use of the sparsity structure with careful index tracking. The resulting method is implemented into a publicly available R-package, \verb|sdd|. 

The remainder of this paper is organized as follows. We formally state our problem in \cref{sec:problem} and describe our proposal, the de-biased SDD procedure, in \cref{sec:method}. In \cref{sec:theory}, we turn to theoretical considerations, including a Gaussian approximation result under a weak functional dependence condition. The efficacy of our method is empirically demonstrated, first with simulated data in \cref{sec:simulations} and then with real EEG data in \cref{sec:eeg}. Finally, we conclude with a few discussion points in \cref{sec:discussion}.

\paragraph{Notation.} A complete table of notation is provided in \cref{tab:notation,tab:notation:cont}. For a matrix $\matA \in \complex^{m \times n}$, let $\Re \matA \in \reals^{m \times n}$ and $\Im \matA \in \reals^{m \times n}$ be the real and imaginary parts of $\matA$. The absolute value of a complex number is defined as $\abs{a} = \sqrt{(\Re a)^2 + (\Im a)^2}$. We use the notation $\overline{\matA}$ to represent the complex conjugate of $\matA$. An eigenvalue of a matrix is denoted as $\mineigen$. We define the $\ell_1$, $\ell_2$, and $\ell_{\infty}$ norms of a matrix $\matA = (a_{ij})$ as $\onenorm{\matA} = \sum_{i,j} \abs{a_{ij}}$, $\norm{\matA}_2 = (\sum_{i,j} \abs{a_{ij}}^2)^{1/2}$, and $\infnorm{\matA} = \max_{i,j} |a_{ij}|$, respectively. The $\ell_0$ ``norm'', $\norm{\matA}_0$, is defined as the number of nonzero elements of $\matA$. For $q \geq 1$, we denote the $L_q$-norm of a random variable $X$ as $\Lqnorm{X} = \rb{ \EE\abs{X}^q }^{1/q}$ where $\EE$ is the expected value. The Kronecker product between two matrices is denoted as $\matA \otimes \matB$. We also use $\vec(\matA)$ to denote the operator that stacks the columns of $\matA$. The notation $x_t \asymp y_t$ indicates that $x_t$ and $y_t$ are asymptotically of the same order while $x \lesssim y$ denotes that $x \leq Cy$ for some constant $C > 0$ and $x \lesssim_{\alpha} y$ indicates that $x \leq C_{\alpha} y$ where $C_{\alpha}$ is a constant that depends on the problem parameter $\alpha$. The $p \times p$ identity matrix is denoted as $\Identity_p$. The gradient of $\lossD$ with respect to the parameters $\D$ is denoted as $\gradD(\D; \Sxh, \Syh)$; arguments after the semicolon are treated as fixed or conditioning quantities. A full list of major definitions can be found in \cref{tab:notation} at the beginning of the appendix.

\section{Problem}
\label{sec:problem}
Let $X_{1, \cdot} = \rb{X_{1, t}}_{t \in \integers}$ and $X_{2, \cdot} = \rb{X_{2, t}}_{t \in \integers}$ be two $p$-dimensional stationary time series, i.e., $X_{l, t} \in \reals^p$ for $l \in \cb{1, 2}$; without loss of generality, we assume $\EE\sqb{X_{l, t}} = 0$. Let $\Autocov_l(h) = \EE\sqb{X_{l, t+h} X_{l, t}^\top} \in \reals^{p \times p}$ for $h \in \integers$ be the \emph{autocovariance} function for condition $l$. Assuming $\sum_{h=-\infty}^{\infty} \abs{\Autocov_{l, ij}(h)} < \infty$ for all $1 \leq i, j \leq p$, the \emph{spectral density} at frequency $\freq$ for condition $l$ is defined as
\begin{equation} \label{eq:SD}
\SDs_l(\freq) = \frac{1}{2 \pi} \sum_{h = -\infty}^{\infty} e^{-\imath \freq h} \Autocov_l(h).
\end{equation}
Since $\SDs_l(\freq)$ is $2 \pi$-periodic with $\SDs_l(-\freq) = \overline{\SDs_l(\freq)}$, it suffices to consider $\freq \in [0, \pi]$. Note that $\SDs_l(\freq) \in \complex^{p \times p}$; it is real if and only if $\freq$ is an integer multiple of $\pi$.

Our set-up is inspired by applications in neuroscience, where brain dynamics have been shown to be governed by neuronal oscillations spanning distinct frequency bands, each playing a different functional role \citep{buzsaki2004neuronal,fries2005mechanism,lachaux2012high}. Accordingly, frequency domain analyses have become a central approach for studying brain networks and their role in cognition and behavior \citep{michalareas2016alpha}. Specifically, as a motivating example, we consider the experiments of \citet{bloch2022network}. In this study, the researchers investigated how the brains of non-human primates respond to an optogenetic stimulation using a 96-electrode micro-electrocorticography array ($p = 96$). The data consist of a pair of high-dimensional time series, one from the resting state ($l = 1$) and the other from the stimulated state ($l = 2$). Of interest is the structural difference in brain functional connectivity networks under the two conditions. Thus, we would like to answer questions about the structure of the differential brain network, such as whether there is evidence of any change in the connectivity of nodes $i_0$ and $j_0$ between the two conditions and, if so, quantify it. Encoding the brain network in condition $l$ by the inverse spectral density $\SDsinv_l$, the differential brain network analysis amounts to carrying out statistical inference about components of $\SDsinv_1(\freq) - \SDsinv_2(\freq)$. We assume that both $\freq$ and the components of inferential interest are fixed in advance by the analyst---perhaps following results of previous experiments---and drop $\freq$ from the notation.

Since $\SDsinv_l \in \complex^{p \times p}$ for each $l$, $\SDsinv_1 - \SDsinv_2 \in \complex^{p \times p}$. To avoid dealing with complex quantities, we formulate the problem using their real representations as in \citet{hellstern2025spectral}. The \emph{realification} or \emph{real representation} of a complex matrix $\SD = \Re \SD + \imath \Im \SD \in \complex^{p \times p}$ is the real matrix $\Sigma = \lrsqb{\begin{smallmatrix} \Re \SD & -\Im \SD \\ \Im \SD & \phantom{-} \Re \SD \end{smallmatrix}} \in \reals^{2p \times 2p}$. Let $\SDinv \in \complex^{p \times p}$ be the inverse of $\SD$, i.e., $\SDinv$ is the unique complex matrix satisfying $\SD \SDinv = \SDinv \SD = \Identity_p$. It can be seen that the inverse of $\Sigma$ satisfies $\Sigma^{-1} = \lrsqb{\begin{smallmatrix} \Re\lrrb{\SDinv} & -\Im\lrrb{\SDinv} \\ \Im\lrrb{\SDinv} & \phantom{-}\Re\lrrb{\SDinv} \end{smallmatrix}}$. In other words, $\Sigma^{-1}$ is the realification of $\SDinv$. Thus, components of $\SDsinv_1 - \SDsinv_2$, can be equivalently studied by considering components of
\begin{equation} \label{eq:Ds}
\Ds = \Sxsinv - \Sysinv,
\end{equation}
where
\begin{enumerate}
\item in the case of $\freq \notin \cb{0, \pi}$, $\Sxsinv$ and $\Sysinv$ are the realifications of $\SDsinv_1$ and $\SDsinv_2$, respectively,
\item in the case of $\freq \in \cb{0, \pi}$, $\Sxsinv = \SDsinv_1$ and $\Sysinv = \SDsinv_2$.
\end{enumerate}
Henceforth, we will not necessarily distinguish between a spectral density and its realification and use the term ``spectral density" to refer to both $\SDs_l$ and $\Ss_l$; we will only make the distinction when it matters.

Our network difference estimation problem thus amounts to carrying out valid inference about $\Ds_{i_0 j_0}$, where $(i_0, j_0)$ is the index of the component of interest. To motivate our method introduced in the next section, we make the following observation: valid inference about $\Ds_{i_0 j_0}$ is straightforward if we can find an asymptotically Normal and unbiased estimator of $\Ds_{i_0 j_0}$. Unfortunately, known consistent estimators of $\Ds$ will be biased for two reasons: the high dimensionality of the problem ($p \gg T_1, T_2$) and the finite trajectory lengths $T_1, T_2$. To overcome these issues, we propose a two-stage approach that constructs a de-biased estimator of $\Ds_{i_0 j_0}$ starting from an $\ell_1$-penalized initial estimator, $\Dh$, and show that the resulting estimator is asymptotically Normal and unbiased for $\Ds_{i_0, j_0}$ subject to conditions on the sparsity of the underlying parameters and the degree of dependence in the underlying processes.

\section{Method}
\label{sec:method}
To simplify the presentation, with a slight abuse of notation, we partition $\vec(\Ds)$ as $\vec(\Ds) = (\Ds_0, \Ds_u)$ with $\Ds_0 \in \reals$ denoting the parameter of interest and $\Ds_u \in \reals^{d-1}$, where $d = \dim( \vec(\Ds) )$, as the remaining nuisance parameters.

\subsection{Direct difference estimation with the D-trace loss}

It may appear that to estimate $\Ds$, we must first estimate both $\Sxsinv$ and $\Sysinv$ and then take their difference. However, this is unnecessary; we can estimate $\Ds$ directly by taking advantage of its linear relationship with $\Sxs$ and $\Sys$. Specifically, it is easy to check that $\Ds$ is a solution to the following pair of population-level equations:
\begin{equation} \label{eq:pop:average}
\Sys \D \Sxs = \Sys - \Sxs, \quad \Sxs \D \Sys = \Sys - \Sxs.
\end{equation}
Moreover, using the fact that $\rb{\matB^\top \otimes \matA} \vec\rb{\matX} = \vec({\matA \matX \matB})$, we can write the equations in \eqref{eq:pop:average} as 
\begin{equation} \label{eq:pop:average2}
\rb{\Sxs \otimes \Sys} \vec\rb{\D} = \vec\rb{\Sys-\Sxs}, \quad \rb{\Sys \otimes \Sxs} \vec\rb{\D} = \vec\rb{\Sys-\Sxs}.
\end{equation}
Taking a convex combination of the equations in \eqref{eq:pop:average2}, and regarding the resulting equation as a stationary-point equation, we see that $\Ds$ can be characterized as a minimizer of the following population-level quadratic loss function:
\begin{equation} \label{eq:Dtrace:general}
\lossDgen\rb{\D; \Sxs, \Sys} = \frac 1 2 \vec\rb{\D}^\top \Hessian^*_\EEcoef \vec\rb{\D} - \vec\rb{\D}^\top \vec\rb{\Sys-\Sxs}, \quad \EEcoef \in [0, 1],
\end{equation}
where
\begin{equation*}
\Hessian^*_\EEcoef = \Hessian_\EEcoef \rb{\Sxs, \Sys} = \EEcoef \lrrb{\Sxs \otimes \Sys} + \lrrb{1-\EEcoef} \lrrb{\Sys \otimes \Sxs}.
\end{equation*}
The loss function in \eqref{eq:Dtrace:general} generalizes the \emph{D-trace loss} \citep{zhang2014sparse, yuan2017differential, kim2021uncertainty}, which is obtained by setting $\EEcoef = 1/2$:
\begin{equation} \label{eq:Dtrace}
\lossD\rb{\D; \Sxs, \Sys} = \frac 1 2 \vec\rb{\D}^\top \Hessian^* \vec\rb{\D} - \vec\rb{\D}^\top \vec\rb{\Sys-\Sxs},
\end{equation}
where $\Hessian^* = \Hessian\lrrb{\Sxs, \Sys} = \rb{\Sxs \otimes \Sys + \Sys \otimes \Sxs}/2$. 

As shown in \cref{subsubsec:computation:general}, the generalization in the \emph{generalized D-trace loss} function in \eqref{eq:Dtrace:general} is critical for developing a computationally-efficient inference procedure for $\Ds$ in high dimensions. Specifically, while inference for $\EEcoef = 1/2$ becomes computationally prohibitive, $\EEcoef \in \cb{0,1}$ leads to scalable inference.

Replacing $\Sxs$ and $\Sys$ in \eqref{eq:Dtrace:general} with their estimates $\Sxh$ and $\Syh$, we obtain a sample procedure for direct estimation of $\Ds$. 
A similar loss, with $\EEcoef = 1/2$, was previously used in \citet{hellstern2025spectral} for directly estimating the differential network of high-dimensional time series. The authors showed that as long as $\Sxh$ and $\Syh$ are consistent estimators of $\Sxs$ and $\Sys$, the resulting $\ell_1$-penalized D-trace estimator $\Dh$ is also consistent for $\Ds$. They further established convergence rates when $\Sxh$ and $\Syh$ are smoothed periodograms, assuming that the dependences of the underlying processes satisfy a geometric decay condition and the observations have sub-exponential tails.

\subsection{SDD inference} \label{subsec:SDDinference}

Unfortunately, the sampling distribution of an $\ell_1$-penalized estimator is notoriously intractable \citep{knight2000asymptotics} and produces non-neglible bias \citep{fan2001variable}, rendering inference based on $\Dh$ all but infeasible. Therefore, we instead use $\Dh$ as a starting point for finding another estimator $\Dt_0$ that is asymptotically Normal and unbiased for $\Ds_0$. Later in \cref{sec:theory}, we shall see that a solution to the estimating equation
\begin{equation} \label{eqn:ee}
\dbdhT \gradD\lrrb{\lrrb{\D_0, \Dh_u}; \Sxh, \Syh} = 0
\end{equation}
is an asymptotically Normal and unbiased estimator of $\Ds_0$ provided that $\Dh_u$ is a consistent estimator of the high-dimensional nuisance parameter $\Ds_u$, $\dbdh$ is a consistent estimator of $\dbds = \Hessian^{*-1} \basis_0$, the column (or row) of the \emph{inverse} of $\Hessian^* = \Hessian(\Sxs, \Sys)$ associated with the parameter of inferential interest, and the distributions of $\mdbdsT_1 \vec\rb{\Sxh}$ and $\mdbdsT_2 \vec\rb{\Syh}$ for some $\mdbds_1$ and $\mdbds_2$ to be specified are asymptotically Normal. Since $\dbds$ is also high-dimensional, we again employ a constrained or a penalized estimation approach. This yields our procedure, a pseudo-code of which is given in \cref{algo:method}.
It can be checked that $\Dt_0$ in \cref{eq:Dt} is indeed the solution to \cref{eqn:ee}.

\begin{algorithm}[t]
\caption{SDD inference} \label{algo:method}
\begin{algorithmic}[0]
\Input
  \Desc{Data}{Observed time series for each condition, $\rb{x_{1,t}}_{t=1}^{T_1}$ and $\rb{x_{2,t}}_{t=1}^{T_2}$}
  \Desc{Window sizes}{$B_1 \in \cb{1, \ldots, T_1}$ and $B_2 \in \cb{1, \ldots, T_2}$}
  \Desc{Regularization penalties}{$\regD, \regH > 0$}
\Output
  \Desc{Debiased estimate of $\Ds_0$}{$\Dt_0$}
\State
\State 1.~Estimate the spectral densities $\SDh_1$ and $\SDh_2$ at frequency $\freq$.
\Indent
    \If{$\freq \notin \cb{0, \pi}$}
        \State Realify $\SDh_1$ and $\SDh_2$ as $\Sxh$ and $\Syh$.
    \Else
        \State Set $\Sxh \gets \SDh_1$ and $\Syh \gets \SDh_2$.
    \EndIf
\EndIndent
\State 2.~Obtain an initial estimate of $\Ds$ as
  \begin{equation} \label{eqn:sdd_estimator}
    \Dh = \arg\min_{\D}\; \lossD\lrrb{\D; \Sxh, \Syh} + \regD \lrnorm{\D}_1.
  \end{equation}
\State 3.~Obtain an estimate of the projection direction as
  \begin{equation} \label{eqn:debias_direction}
    \dbdh = \arg\min_{\vecv} \lrnorm{\vecv}_1
    \quad \text{subject to} \quad
    \infnorm{\hat{\Hessian}\rb{\Sxh, \Syh} \vecv - \basis_0} \leq \regH.
  \end{equation}
\State \textbf{Output:}
  \begin{equation} \label{eq:Dt}
    \Dt_0 = \Dh_0
    - \lrsqb{\hat{\Hessian}\, \dbdh}_0^{-1}
      \dbdh^\top \gradD\!\lrrb{\Dh;\, \Sxh, \Syh}.
  \end{equation}
\end{algorithmic}
\end{algorithm}

\subsubsection{Welch's estimator of the spectral density}

As remarked above, asymptotic Normality of $\Dt_0$ is contingent on consistency of the estimator $\Dh_u$ of the high-dimensional nuisance parameter $\Ds_u$, which itself depends on consistency of the pair of estimators $\Sxh$ and $\Syh$ of the spectral densities $\Sxs$ and $\Sys$. Additionally, the distributions of $\mdbdsT_1 \vec\rb{\Sxh}$ and $\mdbdsT_2 \vec\rb{\Syh}$ for some $\mdbds_1$ and $\mdbds_2$ are required to be asymptotically Normal. Later in \cref{sec:theory}, we shall see that Welch's non-overlapping estimator can satisfy all these conditions assuming that the dependence-adjusted Orlicz-type norms of the underlying processes are bounded. \emph{Welch's estimator} of the spectral density at frequency $\freq$ for condition $l$ is given by
\begin{equation}\label{eqn:welch_estimator}
\SDh_l = \SDh_l(\freq) = \frac{1}{2 \pi T_l} \sum_{n = 1}^{T_l/B_l} \lrrb{\sum_{s = (n-1) B_l+1}^{n B_l} e^{-\imath \freq s} X_{l, s}} \lrrb{\sum_{t = (n-1) B_l+1}^{n B_l} e^{\imath \freq t} X_{l, t}^\top},
\end{equation}
where $B_l \in \cb{1, \ldots, T_l}$ is the window size, which is a hyperparameter to be chosen or tuned by the user. The effect of this choice is discussed in \cref{subsubsec:Welch}. The value recommended by our theoretical result in \cref{cor:BE:SDD:abbr} is $B_l \asymp \sqrt{T_l}$.

Once we obtain $\Sxh$ and $\Syh$, we plug them into \cref{eqn:sdd_estimator,eqn:debias_direction} to obtain $\Dh$ and $\dbdh$. For the former, we use the alternating direction method of multipliers (ADMM) algorithm from \citet{yuan2017differential}, selecting $\regD$ by minimizing eBIC \citep{foygel2010extended}. The latter is a subproblem of CLIME, a procedure for sparse matrix inversion \citep{cai2011constrained}, which is known to require substantial memory and have $O(p^6)$ computational complexity. Therefore, in our implementations, we rely on GLASSO \citep{friedman2008sparse} instead. This switch has no effect on our theory, as the order of convergence is unchanged \citep{bickel2009simultaneous}.

\subsubsection{Details of the inference procedure} \label{subsubsec:inference:details}

Our \cref{thm:BE:SDD} and \cref{cor:BE:SDD:abbr} immediately imply that
\begin{equation*}
\PP\lrcb{\asympSD^{-1} \sqrt{\frac{T_1}{B_1} + \frac{T_2}{B_2}} \lrrb{\Dt_0 - \Ds_0} \leq z} \approx \PP\lrcb{Z \leq z},
\end{equation*}
where $Z \sim \Normal\rb{0, 1}$ and $\asympSD^2 > 0$ is the asymptotic variance, which is shown to be $\asympSD^2 = \rb{\effN / \effN_1} \asympSD^2_1 + \rb{\effN / \effN_2} \asympSD^2_2$, where $\asympSD^2_l = \rb{\Permutation \mdbds_l}^\top \VarSh_l \Permutation \mdbds_l$, $l \in \cb{1, 2}$, in the case of $\freq \notin \cb{0, \pi}$, and $\asympSD^2_l = \mdbdsT_l \VarSh_l \mdbds_l$, $l \in \cb{1, 2}$, in the case of $\freq \in \cb{0, \pi}$. Moreover,
\begin{equation*}
\VarSh_l = \frac 1 2 \begin{bmatrix}
\lrrb{\Identity_{p^2} + \Commutation} \lrrb{\Re \SDs_l \otimes \Re \SDs_l + \Im \SDs_l \otimes \Im \SDs_l} &
\lrrb{\Identity_{p^2} + \Commutation} \lrrb{\Im \SDs_l \otimes \Re \SDs_l - \Re \SDs_l \otimes \Im \SDs_l} \\
\lrrb{\Identity_{p^2} - \Commutation} \lrrb{\Re \SDs_l \otimes \Im \SDs_l - \Im \SDs_l \otimes \Re \SDs_l} &
\lrrb{\Identity_{p^2} - \Commutation} \lrrb{\Re \SDs_l \otimes \Re \SDs_l + \Im \SDs_l \otimes \Im \SDs_l}
\end{bmatrix}, \quad l \in \cb{1, 2},
\end{equation*}
in the case of $\freq \notin \cb{0, \pi}$ and
\begin{equation*}
\VarSh_l = \rb{\Identity_{p^2} + \Commutation} \rb{\SDs_l \otimes \SDs_l}, \quad l \in \cb{1, 2}
\end{equation*}
in the case of $\freq \in \cb{0, \pi}$, with $\Commutation = \sum_{i = 1}^{p} \sum_{j = 1}^{p} \Basis_{ij} \otimes \Basis_{ji}$, where $\Basis_{ij}$ is the $(i, j)$-th standard basis matrix in $\reals^{p \times p}$, i.e., $\Basis_{ij}$ is the $p \times p$ matrix with 1 on its $(i, j)$ entry and 0 everywhere else; finally, 
\begin{align*}
\mdbds_1 & = -\frac 1 2 \lrrb{\Identity_{p'} \otimes \Ds \Ss_{2} + \Ds \Ss_{2} \otimes \Identity_{p'} + 2 \Identity_{{p'}^2}} \dbds \\
\mdbds_2 & = -\frac 1 2 \lrrb{\Identity_{p'} \otimes \Ds \Ss_{1} + \Ds \Ss_{1} \otimes \Identity_{p'} - 2 \Identity_{{p'}^2}} \dbds,
\end{align*}
where $p' = 2p$ in the case of $\freq \notin \cb{0, \pi}$ and $p' = p$ in the case of $\freq \in \cb{0, \pi}$; and $\Permutation = \lrsqb{\Permutation_1\ \Permutation_2}$ with
\begin{equation*}
\Permutation_1 = \begin{bmatrix}
  \rule[-0.5ex]{0.5pt}{2.5ex} & \rule[-0.5ex]{0.5pt}{2.5ex} & \rule[-0.5ex]{0.5pt}{2.5ex} & \rule[-0.5ex]{0.5pt}{2.5ex} & & \rule[-0.5ex]{0.5pt}{2.5ex} & \rule[-0.5ex]{0.5pt}{2.5ex} \\
  \basis_1 & \basis_{p+1} & \basis_2 & \basis_{p+2} & \cdots & \basis_p & \basis_{2p} \\
  \rule[-0.5ex]{0.5pt}{2.5ex} & \rule[-0.5ex]{0.5pt}{2.5ex} & \rule[-0.5ex]{0.5pt}{2.5ex} & \rule[-0.5ex]{0.5pt}{2.5ex} & & \rule[-0.5ex]{0.5pt}{2.5ex} & \rule[-0.5ex]{0.5pt}{2.5ex}
\end{bmatrix} \otimes \Identity_p, \quad
\Permutation_2 = \begin{bmatrix}
  \rule[-0.5ex]{0.5pt}{2.5ex} & \rule[-0.5ex]{0.5pt}{2.5ex} & \rule[-0.5ex]{0.5pt}{2.5ex} & \rule[-0.5ex]{0.5pt}{2.5ex} & & \rule[-0.5ex]{0.5pt}{2.5ex} & \rule[-0.5ex]{0.5pt}{2.5ex} \\
  -\basis_{p+1} & \basis_1 & -\basis_{p+2} & \basis_2 & \cdots & -\basis_{2p} & \basis_p \\
  \rule[-0.5ex]{0.5pt}{2.5ex} & \rule[-0.5ex]{0.5pt}{2.5ex} & \rule[-0.5ex]{0.5pt}{2.5ex} & \rule[-0.5ex]{0.5pt}{2.5ex} & & \rule[-0.5ex]{0.5pt}{2.5ex} & \rule[-0.5ex]{0.5pt}{2.5ex}
\end{bmatrix} \otimes \Identity_p,
\end{equation*}
$\basis_i$ is the $i$-th standard basis vector in $\reals^{2p}$ and $-\basis_i$ is the negation of $\basis_i$. A natural estimator of $\asympSD^2$ is the plug-in estimator $\hat{\asympSD}^2$ that replaces the unknown parameters by their consistent estimates, all of which we already have from intermediate calculations of \cref{algo:method}. \cref{lem:plugin_var_est} in \cref{app:plugin_var_est} shows that this plug-in variance estimator $\hat{\asympSD}^2$ is consistent in the setting of \cref{cor:BE:SDD:abbr}. With $\hat{\asympSD}^2$, valid Normal-approximation-based inference for $\Ds_0$ is possible. For example, an asymptotic test of size $\alpha$ of the null hypothesis that $\Ds_0 = \delta_0$ versus the alternative hypothesis $\Ds_0 \neq \delta_0$ is given by the test that rejects $H_0$ whenever
\begin{equation*}
\hat{\asympSD}^{-1} \sqrt{\frac{T_1}{B_1} + \frac{T_2}{B_2}} \lrabs{\Dt_0 - \delta_0} > z_{\alpha/2}.
\end{equation*}
A $100 (1-\alpha) \%$ confidence interval for $\Ds_0$ is given by the interval
\begin{equation*}
\Dt_0 \pm z_{\alpha/2} \frac{\hat{\asympSD}}{\sqrt{\frac{T_1}{B_1} + \frac{T_2}{B_2}}}.
\end{equation*}

\subsection{Computational considerations}

In a typical application, the analyst would have more than a single inferential target of interest. Our procedure easily extends to handle such applications: \cref{eqn:debias_direction} is repeated for each inferential target of interest. This amounts to obtaining many columns (or rows) of the inverse of $\hat{\Hessian}$, and the cost can quickly accumulate. In addition, the analyst must also estimate the standard errors. Although our proposed plug-in estimator does not require estimation of additional parameters, it involves several high-dimensional matrix-vector multiplications. In the following, we discuss our strategies for alleviating the computational burden.

\subsubsection{Fast Hessian inversion via alternative estimating equations} 
\label{subsubsec:computation:general}

First, we discuss alternative approaches to the de-biasing step \cref{eqn:debias_direction} based on other estimating equations for $\Ds$. As remarked above, \cref{eqn:debias_direction} can be regarded as a subproblem of the CLIME procedure for sparse matrix inversion and can be replaced by GLASSO, which yields an estimator that has the same order of convergence and is often computationally more advantageous, especially for large number of targets. Alternatively, we can avoid the large matrix inversion all together by thinking of $\Ds$ as a solution to a simpler estimating equation.

Recall that the generalized D-trace loss in \cref{eq:Dtrace:general} is valid for any  $\EEcoef$. Of  particular interest are equations obtained by setting $\EEcoef = \cb{0,1}$:
\begin{align}
\gradDL\rb{\D; \Sxh, \Syh} &= \lrrb{\Syh \otimes \Sxh} \vec\rb{\D} - \vec\lrrb{\Syh-\Sxh} = 0, \label{eq:pop:1} \\
\gradDR\rb{\D; \Sxh, \Syh} &= \lrrb{\Sxh \otimes \Syh} \vec\rb{\D} - \vec\lrrb{\Syh-\Sxh} = 0. \label{eq:pop:0}
\end{align}
The notations $\gradDL$ and $\gradDR$ make clear that each estimating equation comes from minimizing a quadratic loss. For these two estimating equations, the corresponding Hessian is a Kronecker product of two positive semi-definite matrices
\begin{equation*}
\hat{\Hessian}_{0} = \Syh \otimes \Sxh, \quad \hat{\Hessian}_{1} = \Sxh \otimes \Syh.
\end{equation*}
This suggests that we can leverage the well-known property $\rb{\matA \otimes \matB}^{-1} = \matA^{-1} \otimes \matB^{-1}$ and estimate
\[
\hat{\Hessian}_{0}^{-1} = \Syhinv \otimes \Sxhinv, \quad \hat{\Hessian}_{1}^{-1} = \Sxhinv \otimes \Syhinv.
\]
This approach requires inverting two $p' \times p'$ matrices instead of a single ${p'}^2 \times {p'}^2$ matrix, where $p' = 2p$ in the case of $\freq \notin \cb{0, \pi}$ and $p' = p$ in the case of $\freq \in \cb{0, \pi}$. This leads to significant computational scalability. For instance, when using CLIME to estimate $\Sxhinv$ and $\Syhinv$, it reduces the computational complexity from $O({p}^6)$ to $O({p}^3)$. 

We next discuss the effect of switching to a different estimating equation on \cref{algo:method} and the theoretical properties of our procedure. 

First, changing the estimating equation does not impact the other parts of \cref{algo:method}. \cref{eqn:sdd_estimator} need not be modified, as we only require that we start from a ``good enough'' estimator of $\Ds$. The final estimator, then, is
\begin{equation} \label{eq:Dt:general}
\Dtgen{0} = \Dh_0 - \lrsqb{\hat{\Hessian}_\EEcoef \dbdh_\EEcoef}_0^{-1} \dbdhT_\EEcoef \gradDgen \lrrb{\Dh; \Sxh, \Syh}, \quad \EEcoef \in \cb{0, 1},
\end{equation}
where $\dbdh_\EEcoef = \hat{\Hessian}_\EEcoef^{-1} \basis_0$. The plug-in variance estimator has the same form with $\mdbdh_{\EEcoef, 1} = -\rb{\Sxhinv \otimes \Sxhinv} \basis_0$ and $\mdbdh_{\EEcoef, 2} = \rb{\Syhinv \otimes \Syhinv} \basis_0$.

Second, the change in estimating equation also has no impact on the  theoretical properties of our procedure: $\Dtgen{0}$ can be analyzed in the same way as $\Dt_0$ with analogous properties. In particular, $\Dtgen{0}$ is also an asymptotically Normal and unbiased estimator of $\Ds_0$. There are, however, some subtleties in that the assumptions are now placed on $\Hessian^*_\EEcoef$ rather than $\Hessian^*$. See \cref{app:loss:general} for details.

\subsubsection{On-the-fly matrix-vector multiplication for the plug-in variance estimator} \label{sec:method:var:est}

Next, we turn to computational strategies for the plug-in variance estimator. At first glance, the form of the plug-in estimator appears to require the formation of a $O(p^2) \times O(p^2)$ matrices $\VarSh_l$, $l \in \cb{1, 2}$; unfortunately, such matrices could not be stored in the working memory of most personal computers for values of $p$ common in neuroscience applications. Our implementation overcomes this difficulty by exploiting structural properties of the problem. First, observe that $\dbdh$ is a sparse estimate of the inverse of $\Hessian^*$. Secondly, $\VarSh_l$ is a function of $\SDh_l = \rSDh_l + \imath \iSDh_l$, so there are at most $O(p^2)$ free parameters---not $O(p^4)$, as the na\"{i}ve dimension count would suggest. Thus, through custom matrix-vector multiplication functions with careful index bookkeeping, we are able to evaluate the variance estimate without ever forming $\VarSh_l$ directly.

Our implementation is available in the R-package \texttt{sdd}. It uses sparse matrix and vector operations \citep{Rcpp,RcppEigen} to minimize memory overhead and computation time. For more details see \cref{app:method_details}.

\section{Theoretical Results}
\label{sec:theory}
In this section, we give our main theoretical result, \cref{thm:BE:SDD}, which states that our procedure leads to an asymptotically Normal and unbiased estimator. To establish this result, we first analyze how the de-biased D-trace estimation propagates the errors in the initial covariance structure estimates in full generality; this will make clear conditions on the initial covariance structure estimators that are sufficient for the Normal approximation to hold. As a concrete candidate estimator of the spectral density that can satisfy the aforementioned conditions, we consider Welch's estimator. We prove that Welch's estimator indeed possesses the required properties as long as the degrees of dependence in the underlying processes are sufficiently mild in a sense to be made precise below. Finally, we combine the analysis of the de-biased D-trace information geometry with the statistical properties of Welch's estimators to arrive at a Normal approximation error bound for the final estimator and a sufficient condition for the distributional convergence. A plug-in estimator for the variance of the final estimator is also proposed and shown to be consistent.

\subsection{Information geometry of de-biased D-trace estimation} \label{subsec:infogeom}

In this subsection, we examine how the information geometry of de-biased D-trace estimation propagates the errors in the initial estimates of covariance structures, which leads to a Berry-Esseen inequality for a general de-biased D-trace estimator (\cref{lem:BE:Dtrace:abbr}). While our primary interest is in applying this conclusion to our problem of comparing inverse spectral densities of two time series, the analysis presented here is more general and applies to any de-biased D-trace estimation procedure, spectral or otherwise. To simplify exposition, we base our analysis on the standard \emph{D-trace loss} in \cref{eq:Dtrace}, but when relevant we refer the reader to analogous results for the \emph{generalized D-trace loss} from \cref{eq:Dtrace:general}; the generalized D-trace loss with $\EEcoef \in \cb{0,1}$ is necessary for scalable inference, as discussed in \cref{subsubsec:computation:general}.

Recall that our estimator is defined as a solution to a \emph{projected} high-dimensional estimating equation (Eq.~\ref{eqn:ee}). What benefit does this additional projection step provide? The key to understanding this step is in how we defined the direction of projection $\dbdh$ in \cref{eqn:debias_direction}. Since the sample Hessian $\hat{\Hessian}$ can be thought of as an estimator of the true Hessian $\Hessian^*$, $\dbdh$ can be thought of as an estimator of $\dbds = \Hessian^{*-1} \basis_0$, the column (or row) of the \emph{inverse} of $\Hessian^*$ corresponding to the parameter of interest. The distributional properties of the LASSO---of which our initial estimator $\Dh$ is an example---are known to be highly complicated and intractable \citep{knight2000asymptotics, potscher2009distribution}. However, we shall see that projecting $\gradD\rb{\D; \Sxh, \Syh}$ along the direction of $\dbds$ removes the estimation error in the high-dimensional nuisance parameter $\Ds_u$, leading to an asymptotic linear representation for $\Dt_0$, which in turn can be leveraged to yield a Normal approximation bound. In the case of the D-trace loss, this asymptotic linear representation turns out to be a sum of a pair of one-dimensional projections of spectral density estimators.

We start from a simple observation to derive the asymptotic linear representation for $\Dt_0$. Because the D-trace loss is quadratic in $\D$, we have
\begin{equation*}
\dbdh^\top \gradD\lrrb{\Dh; \Sxh, \Syh} = \dbdh^\top \gradD\lrrb{\Ds; \Sxh, \Syh} + \dbdh^\top \hat{\Hessian} \vec\lrrb{\Dh-\Ds}.
\end{equation*}
Recalling that $\dbdh \approx \dbds = \Hessian^{*-1} \basis_0$, this can be written as
\begin{equation*}
\dbdh^\top \gradD\lrrb{\Dh; \Sxh, \Syh} = \dbdsT \gradD\lrrb{\Ds; \Sxh, \Syh} + \dbdh^\top \hat{\Hessian} \vec\lrrb{\Dh-\Ds} + \lrrb{\dbdh-\dbds}^\top \gradD\lrrb{\Ds; \Sxh, \Syh}.
\end{equation*}
Expanding the middle term on the right-hand side as $\dbdh^\top \hat{\Hessian} \vec\rb{\Dh-\Ds} = \sqb{\hat{\Hessian} \dbdh}_0 \rb{\Dh_0-\Ds_0} + \sqb{\hat{\Hessian} \dbdh}_u^\top \rb{\Dh_u-\Ds_u}$ and rearranging,
\begin{align} \label{eq:decomposition:1}
\Dt_0-\Ds_0
&= \lrrb{\Dh_0-\lrsqb{\hat{\Hessian} \dbdh}_0^{-1} \dbdh^\top \gradD\lrrb{\Dh; \Sxh, \Syh}}-\Ds_0 \nonumber\\
&=
\begin{multlined}[t]
\lrcb{1+\lrrb{\lrsqb{\hat{\Hessian} \dbdh}_0-1}}^{-1} \bigg\{-\dbdsT \gradD\lrrb{\Ds; \Sxh, \Syh} \\
- \lrsqb{\hat{\Hessian} \dbdh}_u^\top \lrrb{\Dh_u-\Ds_u} - \lrrb{\dbdh-\dbds}^\top \gradD\lrrb{\Ds; \Sxh, \Syh} \bigg\}.
\end{multlined}
\end{align}
Now, it can be checked (see \cref{eq:gradD} in \cref{app:loss}) that
\begin{multline} \label{eq:decomposition:2}
-\dbdsT \gradD\lrrb{\Ds; \Sxh, \Syh}
= \mdbdsT_1 \vec\lrrb{\Sxh-\Sxs} + \mdbdsT_2 \vec\lrrb{\Syh-\Sys} \\
- \dbdsT \frac 12 \lrcb{\lrrb{\Sxh-\Sxs} \otimes \lrrb{\Syh-\Sys} + \lrrb{\Syh-\Sys} \otimes \lrrb{\Sxh-\Sxs}} \vec(\Ds),
\end{multline}
where
\begin{equation} \label{eq:mdbds}
\mdbds_1 = -\frac 1 2 \lrrb{\Identity_{p'} \otimes \Sxsinv \Sys + \Sxsinv \Sys \otimes \Identity_{p'}} \dbds, \, \text{ and } \, 
\mdbds_2 = \frac 1 2 \lrrb{\Identity_{p'} \otimes \Sysinv \Sxs + \Sysinv \Sxs \otimes \Identity_{p'}} \dbds,
\end{equation}
with $p' = 2p$ for $\freq \notin \cb{0, \pi}$ and $p' = p$ for $\freq \in \cb{0, \pi}$. Combining \cref{eq:decomposition:1,eq:decomposition:2},
\begin{align} \label{eq:decomposition}
\Dt_0-\Ds_0
&=
\begin{multlined}[t]
\biggcb{1+\underbrace{\lrrb{\lrsqb{\hat{\Hessian} \dbdh}_0-1}}_{W}}^{-1} \bigg[\underbrace{\mdbdsT_1 \vec\lrrb{\Sxh-\Sxs} + \mdbdsT_2 \vec\lrrb{\Syh-\Sys}}_{U} \\
- \underbrace{\dbdsT \frac 12 \lrcb{\lrrb{\Sxh-\Sxs} \otimes \lrrb{\Syh-\Sys} + \lrrb{\Syh-\Sys} \otimes \lrrb{\Sxh-\Sxs}} \vec(\Ds)}_{V_1} \\
- \underbrace{\lrsqb{\hat{\Hessian} \dbdh}_u^\top \lrrb{\Dh_u-\Ds_u}}_{V_2} - \underbrace{\lrrb{\dbdh-\dbds}^\top \gradD\lrrb{\Ds; \Sxh, \Syh}}_{V_3} \bigg].
\end{multlined}
\end{align}
Term $U$ is the leading term, and it is a sum of initial covariance structure estimators after projection. From this decomposition, we can deduce that a de-biased D-trace estimator is asymptotically Normal and unbiased when both of the following conditions hold:
\begin{itemize}
\item The pair of initial covariance structure estimators $\Sxh$ and $\Syh$ are asymptotically Normal and unbiased after projecting in the direction of $\mdbds_1$ and $\mdbds_2$, respectively;
\item The error terms $V_1, V_2$, $V_3$, and $W$ are $o_{\PP}\rb{\effN^{1/2}}$, where $\effN = (T_1/B_1) + (T_2/B_2)$ is the effective sample size.
\end{itemize}
This result is formalized as \cref{lem:BE:Dtrace:abbr}.

\begin{lemma}[Berry-Esseen inequality for a de-biased D-trace estimator] \label{lem:BE:Dtrace:abbr}
Let $\mineigen_1$ and $\mineigen_2$ be the minimum eigenvalues of $\Sxs$ and $\Sys$, respectively. Let $\sparsity_{\Ds} = \norm{\Ds}_0$ and $\sparsity_{\dbds} = \norm{\dbds}_0$. Suppose
\begin{gather}
\sup_z \lrabs{\PP\lrcb{\asympSD^{-1} \sqrt{\effN} \lrcb{\mdbdsT_1 \vec\lrrb{\Sxh-\Sxs} + \mdbdsT_2 \vec\lrrb{\Syh-\Sys}} \leq z} - \Phi(z)} \leq \BEprob, \label{cond:BE} \\
\PP\lrcb{\infnorm{\Sxh-\Sxs} > \Sxhrate \text{ or } \infnorm{\Syh-\Sys} > \Syhrate} \leq \LDprob \label{cond:LD}
\end{gather}
for some $\asympSD^2 > 0$, $\effN = \effN(T_1, T_2) > 0$, $\Sxhrate = \Sxhrate(d, T_1) \in (0, 1)$, $\Syhrate = \Syhrate(d, T_2) \in (0, 1)$, $\BEprob = \BEprob(d, T_1, T_2) \in (0, 1)$, $\LDprob = \LDprob(d, T_1, T_2) \in (0, 1)$ such that
\begin{equation} \label{cond:SC:RSC:abbr}
\Sxhrate + \Syhrate \leq \frac{\mineigen_1 \mineigen_2}{\max\lrcb{32 \sparsity_{\Ds}, 8 \sparsity_{\dbds}} \lrrb{\max\lrcb{\infnorm{\Sxs}, \infnorm{\Sys}} + 1/2}}.
\end{equation}
If the regularization parameters satisfy
\begin{equation} \label{cond:regDH}
\regD \geq 2 \lrrb{\onenorm{\Ds} \Hessrate + \Sxhrate + \Syhrate}, \quad \onenorm{\dbds} \Hessrate \leq \regH \leq 1/2,
\end{equation}
where $\Hessrate = \infnorm{\Sys} \Sxhrate + \infnorm{\Sxs} \Syhrate + \Sxhrate \Syhrate$. Then,
\begin{equation*}
\sup_z \lrabs{\PP\lrcb{\asympSD^{-1} \sqrt{\effN} \lrrb{\Dt_0-\Ds_0} \leq z} - \Phi(z)} \lesssim \BEprob + \LDprob + \LDrate
\end{equation*}
with
\begin{equation*}
\LDrate \lesssim_{\Sxs, \Sys} \frac{\sqrt{\effN}}{\asympSD} \lrrb{\onenorm{\Ds} \onenorm{\dbds} + \sparsity_{\Ds} + \sparsity_{\dbds}} \lrrb{\Sxhrate^2 + \Syhrate^2} + \Sxhrate + \Syhrate.
\end{equation*}
\end{lemma}

The proof is in \cref{app:proof:BE:Dtrace:abbr}. The analogous result for alternative estimating equations in \cref{eq:pop:1}--\cref{eq:pop:0} is given as \cref{lem:BE:Dtrace:abbr:general} in \cref{app:loss:general}.

\cref{lem:BE:Dtrace:abbr} gives a Normal approximation error bound for a de-biased D-trace estimation procedure that starts from a pair of initial covariance structure estimators satisfying a Berry-Esseen condition (Eq.~\ref{cond:BE}) and a large deviation condition (Eq.~\ref{cond:LD}). As such, it can be interpreted as a result that clarifies how the errors in the initial estimates propagate through the information geometry of de-biased D-trace estimation to the final difference estimator. The benefit of this approach is that it successfully modularizes the initial estimation stage and what comes after, thereby simplifying the analysis.

Note that this Normal approximation error bound increases with $\sparsity_{\dbds}$, where we recall from \cref{subsec:SDDinference} that $\dbds = \Hessian^{*-1}_{\EEcoef} \rb{\basis_{i_0} \otimes \basis_{j_0}}$. When $\EEcoef = 1$, $\Hessian^{*-1}_{1} = \Sxsinv \otimes \Sysinv$, and hence, $\sparsity_{\dbds} = \norm{\Sxsinv \basis_{i_0}}_0 \norm{\Sysinv \basis_{j_0}}_0$. An analogous conclusion holds for $\EEcoef = 0$. Thus, we see that even for an edge in a differential network, valid inference is easier if its nodes are connected to only a few edges in the underlying networks, which is intuitive. However, that this is a restriction on the local structures of the underlying spectral networks and not their global structures. We make some further remarks in \cref{sec:discussion}.

\subsection{SDD inference with Welch's estimator}

So far, we have not touched on the temporal dependence aspect of our problem. This is because the effects of temporal dependence are the most apparent in the large sample properties of initial estimators. In the relatively straightforward setting in which both $X_{1, \cdot}$ and $X_{2, \cdot}$ are each i.i.d., it is not difficult to imagine that estimators satisfying the requirements of \cref{lem:BE:Dtrace:abbr} can be obtained even in high dimensions, as long as  sufficient data are available. However, when $X_{1, \cdot}$ and $X_{2, \cdot}$ have temporal dependence, the situation is more precarious. For one thing, we can expect the effective sample size to be somewhat smaller than $T_1$ or $T_2$. For another, our parameter of interest is a functional of $X_{1, \cdot}$ and $X_{2, \cdot}$ which have infinite lengths, but we only observe single finite trajectories $\rb{X_{1, t}}_{t=1}^{T_1}$ and $\rb{X_{2, t}}_{t=1}^{T_2}$. Therefore, to have any hope at inference, we will need to assume that the long lag dependences are negligible. In fact, unlike the i.i.d.~setting, we will see that the truncation bias that results from the restricted observation window is often the most dangerous.

What are some candidate estimators that can be shown to satisfy the Berry-Esseen condition (Eq.~\ref{cond:BE}) and a large deviation condition (Eq.~\ref{cond:LD}) of \cref{lem:BE:Dtrace:abbr}? In this subsection, we shall show that Welch's estimator can satisfy these conditions when both processes $X_{1, \cdot}$ and $X_{2, \cdot}$ exhibit weak functional dependence, culminating in a Normal approximation error bound for the SDD procedure with Welch's estimator in \cref{thm:BE:SDD}. As explained in the opening paragraph, estimation of the spectral density of a dependent process is a delicate matter, especially in high dimensions. Known consistent estimators typically require specification of additional tuning parameters (e.g., the window size $B$), which are usually chosen to balance the bias with the variance. 
However, such a choice may not be suitable for inference. Therefore, for our purposes, we employ a different strategy: We will first derive a Normal approximation error bound, and this will dictate our choice of the tuning parameters, as well as requirements about how fast $T_1$ and $T_2$ should grow relative to other problem parameters.

Previously, \citet{zhang2025spectral} studied the large deviation and Normal approximation properties of Welch's estimator in the context of a single sample problem. To set the proper context for our results, we first discuss these earlier results, which we restate as \cref{thm:B_rates,thm:asympnormal}. These theorems apply to a certain class of weakly dependent processes defined via the \emph{functional dependence} framework of \citet{wu2005nonlinear}. For the sake of completeness, we briefly depart from our two sample setting to provide a brief background on these earlier results for a single process $X_\cdot = \rb{X_t}_{t \in \integers}$.

\subsubsection{More on Welch's estimator of the spectral density} \label{subsubsec:Welch}

Consider a process $X_\cdot = \rb{X_t}_{t \in \integers}$ and recall the definitions of the spectral density and Welch's estimator in \cref{eq:SD,eqn:welch_estimator}. Expanding the product, we can rewrite Welch's estimator as a double sum
\begin{equation} \label{eq:Welch:blocksum}
\SDh = \SDh(\freq) = \frac{1}{2 \pi (T/B)} \sum_{n = 1}^{T/B} \lrrb{\frac{1}{B} \sum_{s, t = (n-1) B+1}^{n B} e^{-\imath \freq (s-t)} X_s X_t^\top}.
\end{equation}
In the above double sum, we first partition the observed trajectory of length $T$ into $T/B$ non-overlapping blocks of size $B$, then sum all the cross products between every pair of observations within each block, and finally aggregate all $T/B$ block-wise averages. Although we expect that the dependence within each block cannot be ignored, the dependence among the block-wise sums may be substantially reduced, especially if the degree of dependence fades with the size of the lag. Thus, we can think of Welch's estimator as a sample average with the effective sample size equal to $T/B$. The smaller $B$ is, the larger the effective sample size $T/B$, and hence, the lower the variance.

What about the bias? Taking the expectation,
\begin{equation*}
\EE\lrsqb{\SDh} = \frac{1}{2 \pi B} \sum_{s, t = (n-1) B+1}^{n B} e^{-\imath \freq (s-t)} \EE\lrsqb{X_s X_t^\top} = \frac{1}{2 \pi} \sum_{h = -B}^{B} \lrrb{1-\frac{|h|}{B}} e^{-\imath \freq h} \Autocov(h),
\end{equation*}
where in the second equality, we set $h = s-t$, re-indexed the sum according to $h$, and collect the terms. Comparing $\EE\sqb{\SDh}$ with $\SDs$,
\begin{equation} \label{eq:Welch:bias}
\EE\lrsqb{\SDh}-\SDs = -\frac{1}{2 \pi B} \sum_{h = -B}^{B} \abs{h} e^{-\imath \freq h} \Autocov(h) - \frac{1}{2 \pi} \sum_{|h| > B} e^{-\imath \freq h} \Autocov(h).
\end{equation}
Thus, we see that the finiteness of $B$ introduces a truncation bias. This decreases with increasing $B$, provided that long lag dependences decay sufficiently fast.

In summary, a smaller window size $B$ leads to a larger effective sample size $T/B$ and a lower variance, but at the cost of greater truncation bias. Conversely, a larger window size $B$ leads to a lower truncation bias but also increases the variance. Typically, $B$ is chosen to balance the bias with the variance. If Normal approximation is desired, the bias term must decay faster than $\sqrt{B/T}$, which restricts how slow $B$ can grow relative to $T$. In high dimensions, we need to also consider the parameter dimension $p$. 

\subsubsection{Functional dependence framework}

In the \emph{functional dependence} framework of \citet{wu2005nonlinear}, a process $\rb{X_t}_{t \in \integers}$ is viewed as filtered outputs of i.i.d.~innovations, which naturally leads to a measure of dependence based on couplings of $\rb{X_t}_{t \in \integers}$. Let $X_t \in \reals^p$ be $p$-dimensional. Suppose the $j$-th component process $\rb{X_{t, j}}_{t \in \integers}$ of $\rb{X_t}_{t \in \integers}$ is given by
\begin{equation} \label{eqn:fn_dep_rep}
X_{t, j} = \filter_{j}\rb{\inno_{t}, \inno_{t-1}, \ldots},
\end{equation}
where $\filter_{j}$ is a measurable function and $\inno_{t}$ are i.i.d.~random variables. In this characterization, the $p$ component processes $\rb{X_{j, t}}_{t \in \integers}$, $j \in \cb{1, \ldots, p}$, are linked via a common innovation process $\rb{\inno_t}_{t \in \integers}$.

The degree of dependence in $\rb{X_{t, j}}_{t \in \integers}$ can be quantified by measuring the distance between coupled copies of $\rb{X_{t, j}}_{t \in \integers}$, i.e.,
\begin{equation} \label{eqn:dependence}
\dependence_{q, j}(h) = \Lqnorm{\filter_{j}\rb{\inno_{h}, \ldots, \inno_{1}, \inno_{0}, \inno_{-1}, \ldots} - \filter_{j}\rb{\inno_{h}, \ldots, \inno_{1}, \inno_{0}', \inno_{-1}, \ldots}},
\end{equation}
where $\inno_{0}'$ is an i.i.d.~copy of $\inno_{0}$. By convention, $\dependence_{q, j}(h) = 0$ for $h < 0$. $\dependence_{q, j}(h)$ can be interpreted as the effect of the innovation at time $0$ on the current response $X_{h, j}$. We refer to $\dependence_{q, j}(h)$ as the \emph{functional dependence measure} for the $j$-th component process $\rb{X_{t, j}}_{t \in \integers}$ at lag $h$ and moment order $q$. Although each $\dependence_{q, j}$ is ostensibly only about a single univariate component process $\rb{X_{t, j}}_{t \in \integers}$, because each $X_{t, j}$ depend on other $X_{s, i}$ for $i \neq j$ and $s \neq t$ through a common $\rb{\inno_t}_{t \in \integers}$, these functional dependence measures capture all the intricate dependencies among component processes in a concise and elegant way. \cref{eqn:fn_dep_rep} encompasses a plethora of processes, including many models that are widely used in the time series literature, such as ARMA, ARMA-ARCH, and ARMA-GARCH, as well as examples from dynamic system theory \citep[and references therein]{wu2005nonlinear, shao2007asymptotic, liu2010asymptotics, jirak2016berry}.

Establishing asymptotic properties of statistical procedures typically requires moment conditions. For our purpose, we require a generalized notion of moments that builds on \cref{eqn:dependence} and take the dependence structure into account. Following \citet{zhang2025spectral}, we define the \emph{dependence adjusted moment} for the $j$-th component process $\rb{X_{t, j}}_{t \in \integers}$ of moment order $q$ as
\begin{equation} \label{eq:dam}
\Lqnorm{X_{\cdot, j}} = \sup_{m \geq 0} \damrho^{-m} \sum_{h = m}^{\infty} \dependence_{q, j}(h),
\end{equation}
where $\damrho \in (0, 1)$ is some constant. We also define the \emph{dependence adjusted $\daOpsi_\daOalpha$-norm} as
\begin{equation} \label{eq:daO}
\norm{X_{\cdot, j}}_{\daOpsi_\daOalpha} = \sup_{q \geq 2} q^{-\daOalpha} \Lqnorm{X_{\cdot, j}},
\end{equation}
where $\daOalpha \geq 0$ is some constant.

To contextualize the concepts of dependence adjusted moments and norms, we briefly discuss a few of their key properties. 
First, for i.i.d.~processes, the dependence adjusted moments, $\Lqnorm{X_{\cdot, j}}$, are equivalent to the classical moments  \citep{zhang2025spectral}. This equivalence  immediately implies the equivalence of the dependence adjusted $\daOpsi_\daOalpha$-norms and the classical $\daOpsi_\daOalpha$-norms. 
Second, \cref{eq:dam} effectively imposes a decay condition on long lag functional dependence measures, since it implies that the tail sum $\sum_{t = h}^{\infty} \dependence_{q, j}(h) \leq \Lqnorm{X_{\cdot, j}} \rho^h$ for all $h \geq 0$.
Finally, similar to the usual $\daOpsi_\daOalpha$-norm, the parameter $\daOalpha$ controls how fast the dependence adjusted moments can grow with the moment order. Specifically, $\daOalpha = 1/2$ corresponds to the sub-Gaussian tails, while $\daOalpha = 1$ corresponds to the sub-exponential tails. Larger values of $\daOalpha$ indicate heavier tails.

Before presenting our results on asymptotic distribution of the de-biased estimator, we next present two results from \citet{zhang2025spectral}, \cref{thm:B_rates,thm:asympnormal}, about the properties of the Welch's estimator. 
These results require the process $\rb{X_t}_{t \in \integers}$ to satisfy the following condition: there exists a constant $\daOalpha \geq 0$ such that
\begin{equation} \label{cond:bdddaO}
\norm{X_{\cdot}}_{\daOpsi_\daOalpha} = \max_j \norm{X_{\cdot, j}}_{\daOpsi_\daOalpha} < \infty.
\end{equation}

\subsubsection{Consistency of Welch's estimator under weak functional dependence} \label{subsubsec:Zhang:consistency}

\citet{zhang2025spectral} provide the following large deviation bound for Welch's estimators in high dimensions.

\begin{theorem}[Theorem 1 \& Proposition 2 of \citet{zhang2025spectral}] \label{thm:B_rates}
Assume \cref{cond:bdddaO}. For any frequency $\freq \in [0, \pi]$,
\begin{equation*}
\PP\lrcb{\infnorm{\SDh(\freq) - \SDs(\freq)} \gtrsim_{\damrho, \daOalpha} \norm{X_{\cdot}}_{\daOpsi_\daOalpha}^2 \cb{\log \rb{p^2}}^{1 + 2 \daOalpha} \sqrt{\frac{B}{T}} + \frac{\dammax_2^2}{B}} \lesssim \exp\rb{-\log \rb{p^2}},
\end{equation*}
where $\dammax_2 = \max_j \norm{X_{\cdot, j}}_{L_2}$ and $\damrho$ is a constant satisfying the condition of \cref{eq:dam}.
\end{theorem}

The right-hand-side of the bound involves $\log \rb{p^2}$ rather than $\log (B p^2)$, because we are interested in the inference at a specific frequency $\freq \in [0, \pi]$. Of the two error terms, the $O\rb{\norm{X_{\cdot}}_{\daOpsi_\daOalpha}^2 \cb{\log \rb{p^2}}^{1 + 2 \daOalpha} \sqrt{B / T}}$ term reflects the contribution of the variance and the $O\rb{\dammax_2^2 / B}$ term reflects the contribution of the bias.

Although \cref{thm:B_rates} is about the complex-valued $\SDh-\SDs$, the same bound applies to the realified counterpart $\Sh-\Ss$ via the following observation. For any complex matrix $\SD = \rSD + \imath \iSD$, let $\Sigma$ be the realification of $\SD$: $\Sigma = \lrsqb{\begin{smallmatrix} \rSD & -\iSD \\ \iSD & \phantom{-}\rSD \end{smallmatrix}}$. Clearly, $\norm{\Sigma}_{\infty} = \max\cb{\norm{\rSD}_{\infty}, \norm{\iSD}_{\infty}} \leq \norm{\SD}_{\infty}$. Thus, the conclusion of \cref{thm:B_rates} holds with $d = \dim\rb{\Ss}$, $\Sh$ and $\Ss$ replacing $p^2$, $\SDh$ and $\SDs$, respectively.

\subsubsection{Asymptotic Normality of Welch's estimator under weak functional dependence} \label{subsubsec:Zhang:asympnormal}

From \cref{eq:Welch:blocksum}, the Welch's estimator is constructed as a sum of $T/B$ non-overlapping block sums. Thus, it is reasonable to expect that if the degree of dependence in the process $X_\cdot$ is mild, the estimator can be well-approximated by an i.i.d.~sum for sufficiently large $T/B$. If this can be achieved with a large window size $B$, the truncation bias in \cref{eq:Welch:bias} can also be ignored. This is the intuition behind \citet{zhang2025spectral}'s bivariate asymptotic Normality result.

\begin{theorem}[Theorem 3 \& Proposition 2 of \citet{zhang2025spectral}] \label{thm:asympnormal}
Assume \cref{cond:bdddaO}. Suppose $T$, $B \to \infty$ such that
\begin{equation} \label{cond:lowdim}
\frac{B}{T} \to 0, \quad \frac{\dammax_4^8 \log B}{B} \to 0, \quad \frac{\dammax_2^4 T}{B^3} \to 0,
\end{equation}
where $\dammax_2 = \max_j \norm{X_{\cdot, j}}_{L_2}$ and $\dammax_4 = \max_j \norm{X_{\cdot, j}}_{L_4}$. Let $i, j, k, l \in [p]$. Let $f^*_{ij}$ be the $(i, j)$-th component of the spectral density $\SDs$. Then,
\begin{enumerate}
\item For $\freq \notin \cb{0, \pi}$,
\begin{equation*}
\sqrt{\frac{T}{B}} \begin{pmatrix} \SDhcomp{ij} - f^*_{ij} \\ \SDhcomp{kl} - f^*_{kl} \end{pmatrix} \xrightarrow{d} \ComplexNormal\lrrb{0, \begin{bmatrix} f^*_{ii} f^*_{jj} & f^*_{ik} f^*_{lj} \\ f^*_{ki} f^*_{jl} & f^*_{kk} f^*_{ll} \end{bmatrix}, \begin{bmatrix} f^{*2}_{ij} & f^*_{il} f^*_{kj} \\ f^*_{il} f^*_{kj} & f^{*2}_{kl} \end{bmatrix}},
\end{equation*}
where, $\ComplexNormal(\mu, \VarSh_1, \VarSh_2)$ denotes a bivariate complex Normal distribution with mean $\mu$, covariance matrix $\VarSh_1$ and pseudo-covariance matrix $\VarSh_2$.
\item For $\freq \in \cb{0, \pi}$,
\begin{equation*}
\sqrt{\frac{T}{B}} \begin{pmatrix} \SDhcomp{ij} - f^*_{ij} \\ \SDhcomp{kl} - f^*_{kl} \end{pmatrix} \xrightarrow{d} \Normal\lrrb{0, \begin{bmatrix} f^*_{ii} f^*_{jj} + f^{*2}_{ij} & f^*_{ik} f^*_{lj} + f^*_{il} f^*_{kj} \\ f^*_{ki} f^*_{jl} + f^*_{il} f^*_{kj} & f^*_{kk} f^*_{ll} + f^{*2}_{kl} \end{bmatrix}}.
\end{equation*}
\end{enumerate}
\end{theorem}

\cref{cond:lowdim} requires that $B$ grows more slowly than $T$, but not too slowly for the truncation bias to vanish. 
Although \cref{thm:asympnormal} is about the complex random vector $\rb{\SDhcomp{ij}, \SDhcomp{kl}}^\top$, we can deduce the asymptotic covariance of any pair of components of realified Welch's estimator $\rb{\Sh_{ij}, \Sh_{kl}}^\top$ by using the fact that if $Z_1$ and $Z_2$ are complex random variables with covariance $\VarSh_1$ and pseudo-covariance $\VarSh_2$, then
\begin{align*}
2 \Cov\lrrb{\Re Z_1, \Re Z_2} &= \Re\rb{\VarSh_1 + \VarSh_2}, & 2 \Cov\lrrb{\Re Z_1, \Im Z_2} &= \Im\rb{\VarSh_2 - \VarSh_1}, \\
2 \Cov\lrrb{\Im Z_1, \Re Z_2} &= \Im\rb{\VarSh_2 + \VarSh_1}, & 2 \Cov\lrrb{\Im Z_1, \Im Z_2} &= \Re\rb{\VarSh_1 - \VarSh_2}.
\end{align*}
Conversely, $\VarSh_1$ and $\VarSh_2$ can be computed from the covariances of real and imaginary parts.

\cref{thm:asympnormal} tells us that each pair of components of Welch's estimator is asymptotically unbiased and converges in distribution to a bivariate complex Normal random variable with covariance and pseudo-covariance that are functions of the spectral density of $X_\cdot$. It also makes clear that the correct asymptotic scaling is $\sqrt{T / B}$. However, it does not tell us how fast the convergence to the Normal distribution is or what joint covariance and pseudo-covariance structures are being implied by the pairwise result. Since our primary focus is in high-dimensional regimes, it is imperative that we verify the conditions under which the Normal approximation error $\BEprob$ vanishes; we turn to this in \cref{subsubsec:SDDinference}. Before we do that, we first deduce the full limiting joint covariance structure for the realified Welch's estimator $\Sh$. This will help us determine the asymptotic variance in \cref{subsubsec:SDDinference} when we establish a Berry-Esseen inequality for a sum of projected Welch's estimators. 
Below, \cref{prop:asympcov,prop:asympcov:real} handle the $\freq \notin \cb{0, \pi}$ and $\freq \in \cb{0, \pi}$ cases, respectively.

The vectorization of the realification $\Sigma$ of a complex matrix $\SD = \rSD + \imath \iSD$ has the following relationship to the vectorization of  $\lrsqb{\begin{smallmatrix}\vec \rSD \\ \vec \iSD \end{smallmatrix}}$: $\vec\rb{\Sigma} = \Permutation^\top \lrsqb{\begin{smallmatrix} \vec \rSD \\ \vec \iSD \end{smallmatrix}}$, where $\Permutation$ is the matrix defined in \cref{subsubsec:inference:details}. Also, recall the definition of the commutation matrix: $\Commutation = \sum_{i = 1}^{p} \sum_{j = 1}^{p} \Basis_{ij} \otimes \Basis_{ji}$, where $\Basis_{ij}$ is the $(i, j)$-th standard basis matrix in $\reals^{p \times p}$.

\begin{proposition} \label{prop:asympcov}
Suppose $\matZ \in \complex^{p \times p}$ is a complex random matrix such that for any $i, j, k, l \in [p]$,
\begin{gather*}
\Cov\Bigrb{Z_{i, j}, Z_{k, l}} = \EE\Bigsqb{\lrrb{Z_{i, j} - \EE{Z_{i, j}}} \overline{\lrrb{Z_{k, l} - \EE{Z_{k, l}}}}} = f^*_{i, k} f^*_{l, j}, \\
\Cov\Bigrb{Z_{i, j}, \overline{Z_{k, l}}} = \EE\Bigsqb{\lrrb{Z_{i, j} - \EE{Z_{i, j}}} \lrrb{Z_{k, l} - \EE{Z_{k, l}}}} = f^*_{i, l} f^*_{k, j},
\end{gather*}
i.e., $\matZ$ has the same covariance and pseudo-covariance structure as the limiting covariance and pseudo-covariance in the first part of \cref{thm:asympnormal}. Then, for any $\vecw \in \reals^{4p^2}$, $\asympSD^2 = \Var\lrrb{\vecw^\top \Permutation^\top \lrsqb{\begin{smallmatrix} \vec \Re \matZ \\ \vec \Im \matZ \end{smallmatrix}}} = \lrrb{\Permutation \vecw}^\top \VarSh \Permutation \vecw$, where
\begin{equation} \label{eq:VarSh}
\begin{aligned}
\VarSh
&= \Var\lrrb{\begin{bmatrix} \vec \Re \matZ \\ \vec \Im \matZ \end{bmatrix}} \\
&= \frac 1 2 \begin{bmatrix}
\lrrb{\Identity_{p^2} + \Commutation} \lrrb{\Re \SDs \otimes \Re \SDs+\Im \SDs \otimes \Im \SDs} &
\lrrb{\Identity_{p^2} + \Commutation} \lrrb{\Im \SDs \otimes \Re \SDs - \Re \SDs \otimes \Im \SDs} \\
\lrrb{\Identity_{p^2} - \Commutation} \lrrb{\Re \SDs \otimes \Im \SDs - \Im \SDs \otimes \Re \SDs} &
\lrrb{\Identity_{p^2} - \Commutation} \lrrb{\Re \SDs \otimes \Re \SDs + \Im \SDs \otimes \Im \SDs}
\end{bmatrix}.
\end{aligned}
\end{equation}
\end{proposition}

\begin{proposition} \label{prop:asympcov:real}
Suppose $\matZ \in \reals^{p \times p}$ is a real random matrix such that for any $i, j, k, l \in [p]$,
\begin{equation*}
\Cov\bigrb{Z_{i, j}, Z_{k, l}} = \EE\Bigsqb{\lrrb{Z_{i, j} - \EE{Z_{i, j}}} \lrrb{Z_{k, l} - \EE{Z_{k, l}}}} = f^*_{i, k} f^*_{l, j} + f^*_{i, l} f^*_{k, j},
\end{equation*}
i.e., $\matZ$ has the same covariance structure as the limiting covariance in the second part of \cref{thm:asympnormal}. Then, for any $\vecw \in \reals^{p^2}$, $\asympSD^2 = \Var\rb{\vecw^\top \vec \matZ} = \vecw^\top \VarSh \vecw$, where
\begin{equation} \label{eq:VarSh:real}
\VarSh = \Var\rb{\vec \matZ} = \rb{\Identity_{p^2} + \Commutation} \rb{\SDs \otimes \SDs}.
\end{equation}
\end{proposition}

Combining \cref{thm:asympnormal} with \cref{prop:asympcov,prop:asympcov:real} yields the following joint asymptotic Normality result in the classical regime of a fixed $p$ and $T$, $B \to \infty$. Since the result may be of independent interest, it is stated below.

\begin{corollary}
In the setting of \cref{thm:asympnormal}, assume $p$ is fixed.
\begin{enumerate}
\item For $\freq \notin \cb{0, \pi}$,
\begin{equation*}
\sqrt{\frac{T}{B}} \lrrb{\vec \Sh - \vec \Ss} \xrightarrow{d} \Normal\lrrb{0, \Permutation^\top \VarSh \Permutation},
\end{equation*}
where $\VarSh$ is as defined in \cref{eq:VarSh}.
\item For $\freq \in \cb{0, \pi}$,
\begin{equation*}
\sqrt{\frac{T}{B}} \lrrb{\vec \Sh - \vec \Ss} \xrightarrow{d} \Normal\lrrb{0, \VarSh},
\end{equation*}
where $\VarSh$ is as defined in \cref{eq:VarSh:real}.
\end{enumerate}
\end{corollary}

Having discussed the single sample setting, we next return to the two sample setting of our problem.

\subsubsection{SDD inference with Welch's estimator under weak functional dependence} \label{subsubsec:SDDinference}

We assume that the finite $\daOpsi_\daOalpha$-norm condition is satisfied for both $X_{1, \cdot}$ and $X_{2, \cdot}$ for some $\daOalpha_1, \daOalpha_2 \geq 0$:

\begin{assumption} \label{assm}
There exist a pair of constants $\daOalpha_1, \daOalpha_2 \geq 0$ such that
\begin{equation*}
\norm{X_{1, \cdot}}_{\daOpsi_{\daOalpha_1}} = \max_j \norm{X_{1, \cdot, j}}_{\daOpsi_{\daOalpha_1}} < \infty, \quad \norm{X_{2, \cdot}}_{\daOpsi_{\daOalpha_2}} = \max_j \norm{X_{2, \cdot, j}}_{\daOpsi_{\daOalpha_2}} < \infty.
\end{equation*}
\end{assumption}

As \citet{zhang2025spectral} also notes, \cref{assm} implies that $X_{l, t, j}$ has finite moments of any order. For specific time series models, this translates to restrictions on the model parameters. For example, \citet{zhang2025spectral} show that in the case of high-dimensional ARMA or ARCH models, \cref{assm} amounts to imposing a bound on the moments of the innovations $\inno_t$.

Under \cref{assm}, by \cref{thm:B_rates}, the large deviation condition \cref{cond:LD} of \cref{lem:BE:Dtrace:abbr} holds with
\begin{gather}
\Sxhrate \asymp_{\damrho_1, \daOalpha_1} \norm{X_{1, \cdot}}_{\daOpsi_{\daOalpha_1}}^2 \sqrt{\frac{B_1}{T_1}} \rb{\log d}^{1 + 2 \daOalpha_1} + \frac{\dammax_{1, 2}^2}{B_1}, \label{eq:Sxhrate} \\
\Syhrate \asymp_{\damrho_2, \daOalpha_2} \norm{X_{2, \cdot}}_{\daOpsi_{\daOalpha_2}}^2 \sqrt{\frac{B_2}{T_2}} \rb{\log d}^{1 + 2 \daOalpha_2} + \frac{\dammax_{2, 2}^2}{B_2}, \label{eq:Syhrate}
\end{gather}
and $\LDprob \asymp d^{-1}$. The Berry-Esseen condition for the sum of projected Welch's estimators is given in the next theorem.

\begin{theorem} \label{thm:BE:projWelch}
Suppose \cref{assm} holds. Let $\asympSD^2_1$ be the variance obtained by applying \cref{prop:asympcov} or \ref{prop:asympcov:real} with $\SDs = \SDs_1$ and $\vecw = \mdbds_1$ from \cref{eq:mdbds}. Let $\asympSD^2_2$ be the variance obtained by applying \cref{prop:asympcov} or \ref{prop:asympcov:real} with $\SDs = \SDs_2$ and $\vecw = \mdbds_2$ from \cref{eq:mdbds}. Suppose the window sizes $B_1$ and $B_2$ are large enough to allow
\begin{equation*}
\frac{c_{\damrho_l} \norm{\mdbds_l}_1^2 \dammax_{l, 4}^4}{\asympSD^2_l} \sqrt{\frac{\log B_l}{B_l}} \leq \frac 3 4, \quad l = 1,2,
\end{equation*}
where for $l=1,2$, $c_{\damrho_l} > 0$ is a constant that depend on $\damrho_l$ and is specified in the proof.

Then, the Berry-Esseen condition for the sum of projected Welch's estimators in \cref{cond:BE} holds for the following specifications of $\effN$, $\asympSD^2$, and $\BEprob$:
\begin{itemize}
\item The total effective sample size is $\effN = \effN_1 + \effN_2$, where $\effN_1 = T_1/B_1$ and $\effN_2 = T_2/B_2$;
\item The asymptotic variance is $\asympSD^2 = \rb{\effN/\effN_1} \asympSD^2_1 + \rb{\effN/\effN_2} \asympSD^2_2$;
\item The Normal approximation error bound is
\begin{multline*}
\BEprob \lesssim_{\mdbds_1, \mdbds_2, \asympSD, \damrho_1, \damrho_2, \daOalpha_1, \daOalpha_2} \frac{1}{\sqrt{\effN}} + d^{-1}
+ \norm{X_{1, \cdot}}_{\daOpsi_{\daOalpha_1}}^4 \rb{\log d}^{1 + 2 \daOalpha_1} \sqrt{\frac{\log B_1}{B_1}} + \dammax_{1, 2}^2 \sqrt{\frac{T_1}{B_1^3}} \\
+ \norm{X_{2, \cdot}}_{\daOpsi_{\daOalpha_2}}^4 \rb{\log d}^{1 + 2 \daOalpha_2} \sqrt{\frac{\log B_2}{B_2}} + \dammax_{2, 2}^2 \sqrt{\frac{T_2}{B_2^3}}.
\end{multline*}
\end{itemize}
\end{theorem}

\cref{thm:BE:projWelch} is a crucial stepping stone to our main result, \cref{thm:BE:SDD}. It substantially refines \cref{thm:asympnormal}: the proof follows the same intuition, but the analysis carefully tracks the contributions of various error terms to arrive at a Normal approximation error bound. The proof is therefore highly technical, and is relegated to \cref{app:proof:BE:projWelch}.

It remains to tie the implications of \cref{eq:Sxhrate,eq:Syhrate,thm:BE:projWelch} together via \cref{lem:BE:Dtrace:abbr}. Although this may appear to give a complicated error bound with many terms, it can be reduced to control of three terms that are reminiscent of the conditions of \cref{thm:asympnormal} in \cref{cond:lowdim}. In particular, $T_1 / B_1^3$ and $T_2 / B_2^3$ appear again, dictating the scaling of the window sizes $B_1$ and $B_2$.

\begin{theorem} \label{thm:BE:SDD}
Let $\sparsity_{\text{tot}} = \onenorm{\Ds} \onenorm{\dbds} + \sparsity_{\Ds} + \sparsity_{\dbds}$, and suppose $B_1, B_2 \geq \sparsity_{\text{tot}}$. Then, in the combined setting of \cref{lem:BE:Dtrace:abbr} and \cref{thm:BE:projWelch},
\begin{multline} \label{eq:BE:SDD}
\sup_z \lrabs{\PP\lrcb{\asympSD^{-1} \sqrt{\effN} \lrrb{\Dt_0-\Ds_0} \leq z} - \Phi(z)} \lesssim_{\Sxs, \Sys, \mdbds_1, \mdbds_2, \asympSD, \damrho_1, \damrho_2, \daOalpha_1, \daOalpha_2} d^{-1} \\
+ \max\lrcb{\norm{X_{1, \cdot}}_{\daOpsi_{\daOalpha_1}}^2, \norm{X_{1, \cdot}}_{\daOpsi_{\daOalpha_1}}^4} \sparsity_{\text{tot}} \rb{\log d}^{2 + 4 \daOalpha_1} \sqrt{\frac{B_1}{T_1}} + \norm{X_{1, \cdot}}_{\daOpsi_{\daOalpha_1}}^4 \rb{\log d}^{1 + 2 \daOalpha_1} \sqrt{\frac{\log B_1}{B_1}} + \max\lrcb{\dammax_{1, 2}^2, \dammax_{1, 2}^4} \sqrt{\frac{T_1}{B_1^3}} \\
+ \max\lrcb{\norm{X_{2, \cdot}}_{\daOpsi_{\daOalpha_2}}^2, \norm{X_{2, \cdot}}_{\daOpsi_{\daOalpha_2}}^4} \sparsity_{\text{tot}} \rb{\log d}^{2 + 4 \daOalpha_2} \sqrt{\frac{B_2}{T_2}} + \norm{X_{2, \cdot}}_{\daOpsi_{\daOalpha_2}}^4 \rb{\log d}^{1 + 2 \daOalpha_2} \sqrt{\frac{\log B_2}{B_2}} + \max\lrcb{\dammax_{2, 2}^2, \dammax_{2, 2}^4} \sqrt{\frac{T_2}{B_2^3}}.
\end{multline}
\end{theorem}

\cref{thm:BE:SDD} requires that $B_1$ and $B_2$ increase faster than $T_1^{1/3}$ and $T_2^{1/3}$, respectively. 
Indeed, suppose $\sparsity_{\text{tot}}$, $\norm{X_{1, \cdot}}_{\daOpsi_{\daOalpha_1}}$, and $\norm{X_{2, \cdot}}_{\daOpsi_{\daOalpha_2}}$ are all of a constant order. Because $\dammax_{1, 2} \leq 2^{\daOalpha_1} \norm{X_{1, \cdot}}_{\daOpsi_{\daOalpha_1}}$ and $\dammax_{2, 2} \leq 2^{\daOalpha_2} \norm{X_{2, \cdot}}_{\daOpsi_{\daOalpha_2}}$ by \cref{eq:daO}, $\dammax_{1, 2}$ and $\dammax_{2, 2}$ are also of a constant order. Ignoring the $\log B_1$ and $\log B_2$ factors, we see that in this case, the Normal approximation error is minimized for
\begin{equation} \label{eq:BE:optB}
B_1 \asymp \frac{\sqrt{T_1}}{\rb{\log d}^{1 + 2 \daOalpha_1}}, \quad B_2 \asymp \frac{\sqrt{T_2}}{\rb{\log d}^{1 + 2 \daOalpha_2}}.
\end{equation}

\begin{corollary} \label{cor:BE:SDD:abbr}
In the setting of \cref{thm:BE:SDD}, suppose the window sizes are chosen as in \cref{eq:BE:optB}. Then, provided that
\begin{gather*}
\max\lrcb{\norm{X_{1, \cdot}}_{\daOpsi_{\daOalpha_1}}^2, \norm{X_{1, \cdot}}_{\daOpsi_{\daOalpha_1}}^4} \sparsity_{\text{tot}} \rb{\log d}^{\frac{3}{2} + 3 \daOalpha_1} = o\lrrb{T_1^{\frac 1 4}}, \\
\max\lrcb{\norm{X_{2, \cdot}}_{\daOpsi_{\daOalpha_2}}^2, \norm{X_{2, \cdot}}_{\daOpsi_{\daOalpha_2}}^4} \sparsity_{\text{tot}} \rb{\log d}^{\frac{3}{2} + 3 \daOalpha_2} = o\lrrb{T_2^{\frac 1 4}},
\end{gather*}
we have
\begin{equation*}
\sup_z \lrabs{\PP\lrcb{\asympSD^{-1} \sqrt{\effN} \lrrb{\Dt_0-\Ds_0} \leq z} - \Phi(z)} \to 0 \quad \text{as} \quad d, \, T_1, \, T_2 \to \infty.
\end{equation*}
\end{corollary}

The proofs of \cref{thm:BE:SDD} and \cref{cor:BE:SDD:abbr} can be found in \cref{app:proof:BE:SDD}.

\paragraph{Estimating the variance} \label{para:varest}

\cref{thm:BE:SDD,cor:BE:SDD:abbr} tell us that it is possible to construct confidence intervals or conduct hypothesis tests based on Normal approximation theory using our estimator $\Dt_0$. To implement such a procedure, however, we must also estimate the variance of $\Dt_0$. We have shown that asymptotically, this tends to $\asympSD^2 = \rb{\effN / \effN_1} \asympSD^2_1 + \rb{\effN / \effN_2} \asympSD^2_2$, where $\asympSD^2_1 = \rb{\Permutation \mdbds_1}^\top \VarSxh \Permutation \mdbds_1$ and $\asympSD^2_2 = \rb{\Permutation \mdbds_2}^\top \VarSyh \Permutation \mdbds_2$ in the case of $\freq \notin \cb{0, \pi}$, and $\asympSD^2_1 = \mdbdsT_1 \VarSxh \mdbds_1$ and $\asympSD^2_2 = \mdbdsT_2 \VarSyh \mdbds_2$ in the case of $\freq \in \cb{0, \pi}$. Here, $\mdbds_1$ and $\mdbds_2$ are given by \cref{eq:mdbds}, and $\VarSxh$ and $\VarSyh$ are given by \cref{eq:VarSh} in the case of $\freq \notin \cb{0, \pi}$ and \cref{eq:VarSh:real} in the case of $\freq \in \cb{0, \pi}$.

Now, using \cref{eq:Dsrelations}, it can be checked that
\begin{gather*}
\mdbds_1 = -\frac 1 2 \lrrb{\Identity_{p'} \otimes \Sxsinv \Sys + \Sxsinv \Sys \otimes \Identity_{p'}} \dbds = -\frac 1 2 \lrrb{\Identity_{p'} \otimes \Ds \Sys + \Ds \Sys \otimes \Identity_{p'} + 2 \Identity_{{p'}^2}} \dbds, \\
\mdbds_2 = \frac 1 2 \lrrb{\Identity_{p'} \otimes \Sysinv \Sxs + \Sysinv \Sxs \otimes \Identity_{p'}} \dbds = -\frac 1 2 \lrrb{\Identity_{p'} \otimes \Ds \Sxs + \Ds \Sxs \otimes \Identity_{p'} - 2 \Identity_{{p'}^2}} \dbds.
\end{gather*}
Thus, the natural plug-in estimator of $\asympSD^2$ is $\hat{\asympSD}^2 = \rb{\effN/\effN_1} \hat{\asympSD}^2_1 + \rb{\effN/\effN_2} \hat{\asympSD}^2_2$, where $\hat{\asympSD}^2_1 = \rb{\Permutation \mdbdh_1}^\top \hat{\VarSh}_1 \Permutation \mdbdh_1$ and $\hat{\asympSD}^2_2 = \rb{\Permutation \mdbdh_2}^\top \hat{\VarSh}_2 \Permutation \mdbdh_2$ in the case of $\freq \notin \cb{0, \pi}$, and $\hat{\asympSD}^2_1 = \mdbdh_1^\top \hat{\VarSh}_1 \mdbdh_1$ and $\hat{\asympSD}^2_2 = \mdbdh_2^\top \hat{\VarSh}_2 \mdbdh_2$ in the case of $\freq \in \cb{0, \pi}$. Here,
\begin{equation*}
\mdbdh_1 = -\frac 1 2 \lrrb{\Identity_{p'} \otimes \Dh \Syh + \Dh \Syh \otimes \Identity_{p'} + 2 \Identity_{{p'}^2}} \dbdh, \quad
\mdbdh_2 = -\frac 1 2 \lrrb{\Identity_{p'} \otimes \Dh \Sxh + \Dh \Sxh \otimes \Identity_{p'} - 2 \Identity_{{p'}^2}} \dbdh,
\end{equation*}
and
\begin{equation*}
\hat{\VarSh}_l = \frac 1 2 \begin{bmatrix}
\lrrb{\Identity_{p^2} + \Commutation} \lrrb{\Re \SDh_l \otimes \Re \SDh_l + \Im \SDh_l \otimes \Im \SDh_l} &
\lrrb{\Identity_{p^2} + \Commutation} \lrrb{\Im \SDh_l \otimes \Re \SDh_l - \Re \SDh_l \otimes \Im \SDh_l} \\
\lrrb{\Identity_{p^2} - \Commutation} \lrrb{\Re \SDh_l \otimes \Im \SDh_l - \Im \SDh_l \otimes \Re \SDh_l} &
\lrrb{\Identity_{p^2} - \Commutation} \lrrb{\Re \SDh_l \otimes \Re \SDh_l + \Im \SDh_l \otimes \Im \SDh_l}
\end{bmatrix},
\end{equation*}
for $l \in \cb{1, 2}$, in the case of $\freq \notin \cb{0, \pi}$ and
\begin{equation*}
\hat{\VarSh}_l = \lrrb{\Identity_{p^2} + \Commutation} \lrrb{\SDh_l \otimes \SDh_l},
\end{equation*}
for $l \in \cb{1, 2}$, in the case of $\freq \in \cb{0, \pi}$.

\cref{lem:plugin_var_est} in \cref{app:plugin_var_est} shows that this plug-in variance estimator $\hat{\asympSD}^2$ is consistent in the setting of \cref{cor:BE:SDD:abbr}. Together with \cref{cor:BE:SDD:abbr}, this implies that the Normal approximation is valid with $\hat{\asympSD}^2$ replacing $\asympSD^2$ via Slutsky's theorem.

\section{Simulations}
\label{sec:simulations}
In this section, we study the performance of our inference procedure using simulations. Data in all simulations and conditions $l \in \cb{ 1,2}$ is generated as a VAR$(1)$ process,
\begin{equation*}
    X_{l,t} = \matA_l X_{l,t-1} + \epsilon_{l,t} \, ,
\end{equation*}
where $\epsilon_{l,t} \sim N_{p}\lrrb{0, \Identity_p}$. The VAR(1) process is particularly convenient as by results from \citet{sun2018large}, the true spectral density at frequency $\freq$ for this process can be computed in a closed form as
\begin{equation*}
    \SDs_l(\freq) = \frac{1}{2\pi}\lrrb{\calA_l(e^{-\imath \freq})}^{-1} \Identity_p \lrrb{\lrrb{ \calA_l(e^{-\imath \freq})}^{-1}}^{\dagger}, 
\end{equation*}
where $\calA_l(z) = \Identity_p - \matA_l z$. Moreover, higher-order VAR processes can be rewritten as a VAR(1) process \citep{lutkepohl_VAR}. 
Finally, since a block diagonal transition matrix $\matA_l$ results in a block diagonal spectral density matrix $\SDs_l(\freq)$,  we can generate a sparse difference in inverse spectral densities by using block diagonal transition matrices in conditions $1$ and $2$ that only differ in a small block.

We perform 200 simulations for each combination of $p = 15, 50, 100$ and $T = 100, 500, 1500, 2500, 25000$. Data in conditions $1$ and $2$ is simulated as a VAR$(1)$ process with $T$ observations where the transition matrix is block diagonal with one large $\rb{p - 3} \times \rb{p - 3}$ block and one small $3 \times 3$ block. All non-zero entries were randomly generated from either a $\operatorname{Uniform}(-0.8,-0.4)$ or a $\operatorname{Uniform}(0.4,0.8)$ with equal probability. The transition matrix in condition $1$ is generated with 50\% sparsity in the larger block and 40\% sparsity in the smaller block. The transition matrix in condition $2$ is the same as condition $1$ except the smaller block is multiplied by $-1$. In this setup, the true difference in inverse spectral densities is sparse and only differs in the smaller $3 \times 3$ block.

Since it is not computationally feasible to perform inference using the standard D-trace loss estimating equation $\gradD$ for moderate to large $p$, only results for the two generalized estimating equations $\gradDL$ and $\gradDR$ are reported for $p > 15$. The spectral densities were estimated using Welch's estimator from \cref{eqn:welch_estimator}. Specifically, this corresponds to a time-averaged periodogram with no overlaps, a smoothing span of $B = \lceil T^{1/2} \rceil$ (as informed by Eq.~\ref{eq:BE:optB}), and a constant window value of $1$ \citep{welch1967use}. The SDD estimator and sparse de-biasing directions were estimated as in \cref{app:method_details}.

We perform our inference procedure for both the real and imaginary components of the 6 off-diagonal positions in the $3 \times 3$ block for a total of 12 entries. These are entries where the difference is potentially non-zero. We further perform inference on 12 randomly selected entries from the larger block. Since the larger blocks are the same between conditions, these entries correspond to the null hypothesis,  where the true difference is 0. Inference is performed for $10$ evenly spaced frequency values between $0$ and $\lfloor T^{1/2}/2 \rfloor$. 

We use 95\% confidence intervals and study the Type I error, coverage, and power for each of the combinations of $p$ and $T$ to evaluate the performance of our inference procedure when using different estimating equations. Each metric is averaged for each simulation over all frequencies and boxplots of the results for all 200 simulations are shown in \cref{fig:sim_inference_full}.
\begin{figure}[t!]
    \centering
    \includegraphics[width=\textwidth]{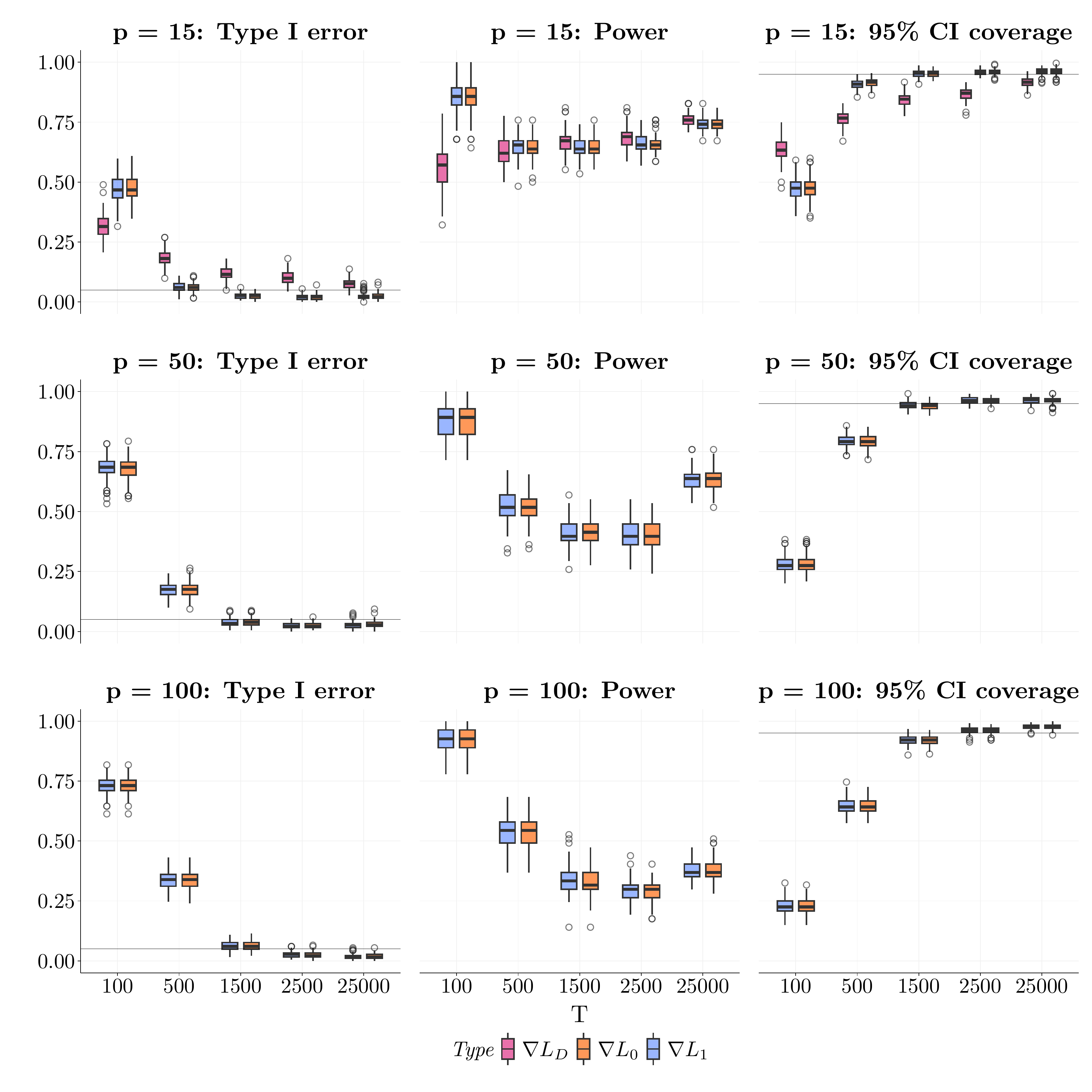} 
    \caption{Type I error, power, and coverage for different combinations of $p$ and $T$ using full variance estimation. }
    \label{fig:sim_inference_full}
\end{figure}
The results in \cref{fig:sim_inference_full} show that for $p = 15$, compared to the procedure based on $\gradD$, the variants based on $\gradDL$ and $\gradDR$ have slightly worse type I error and coverage for small sample sizes; however, they, too, achieve nominal type I error control and coverage when sample sizes are large. Interestingly, the $\gradD$ version has under coverage even when $T = 2500$ and $T = 25000$. This may be due to the use of GLASSO to estimate $\dbdh$ compared to CLIME for $\dbdh_{0}, \dbdh_{1}$. For $p = 50, 100$, we only have results for $\gradDL$ and $\gradDR$ variants due to aforementioned computational bottleneck. Here, we observe that at smaller sample sizes $T = 100, 500$, type I error and coverage are worse compared to $p = 15$, but as the sample sizes increase, they reach the nominal level. As for power, although it may appear that there is initially a decrease---between $T = 100$ and $500$ for $p = 15$ and between $T = 100$ and $1500$ for $p = 50, 100$---these are the range of sample sizes for which the method does \emph{not} yet have type I error control. Once the method hits the target type I error level, the power increases with the sample size, albeit at a slow rate. This is likely due to the fact that we are averaging over frequencies and it takes very large sample sizes to increase power given the scaling $\sqrt{T/B} = T^{1/4}$. This is further supported by our results for $T = 25000$.

While it is possible to compute variance estimates using careful index bookkeeping as mentioned in \cref{sec:method:var:est}, the variance computation can still be computationally expensive requiring long run times. We provide a computationally efficient variance estimation procedure and study its simulation performance in \cref{fig:sim_inference_efficient}. Compared to the full variance estimation in \cref{fig:sim_inference_full}, Type I error control is less tightly controlled for the efficient estimates. Since the computational complexity of the variance estimates scales quadratically with the number of non-zero components of $\dbdh$ and the number of non-zero components in the rows and columns of the variance scaling matrix $\matMh$ in \cref{lem:plugin_var_est} (see \cref{app:method_details}), this procedure works by reducing the number of non-zero components in each of these matrices. Further details are provided in  \cref{app:sim_results}.

\section{Application to EEG Data}
\label{sec:eeg}
In this section, we apply our inference procedure to electroencephalograms (EEG) data from \citet{hatlestad2022bids}. Data was recorded with a 64 channel EEG array for 111 healthy subjects at a sampling frequency of 1024 Hz which we downsampled to 512 Hz. Four minutes of brain activity was recorded while subjects were resting with their eyes closed. For 42 subjects, a second session was recorded 2--3 months after the initial session. We refer to these 42 subjects as those ``with follow-up" and the remaining 69 as those ``without follow-up." We use the pre-cleaned data provided in OpenNeuro Dataset ds003775 \citep{markiewicz2021openneuro}; see \citet{hatlestad2022bids} for specific cleaning steps.

We only analyze subjects without follow-up for which we estimated the difference in networks from 0-60s to 120-180s. We only analyze subjects without follow-up as brain networks have been shown to be temporally dynamic so the sparse difference assumption is more realistic within sessions \citep{zalesky2014time, nobukawa2019changes, Yue2025}. The 60-120s block was used as a rest. We estimated the differential network using the direct difference estimator, $\Dh$, in \cref{eqn:sdd_estimator} across the commonly-used Theta, Beta, Gamma and High-Gamma bands. Using the realification of a complex matrix, an estimate of $\Ds = \SDsinv_1 - \SDsinv_2$, denoted $\Dh$, can be generated by adding the (1,1) and $i$ times the (2,1) blocks of $\Dh$. We report two measures of sparsity for the SDD method based on either the SDD estimate alone (SDD-est) or using our inference (SDD-inf). For SDD-est, sparsity was defined as the proportion of zero entries in $\Dh$. For SDD-inf, for each non-zero entry in $\Dh$, we test the null hypothesis that both the real and imaginary components are zero. To do this, we can establish joint asymptotic normality of the real and imaginary components using methods similar to \cref{thm:BE:SDD} and \cref{cor:BE:SDD:abbr}. With this joint distribution we can generate an asymptotic $\chi^2$ test of the null hypothesis the both the real and imaginary components are 0. A level $\alpha = 0.05$ was used and an edge was determined to be differentially connected if the null hypothesis was rejected. In addition to SDD-inf, we also report results for SDD-est to compare our inference procedure to a direct difference estimator without inference. Five or six evenly spaced frequencies were considered for each band. They were (in Hz): Theta: (4,5,6,7,8), Beta: (12, 16, 20, 24, 28), Gamma: (30, 40, 50, 60, 70), High-gamma: (80, 95, 110, 125, 140, 150). 

While our method is the first direct difference estimator of its kind and the only with a calibrated inference procedure, we also compared its performance to one potential alternative, the joint graphical lasso with a fusion penalty (FGL) \citep{danaher2014joint}. This method aims to simultaneously estimate inverses that are believed to share structural similarities. It accomplishes this by using an $\ell_1$ penalty on the difference of inverse matrices to encourage sparsity. Unlike SDD, where we directly estimate the difference without estimating either inverse matrix, FGL estimates the inverse matrices in each condition. Implementation details for FGL are available in \cref{app:eeg_details}. Inference is not available for FGL and, as such, sparsity for this method was computed as the proportion of zero entries. Results averaged over subjects and frequencies within each band and are presented in \cref{fig:eeg}.

\begin{figure}[t!]
        \centering
        \includegraphics[width=0.8\linewidth]{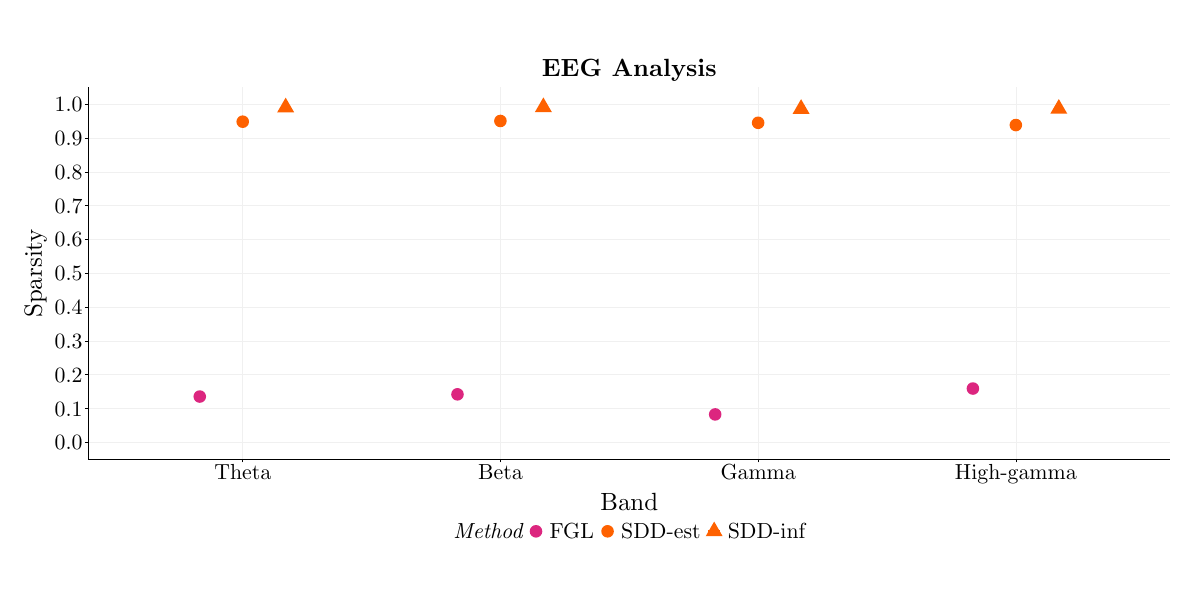}
        \caption{EEG analysis. The only valid inference method (our SDD-inf) is marked with $\blacktriangle$; the other methods that do not have inferential guarantees are marked with $\bullet$.}
    \label{fig:eeg}
\end{figure}

\cref{fig:eeg} shows that both SDD-est and SDD-inf estimate very sparse differences across all frequency bands. On the other hand, FGL, which is not a direct difference estimator estimates a very dense difference across all frequency bands.

\section{Discussion}
\label{sec:discussion}
We conclude with a discussion of our results and suggest directions for future research.

First, we resume the discussion of the relationship of the Normal-approximation error bound and $\dbds = \Hessian^{*-1}_{\EEcoef} \rb{\basis_{i_0} \otimes \basis_{j_0}}$, an issue that is raised by \cref{lem:BE:Dtrace:abbr}. The bound given by the lemma is nontrivial only when $\sparsity_{\dbds}$ is small. Relative to the conditions for consistent estimation of the differential network \citep[e.g., those of][]{hellstern2025spectral}, this is an additional requirement that arises from the need for consistent estimation of the projection direction. Because the Hessian of the D-trace loss is $\Hessian^*_{\EEcoef} = \EEcoef \rb{\Sxs \otimes \Sys} + \rb{1-\EEcoef} \rb{\Sys \otimes \Sxs}$, the characteristics of the individual networks dictate the size of $\sparsity_{\dbds}$. Indeed, for the two special cases of the generalized D-trace loss that are of interest for scalability ($\EEcoef \in \cb{0, 1}$), the requirement effectively translates into a restriction on the degrees of the nodes incident to the edge of inferential interest in the underlying networks. It is worth noting that the appearance of $\dbds$ in the error bound is a feature of the de-biased D-trace framework and not the time-series aspect of the problem.

On the one hand, it is hardly surprising that statistical inference is more difficult for an edge in a dense region of the networks than for one in a sparse region, even if the total number of edges that differ between the two networks remains small. Nevertheless, this naturally raises the question of whether additional restrictions on the local structure are inescapable for any procedure---de-biased D-trace or otherwise. This is a considerably deeper question that gets at the testability \citep{bahadur1956nonexistence, lehmann2005testing} of an edge in a high-dimensional differential network. Similar questions have been explored for related problems of high-dimensional Gaussian graphical models \citep{ren2015asymptotic} and for high-dimensional linear regression \citep{cai2017confidence, javanmard2018debiasing}; all of these lines of work are consistent with the claim that stronger structural assumptions are necessary in high dimensions for valid statistical inference than are needed for mere consistent estimation. More recently, \citet{bradic2022testability} found that the testability of a high-dimensional linear model is fundamentally tied to the column (or row) sparsity of the population precision matrix of the covariates. Since, the role of the population precision matrix is played by the inverse of $\Hessian^*_{\EEcoef}$ in the D-trace framework, na\"{i}vely extrapolating these results would seem to suggest that dependence on $\sparsity_{\dbds}$ may be unavoidable in some fundamental sense. However, to the best of our knowledge, a comprehensive answer to this question is not yet available---certainly for our specific problem of comparing a given component across two high-dimensional inverse spectral densities, but also, it appears, for the more widely considered problem of comparing a component of two high-dimensional precision matrices based on two i.i.d.~samples. We leave the investigation of this important question to future research.

In the two-sample setting, these additional structural assumptions, which are needed for valid statistical inference, tend to relate directly to the underlying networks, rather than only to the differential network. These assumptions differ from the sparsity assumption needed for consistent estimation and inference for each network. Moreover, the direct difference approach affords clear empirical benefits, as demonstrated in simulations and real data analysis. Finally, inference based on direct estimation can also lead to computational savings. Our SDD framework requires only one initial estimation followed by one de-biasing step, whereas a hypothetical procedure based on separate estimation would require two initial estimations followed by two high-dimensional matrix inversions. This also means that the number of hyperparameters that need to be tuned is halved for our method compared to our hypothetical competitor. Although such advantages are unlikely to show up in, say, sample complexity order, our method is easier to implement and tune, and runs in about half the time than a method based on separate estimation would.

Next, we discuss the possibility of using alternative estimators of the spectral density. Our analysis of the information geometry of the de-biased D-trace estimation procedure makes clear that any spectral density estimator will lead to a final estimator that is asymptotically Normal and unbiased as long as it satisfies the two key conditions of \cref{lem:BE:Dtrace:abbr}: a Berry--Esseen condition (Eq.~\ref{cond:BE}) and a large-deviation condition (Eq.~\ref{cond:LD}). The question, then, is which alternative spectral density estimators, if any, satisfy these two conditions.

Of the two, the Berry--Esseen condition presents the greater challenge. Although recent years have seen much progress in the field of statistical inference for high-dimensional time series, the picture is incomplete when it comes to the suitability of various spectral density estimators for high-dimensional statistical inference under different dependence frameworks. One promising class of estimators is that of general lag-window estimators, which have the form
\[
\SDh(\freq) = \frac{1}{2 \pi} \sum_{h = -B}^{B} e^{-\imath \freq h} \hat{\Autocov}(h) K(h/B),
\]
where $\hat{\Autocov}(h) = T^{-1} \sum_{t = \max\cb{1, -h+1}}^{\min\cb{T, T-h}} X_{t+h} X_t^\top$ is the sample autocovariance and $K$ is a kernel function that is bounded, continuous, and even on $[-1, 1]$ with $K(0) = 1$. Indeed, \citet{zhang2025spectral} already considered the possibility of extending their results to such estimators, but observed that strong correlations between nearby frequencies pose substantial obstacles to deriving the asymptotic covariance. On the other hand, \citet{chang2025statistical} proved an analogous Chernozhukov--Chetverikov--Kato-type Gaussian approximation result \citep{chernozhukov2013gaussian, chernozhukov2017central} for a special case of the lag-window estimator that uses the flat-top kernel \citep{politis2011higher}, albeit in an alternative dependence framework of strong mixing. Their result suggests that Berry--Esseen-type bounds for low-dimensional components may be achievable for some special cases of lag-window estimators in the functional dependence framework as well. Relatedly, it is of interest to ask whether the classical asymptotic analysis for the smoothed periodogram can be extended and refined to handle high dimensionality and modern dependence frameworks. These appear to be important questions for future work.

\bibliographystyle{plainnat}
\bibliography{references}

\clearpage
\appendix

\setcounter{equation}{0}
\setcounter{fact}{0}
\setcounter{figure}{0}
\setcounter{lemma}{0}
\setcounter{table}{0}
\setcounter{theorem}{0}

\renewcommand{\theequation}{S\arabic{equation}}
\renewcommand{\thefact}{S\arabic{fact}}
\renewcommand{\thefigure}{S\arabic{figure}}
\renewcommand{\thelemma}{S\arabic{lemma}}
\renewcommand{\thetable}{S\arabic{table}}
\renewcommand{\thetheorem}{S\arabic{theorem}}

\def\theHequation{S\arabic{equation}}
\def\theHfact{S\arabic{fact}}
\def\theHfigure{S\arabic{figure}}
\def\theHlemma{S\arabic{lemma}}
\def\theHtable{S\arabic{table}}
\def\theHtheorem{S\arabic{theorem}}

\onecolumn
\addcontentsline{toc}{section}{Appendix}
\part{Appendix} %
\parttoc %

\clearpage

 
\begin{table}[ht]
\centering
\caption{Summary of select notation.}
\label{tab:notation}
\small
\begin{tabular}{@{}lll@{}}
\toprule
\textbf{Symbol} & \textbf{Definition} & \textbf{Ref.} \\
\midrule
\multicolumn{3}{@{}l}{\textit{Data and spectral densities}} \\[2pt]
$X_{l,\cdot} = \rb{X_{l, t}}_{t \in \integers}$ & $p$-dimensional stationary time series for condition $l \in \cb{1,2}$ & \cref{sec:problem} \\
$T_l$ & Number of data points for condition $l$ & \cref{sec:problem} \\
$p$ & Dimension of the time series & \cref{sec:problem} \\
$\Autocov_l(h)$ & Autocovariance function $\EE\sqb{X_{l, t+h} X_{l, t}^\top}$ for condition $l$ & \cref{sec:problem} \\
$\SDs_l(\freq)$ & Spectral density at frequency $\freq$ for condition $l$ & \cref{eq:SD} \\
$\freq$ & Frequency of interest, $\freq \in [0, \pi]$ & \cref{sec:problem} \\
$\Ss_l$ & Realification of $\SDs_l$; equals $\SDs_l$ when $\freq \in \cb{0, \pi}$ & \cref{sec:problem} \\
$\Ssinv_l$ & Inverse of $\Ss_l$ (realification of $\SDsinv_l$) & \cref{sec:problem} \\
$p'$ & $2p$ when $\freq \notin \cb{0, \pi}$;\; $p$ when $\freq \in \cb{0, \pi}$ & \cref{subsubsec:inference:details} \\
\midrule
\multicolumn{3}{@{}l}{\textit{Target parameter}} \\[2pt]
$\Ds = \Sxsinv - \Sysinv$ & True difference in inverse spectral densities & \cref{eq:Ds} \\
$\Ds_0$ & Parameter of inferential interest; $(1, 1)$-th component of $\Ds$ & \cref{sec:method} \\
$\Ds_u$ & Remaining (nuisance) components of $\vec(\Ds)$ & \cref{sec:method} \\
$d$ & Dimension of $\vec(\Ds)$ & \cref{sec:method} \\
\midrule
\multicolumn{3}{@{}l}{\textit{Estimation}} \\[2pt]
$\SDh_l$ & Welch's non-overlapping estimator of $\SDs_l(\freq)$ & \cref{eqn:welch_estimator} \\
$\Sh_l$ & Estimator of $\Ss_l$ & \cref{algo:method} \\
$B_l$ & Window size (smoothing span) for Welch's estimator, condition $l$ & \cref{eqn:welch_estimator} \\
$\effN_l = T_l/B_l$ & Effective sample size for condition $l$ & \cref{thm:BE:projWelch} \\
$\effN = \effN_1 + \effN_2$ & Total effective sample size & \cref{thm:BE:projWelch} \\
$\lossD(\D;\Sxs,\Sys)$ & D-trace loss function & \cref{eq:Dtrace} \\
$\Hessian^*$ & Population Hessian, $\lrrb{\Sxs \otimes \Sys + \Sys \otimes \Sxs} / 2$ & \cref{eq:Dtrace} \\
$\hat{\Hessian}$ & Sample Hessian, $\Hessian\rb{\Sxh,\Syh}$ & \cref{algo:method} \\
$\Dh$ & $\ell_1$-penalized D-trace estimator of $\Ds$ & \cref{eqn:sdd_estimator} \\
$\dbds = \Hessian^{*-1} \basis_0$ & True projection direction & \cref{eqn:debias_direction} \\
$\dbdh$ & Estimated sparse projection direction & \cref{eqn:debias_direction} \\
$\Dt_0$ & De-biased estimator of $\Ds_0$ & \cref{eq:Dt} \\
$\regD, \regH$ & Regularization parameters for $\Dh$ and $\dbdh$ & \cref{eqn:sdd_estimator,eqn:debias_direction,} \\
\midrule
\multicolumn{3}{@{}l}{\textit{Sparsity and regularity}} \\[2pt]
$\sparsity_{\Ds} = \norm{\Ds}_0$ & Number of nonzero entries in $\Ds$ & \cref{lem:BE:Dtrace:abbr} \\
$\sparsity_{\dbds} = \norm{\dbds}_0$ & Number of nonzero entries in $\dbds$ & \cref{lem:BE:Dtrace:abbr} \\
$\sparsity_{\text{tot}}$ & $\onenorm{\Ds} \onenorm{\dbds} + \sparsity_{\Ds} + \sparsity_{\dbds}$ & \cref{thm:BE:SDD} \\
$\mineigen_l$ & Minimum eigenvalue of $\Ss_l$ & \cref{lem:BE:Dtrace:abbr} \\
\midrule
\multicolumn{3}{@{}l}{\textit{Functional dependence framework}} \\[2pt]
$\inno_{t}$ & i.i.d.\ random innovations & \cref{eqn:fn_dep_rep} \\
$\filter_{j}$ & Measurable function: $X_{t, j} = \filter_{j}\rb{\inno_{t}, \inno_{t-1}, \ldots}$ & \cref{eqn:fn_dep_rep} \\
$\dependence_{q, j}(t)$ & Functional dependence measure for component $j$ at lag $t$, order $q$ & \cref{eqn:dependence} \\
$\Lqnorm{X_{\cdot, j}}$ & Dependence-adjusted moment of order $q$ for component $j$ & \cref{eq:dam} \\
$\damrho \in (0,1)$ & Decay rate constant in the dependence-adjusted moment & \cref{eq:dam} \\
$\norm{X_{\cdot, j}}_{\daOpsi_\daOalpha}$ & Dependence-adjusted $\daOpsi_\daOalpha$-norm for component $j$ & \cref{eq:daO} \\
$\norm{X_{l, \cdot}}_{\daOpsi_{\daOalpha_l}}$ & $\max_j \norm{X_{l, \cdot, j}}_{\daOpsi_{\daOalpha_l}}$;\; assumed finite & \cref{assm} \\
$\daOalpha_l \geq 0$ & Tail dependence for condition $l$ (if i.i.d: $\tfrac{1}{2}=$ sub-Gaussian; $1=$ sub-exponential) & \cref{eq:daO} \\
\bottomrule
\end{tabular}
\end{table}

\begin{table}[ht]
\centering
\caption{Summary of select notation (continued).}
\label{tab:notation:cont}
\small
\begin{tabular}{@{}lll@{}}
\toprule
\textbf{Symbol} & \textbf{Definition} & \textbf{Ref.} \\
\midrule
\multicolumn{3}{@{}l}{\textit{Inference quantities}} \\[2pt]
$\asympSD^2$ & Asymptotic variance of $\Dt_0$;\; $\asympSD^2 = \rb{\effN/\effN_1} \asympSD^2_1 + \rb{\effN/\effN_2} \asympSD^2_2$ & \cref{thm:BE:projWelch,thm:BE:SDD} \\
$\asympSD_l^2$ & Variance contribution from condition $l$ & \cref{thm:BE:projWelch} \\
$\asympSDh^2$ & Plug-in estimator of $\asympSD^2$ & \cref{lem:plugin_var_est} \\
$\mdbds_1, \mdbds_2$ & Projection vectors for the asymptotic linear representation & \cref{eq:mdbds} \\
$\Psi_l$ & Asymptotic variance matrix of $\vec\rb{\Sh_l}$ & \cref{eq:VarSh,eq:VarSh:real} \\
$\Permutation =\lrsqb{\Permutation_1\ \Permutation_2}$ & Permutation relating $\vec\rb{\Sigma}$ to $\lrsqb{\begin{smallmatrix}\vec \rSD \\ \vec \iSD \end{smallmatrix}}$ & \cref{subsubsec:inference:details,prop:asympcov} \\
$\Commutation$ & Commutation matrix, $\sum_{i = 1}^{p} \sum_{j = 1}^{p} \Basis_{ij} \otimes \Basis_{ji}$ & \cref{subsubsec:inference:details,prop:asympcov}\\
$\Basis_{ij}$ & Standard basis matrix in $\reals^{p\times p}$ with $1$ in position $(i,j)$ & \cref{subsubsec:inference:details} \\
$\Sxhrate, \Syhrate$ & Large deviation bound for $\infnorm{\Sxh-\Sxs}, \infnorm{\Syh-\Sys}$ & \cref{cond:LD} \\
$ \LDprob$ & Tail probability for the large deviation event & \cref{lem:BE:Dtrace:abbr} \\
$\BEprob$ & Berry--Esseen approximation error & \cref{lem:BE:Dtrace:abbr}, \cref{thm:BE:projWelch} \\
\midrule
\multicolumn{3}{@{}l}{\textit{Alternative estimating equations}} \\[2pt]
$\gradDL,\, \gradDR$ & Gradients of the $\EEcoef = 0$ and $\EEcoef = 1$ estimating equations & \cref{eq:pop:1,eq:pop:0} \\
$\Dtgen{0}$ & De-biased estimator using the generalized ($\EEcoef$) estimating equation & \cref{eq:Dt:general} \\
\midrule
\multicolumn{3}{@{}l}{\textit{General notation}} \\[2pt]
$\basis_i$ & $i$-th standard basis vector & \cref{subsubsec:inference:details} \\
$\matA \otimes \matB$ & Kronecker product between $\matA$ and $\matB$ & \cref{sec:intro} \\
$\Re \matA, \Im \matA$ & Real and imaginary parts of $\matA$ & \cref{sec:intro} \\
$\overline{\matA}$ & Complex conjugate of $\matA$ & \cref{sec:intro} \\
$\|\cdot\|_1,\,\|\cdot\|_2,\,\|\cdot\|_\infty$ & $\ell_1$, $\ell_2$, and $\ell_{\infty}$ norms & \cref{sec:intro} \\
$\|\cdot\|_0$ & $\ell_0$ ``norm'' & \cref{sec:intro} \\
$\|X\|_{L_q}$ & $L_q$ norm $(\mathbb{E}|X|^q)^{1/q}$ & \cref{sec:intro} \\
$x_t \asymp y_t$ & $x_t$ and $y_t$ are asymptotically of the same order & \cref{sec:intro} \\
$x \lesssim y$ & $x \leq C y$ for some constant $C > 0$ & \cref{sec:intro} \\
$x \lesssim_{\alpha} y$ & $x \leq C_{\alpha} y$ for some constant $C_{\alpha} > 0$ depending on parameter $\alpha$ & \cref{sec:intro} \\
\bottomrule
\end{tabular}
\end{table}

\section{Facts about the D-trace loss} \label{app:loss}
The \emph{D-trace loss} \citep{zhang2014sparse, yuan2017differential} is the quadratic loss
\begin{equation*}
\lossD\rb{\D; \Sx, \Sy} = \frac 1 2 \vec\rb{\D}^\top \Hessian \vec\rb{\D} - \vec\rb{\D}^\top \vec\rb{\Sy-\Sx},
\end{equation*}
where
\begin{equation*}
\Hessian = \Hessian\rb{\Sx, \Sy} = \frac 12 \lrrb{\Sx \otimes \Sy + \Sy \otimes \Sx}.
\end{equation*}
The use of $\lossD$ for estimating the difference $\Ds = \Sxsinv-\Sysinv$ is motivated by the fact that $\Ds$ is a solution to the population-level linear estimating equation
\begin{equation} \label{eq:pop:Dtrace}
0 = \EE \gradD\rb{\D; \Sxh, \Syh} = \gradD\rb{\D; \Sxs, \Sys} = \Hessian^* \vec\rb{\D} - \vec\rb{\Sys-\Sxs},
\end{equation}
where $\Hessian^* = \Hessian\rb{\Sxs, \Sys}$. The solution is unique if $\Sxs$ and $\Sys$ are both full rank. Replacing $\Sxs$ and $\Sys$ with their sample estimates $\Sxh$ and $\Syh$ yields a sample procedure for estimating $\Ds$. 

Put $\hat{\Hessian} = \Hessian\rb{\Sxh, \Syh}$. Note that the Hessian of $\lossD$ does not depend on $\Ds$. Since
\begin{gather}
\Sxh \otimes \Syh - \Sxs \otimes \Sys = \lrrb{\Sxh-\Sxs} \otimes \Sys + \Sxs \otimes \lrrb{\Syh-\Sys} + \lrrb{\Sxh-\Sxs} \otimes \lrrb{\Syh-\Sys}, \label{eq:SxXSydiff}\\
\Syh \otimes \Sxh - \Sys \otimes \Sxs = \lrrb{\Syh-\Sys} \otimes \Sxs + \Sys \otimes \lrrb{\Sxh-\Sxs} + \lrrb{\Syh-\Sys} \otimes \lrrb{\Sxh-\Sxs}, \label{eq:SyXSxdiff}
\end{gather}
we have
\begin{multline} \label{eq:Hessdiff}
\hat{\Hessian}-\Hessian^*
= \frac 1 2 \Big\{\lrrb{\Sxh-\Sxs} \otimes \Sys + \Sys \otimes \lrrb{\Sxh-\Sxs} + \lrrb{\Syh-\Sys} \otimes \Sxs  + \Sxs \otimes \lrrb{\Syh-\Sys} \\
+ \lrrb{\Sxh-\Sxs} \otimes \lrrb{\Syh-\Sys} + \lrrb{\Syh-\Sys} \otimes \lrrb{\Sxh-\Sxs} \Big\}.
\end{multline}

As for the gradient $\gradD$, by \cref{eq:pop:Dtrace}
\begin{align}
\gradD\rb{\Ds; \Sxh, \Syh}
&= \gradD\rb{\Ds; \Sxh, \Syh} - \gradD\rb{\Ds; \Sxs, \Sys} \nonumber\\
&= \lrrb{\hat{\Hessian} - \Hessian^*} \vec\rb{\Ds} - \vec\rb{\Syh-\Sys} + \vec\rb{\Sxh-\Sxs}. \label{eq:graddiff}
\end{align}
Substituting \cref{eq:Hessdiff} into \cref{eq:graddiff} and using the identity
\begin{equation*}
\lrrb{\matB^\top \otimes \matA} \vec\lrrb{\matX} = \vec\lrrb{\matA \matX \matB}
\end{equation*}
with
\begin{equation} \label{eq:Dsrelations}
\begin{aligned}
\Sys \Ds &= \Sys \Sxsinv - \Identity_{p'}, & \Ds \Sys &= \Sxsinv \Sys - \Identity_{p'}, \\
\Sxs \Ds &= \Identity_{p'} - \Sxs \Sysinv, & \Ds \Sxs &= \Identity_{p'} - \Sysinv \Sxs,
\end{aligned}
\end{equation}
we have
\begin{equation*}
\gradD\lrrb{\Ds; \Sxh, \Syh}
= \frac 1 2
\begin{aligned}[t]
\vec\Big\{ & \Sys \Sxsinv \lrrb{\Sxh-\Sxs} + \lrrb{\Sxh-\Sxs} \Sxsinv \Sys \\
& - \Sxs \Sysinv \lrrb{\Syh-\Sys} - \lrrb{\Syh-\Sys} \Sysinv \Sxs \\
& + \lrrb{\Syh-\Sys} \Ds \lrrb{\Sxh-\Sxs} + \lrrb{\Sxh-\Sxs} \Ds \lrrb{\Syh-\Sys} \Big\}.
\end{aligned}
\end{equation*}
Using the identity
\begin{equation*}
\vec(\matA \matX) = \vec(\matA \matX \Identity) = \rb{\Identity \otimes \matA} \vec(\matX), \quad \vec(\matX \matA) = \vec(\Identity \matX \matA) = \rb{\matA^\top \otimes \Identity} \vec(\matX),
\end{equation*}
this is
\begin{align}
\gradD\lrrb{\Ds; \Sxh, \Syh}
= \frac 12 \Big[ & \lrrb{\Identity_{p'} \otimes \Sys \Sxsinv + \Sys \Sxsinv \otimes \Identity_{p'}} \vec\lrrb{\Sxh-\Sxs} \nonumber\\
& - \lrrb{\Identity_{p'} \otimes \Sxs \Sysinv + \Sxs \Sysinv \otimes \Identity_{p'}} \vec\lrrb{\Syh-\Sys} \nonumber\\
& + \lrcb{\lrrb{\Sxh-\Sxs} \otimes \lrrb{\Syh-\Sys} + \lrrb{\Syh-\Sys} \otimes \lrrb{\Sxh-\Sxs}} \vec(\Ds) \Big]. \label{eq:gradD}
\end{align}

\subsection{Restricted strong convexity of the sample Hessian}

In this subsection, we prove that the sample Hessian $\hat{\Hessian}$ satisfies a restricted strong convexity (RSC) condition as long as $\Sxh$ and $\Syh$ are each sufficiently close to $\Sxs$ and $\Sys$.

It is easy to show that when $\Sxs$ and $\Sys$ are both full rank, the population-level D-trace loss is strongly convex: $\vecv^\top \Hessian^* \vecv \geq \mineigen_1 \mineigen_2 \norm{\vecv}_2^2$ for all $\vecv \in \reals^d$, where $\mineigen_1$ and $\mineigen_2$ are the minimum eigenvalues of $\Sxs$ and $\Sys$, respectively. This follows from \cref{fact:kronspec}, and in the case of $\freq \notin \cb{0, \pi}$, \cref{fact:equal_evals}.

\begin{fact}[Spectrum of a Kronecker product \citep{horn1991topics}] \label{fact:kronspec}
Let $\matA$ and $\matB$ be square matrices. If $\mineigen_{\matA}$ is an eigenvalue of $\matA$ and $\mineigen_{\matB}$, an eigenvalue of $\matB$, then $\mineigen_{\matA} \mineigen_{\matB}$ is an eigenvalue of $\matA \otimes \matB$.
\end{fact}

\begin{proof}
Let $\vecv$ be an eigenvector of $\matA$ associated with $\mineigen_{\matA}$, and $\vecw$, an eigenvector of $\matB$ associated with $\mineigen_{\matB}$. Then,
\begin{equation*}
\lrrb{\matA \otimes \matB} \lrrb{\vecv \otimes \vecw} = \lrrb{\matA \vecv} \otimes \lrrb{\matB \vecw} = \lrrb{\mineigen_{\matA} \vecv} \otimes \lrrb{\mineigen_{\matB} \vecw} = \mineigen_{\matA} \mineigen_{\matB} \lrrb{\vecv \otimes \vecw}.
\end{equation*}
\end{proof}

\begin{fact}[Spectrum of the realification of a complex matrix \citep{horn1985matrix}] \label{fact:equal_evals}
Let $\SD = \rSD + \imath \iSD \in \complex^{p \times p}$ be a complex matrix. Let $\Sigma$ be the realification of $\SD$: $\Sigma = \lrsqb{\begin{smallmatrix} \rSD & -\iSD \\ \iSD & \phantom{-}\rSD \end{smallmatrix}}$. If $\mineigen$ is an eigenvalue of $\SD$, then $\mineigen$ and $\bar{\mineigen}$ are both eigenvalues of $\Sigma$.
\end{fact}

\begin{proof}
Let $\vecv$ be an eigenvector of $\SD$ associated with $\mineigen$: $\SD \vecv = \rSD \vecv + \imath \iSD \vecv = \mineigen \vecv$. Then, $\bar{\SD} \bar{\vecv} = \rSD \bar{\vecv} - \imath \iSD \bar{\vecv} = \bar{\mineigen} \bar{\vecv}$, and
\begin{gather*}
\Sigma \begin{bmatrix} \phantom{-\imath} \vecv \\ -\imath \vecv \end{bmatrix}
= \begin{bmatrix} \rSD & -\iSD \\ \iSD & \phantom{-}\rSD \end{bmatrix} \begin{bmatrix} \phantom{-\imath} \vecv \\ -\imath \vecv \end{bmatrix}
= \begin{bmatrix} \rSD \vecv + \imath \iSD \vecv \\ \iSD \vecv - \imath \rSD \vecv \end{bmatrix}
= \begin{bmatrix} \phantom{-\imath} \SD \vecv \\ -\imath \SD \vecv \end{bmatrix}
= \nu \begin{bmatrix} \phantom{-\imath} \vecv \\ -\imath \vecv \end{bmatrix}, \\
\Sigma \begin{bmatrix} \phantom{\imath} \bar{\vecv} \\ \imath \bar{\vecv} \end{bmatrix}
= \begin{bmatrix} \rSD & -\iSD \\ \iSD & \phantom{-}\rSD \end{bmatrix} \begin{bmatrix} \phantom{\imath} \bar{\vecv} \\ \imath \bar{\vecv} \end{bmatrix}
= \begin{bmatrix} \rSD \bar{\vecv} - \imath \iSD \bar{\vecv} \\ \iSD \bar{\vecv} + \imath \rSD \bar{\vecv} \end{bmatrix}
= \begin{bmatrix} \phantom{\imath} \bar{\SD} \bar{\vecv} \\ \imath \bar{\SD} \bar{\vecv} \end{bmatrix}
= \bar{\mineigen} \begin{bmatrix} \phantom{\imath} \bar{\vecv} \\ \imath \bar{\vecv} \end{bmatrix}.
\end{gather*}
\end{proof}

Since $\SDs_1$ and $\SDs_2$ are Hermitian, all their eigenvalues are real. Hence, their realifications $\Sxs$ and $\Sys$ have the same spectrum as $\SDs_1$ and $\SDs_2$ with the multiplicity of each distinct eigenvalue doubled.

When $\Hessian^*$ is strongly convex, $\hat{\Hessian}$ can be shown to satisfy an RSC condition for sufficiently small $\Sxh-\Sxs$ and $\Syh-\Sys$. For $0 \leq \sparsity \leq d$ and $\Conewidth > 0$, let $\Cone\rb{\sparsity, \Conewidth}$ be the set
\begin{equation} \label{eq:cone}
\Cone\rb{\sparsity, \Conewidth} = \lrcb{v \in \reals^d: \onenorm{\vecv_{\Supp^c}} \leq \Conewidth \onenorm{\vecv_{\Supp}} \text{ for some } \Supp \subseteq [d] \text{ such that } \abs{\Supp} \leq s}.
\end{equation}

\begin{lemma} \label{lem:RSC}
Let $\mineigen_1$ and $\mineigen_2$ be the minimum eigenvalues of $\Sxs$ and $\Sys$, respectively. Suppose
\begin{equation*}
\infnorm{\Sxh-\Sxs} \leq \Sxhrate, \quad \infnorm{\Syh-\Sys} \leq \Syhrate,
\end{equation*}
where $\Sxhrate = \Sxhrate(d, T_1) > 0$ and $\Syhrate = \Syhrate(d, T_2) > 0$ satisfy
\begin{equation*}
2 \lrrb{1+\Conewidth}^2 \sparsity \lrrb{\infnorm{\Sys} \Sxhrate + \infnorm{\Sxs} \Syhrate + \Sxhrate \Syhrate} \leq \mineigen_1 \mineigen_2
\end{equation*}
for some $0 \leq \sparsity \leq d$ and $\Conewidth > 0$. Then,
\begin{equation*}
\vecv^\top \hat{\Hessian} \vecv \geq \frac{\mineigen_1 \mineigen_2}{2} \norm{\vecv}_2^2 \quad \text{for all} \quad \vecv \in \Cone\rb{\sparsity, \Conewidth}.
\end{equation*}
\end{lemma}

\begin{proof}
For any $\vecv \in \Cone\rb{\sparsity, \Conewidth}$,
\begin{equation*}
\onenorm{\vecv} = \onenorm{\vecv_{\Supp}} + \onenorm{\vecv_{\Supp^c}} \leq \lrrb{1+\Conewidth} \onenorm{\vecv_{\Supp}} \leq \lrrb{1+\Conewidth} \sqrt{s} \lrnorm{\vecv_{\Supp}}_2 \leq \lrrb{1+\Conewidth} \sqrt{s} \norm{\vecv}_2.
\end{equation*}
By H\"{o}lder's inequality,
\begin{equation*}
\vecv^\top \hat{\Hessian} \vecv = \vecv^\top \Hessian^* \vecv + \vecv^\top \lrrb{\hat{\Hessian}-\Hessian^*} \vecv \geq \vecv^\top \Hessian^* \vecv  - \infnorm{\hat{\Hessian}-\Hessian^*} \norm{\vecv}_1^2.
\end{equation*}
By the definition of $\Hessian^*$, \cref{fact:kronspec}, and in the case of $\freq \notin \cb{0, \pi}$, \cref{fact:equal_evals},
\begin{equation*}
\vecv^\top \Hessian^* \vecv \geq \mineigen_1 \mineigen_2 \norm{\vecv}^2.
\end{equation*}
Thus,
\begin{equation*}
\vecv^\top \hat{\Hessian} \vecv \geq \lrrb{\mineigen_1 \mineigen_2 - \lrrb{1+\Conewidth}^2 \sparsity \infnorm{\hat{\Hessian}-\Hessian^*}} \norm{\vecv}_2^2.
\end{equation*}
By \cref{eq:Hessdiff}, under the condition of the lemma,
\begin{equation*}
\lrrb{1+\Conewidth}^2 \sparsity \infnorm{\hat{\Hessian}-\Hessian^*} \leq \lrrb{1+\Conewidth}^2 \sparsity \lrrb{\infnorm{\Sys} \Sxhrate + \infnorm{\Sxs} \Syhrate + \Sxhrate \Syhrate} \leq \frac{\mineigen_1 \mineigen_2}{2}.
\end{equation*}
Therefore,
\begin{equation*}
\vecv^\top \hat{\Hessian} \vecv \geq \frac{\mineigen_1 \mineigen_2}{2} \norm{\vecv}_2^2.
\end{equation*}
\end{proof}

\subsection{Initial estimators}

\begin{lemma} \label{lem:Dh}
Let $\mineigen_1$ and $\mineigen_2$ be the minimum eigenvalues of $\Sxs$ and $\Sys$, respectively. Let $\sparsity_{\Ds} = \norm{\Ds}_0$. Suppose
\begin{equation*}
\infnorm{\Sxh-\Sxs} \leq \Sxhrate, \quad \infnorm{\Syh-\Sys} \leq \Syhrate,
\end{equation*}
where $\Sxhrate = \Sxhrate(d, T_1) > 0$ and $\Syhrate = \Syhrate(d, T_2) > 0$ satisfy
\begin{equation*}
\Hessrate = \infnorm{\Sys} \Sxhrate + \infnorm{\Sxs} \Syhrate + \Sxhrate \Syhrate \leq \frac{\mineigen_1 \mineigen_2}{32 \sparsity_{\Ds}}.
\end{equation*}
Suppose $\regD$ satisfies
\begin{equation*}
\regD \geq 2 \lrrb{\onenorm{\Ds} \Hessrate + \Sxhrate + \Syhrate}.
\end{equation*}
Then,
\begin{equation*}
\onenorm{\Dh-\Ds} \lesssim \frac{\sparsity_{\Ds} \regD}{\mineigen_1 \mineigen_2}.
\end{equation*}
\end{lemma}

\begin{proof}
The lemma can be proved by the standard analysis for the LASSO, e.g., Corollary 1 of \citet{negahban2012unified}.
\end{proof}

\begin{lemma} \label{lem:dbdh}
Let $\mineigen_1$ and $\mineigen_2$ be the minimum eigenvalues of $\Sxs$ and $\Sys$, respectively. Let $\sparsity_{\dbds} = \norm{\dbds}_0$. Suppose
\begin{equation*}
\infnorm{\Sxh-\Sxs} \leq \Sxhrate, \quad \infnorm{\Syh-\Sys} \leq \Syhrate,
\end{equation*}
where $\Sxhrate = \Sxhrate(d, T_1) > 0$ and $\Syhrate = \Syhrate(d, T_2) > 0$ satisfy
\begin{equation*}
\Hessrate = \infnorm{\Sys} \Sxhrate + \infnorm{\Sxs} \Syhrate + \Sxhrate \Syhrate \leq \frac{\mineigen_1 \mineigen_2}{8 \sparsity_{\dbds}}.
\end{equation*}
Suppose $\regH$ satisfies
\begin{equation*}
\regH \geq \onenorm{\dbds} \Hessrate.
\end{equation*}
Then,
\begin{equation*}
\onenorm{\dbdh-\dbds} \lesssim \frac{\sparsity_{\dbds} \regH}{\mineigen_1 \mineigen_2}.
\end{equation*}
\end{lemma}

\begin{proof}
The proof follows the standard analysis for Dantzig selectors (cf., \citet{candes2007dantzig,bickel2009simultaneous}). On the event of the hypothesis, $\dbds$ satisfies the Dantzig constraint. Thus, on one hand,
\begin{equation} \label{eq:Dantzig:graddiff}
\infnorm{\hat{\Hessian} \lrrb{\dbdh-\dbds}} = \infnorm{\hat{\Hessian} \dbdh - \basis_0 - \hat{\Hessian} \dbds + \basis_0} \leq 2 \regH.
\end{equation}
On the other hand, writing $d\dbdh = \dbdh-\dbds$,
\begin{multline*}
\norm{\dbds}_1 \geq \norm{\dbdh}_1
= \norm{\dbds + d\dbdh}_1
= \norm{\dbds_{\Supp_{\dbds}} + \dbds_{\Supp_{\dbds}^c} + d\dbdh_{\Supp_{\dbds}} + d\dbdh_{\Supp_{\dbds}^c}}_1 \\
\geq \norm{\dbds_{\Supp_{\dbds}}}_1 + \norm{d\dbdh_{\Supp_{\dbds}^c}}_1 - \norm{d\dbdh_{\Supp_{\dbds}}}_1
= \norm{\dbds}_1 + \norm{d\dbdh_{\Supp_{\dbds}^c}}_1 - \norm{d\dbdh_{\Supp_{\dbds}}}_1,
\end{multline*}
and hence,
\begin{equation*}
\norm{\lrrb{\dbdh-\dbds}_{\Supp_{\dbds}^c}}_1 \leq \norm{\lrrb{\dbdh-\dbds}_{\Supp_{\dbds}}}_1.
\end{equation*}
In other words, $d\dbdh = \dbdh-\dbds \in \Cone(\sparsity_{\dbds}, 1)$, and
\begin{multline} \label{eq:Dantzig:dbdhdiff:1}
\norm{\dbdh-\dbds}_1
= \norm{\lrrb{\dbdh-\dbds}_{\Supp_{\dbds}}}_1 + \norm{\lrrb{\dbdh-\dbds}_{\Supp_{\dbds}^c}}_1 \\
\leq 2 \norm{\lrrb{\dbdh-\dbds}_{\Supp_{\dbds}}}_1
\leq 2 \sqrt{\sparsity_{\dbds}} \norm{\lrrb{\dbdh-\dbds}_{\Supp_{\dbds}}}_2
\leq 2 \sqrt{\sparsity_{\dbds}} \norm{\dbdh-\dbds}_2.
\end{multline}

Now, on one hand,
\begin{equation} \label{eq:Dantzig:dbdhdiff:2}
\lrrb{\dbdh-\dbds}^\top \hat{\Hessian} \lrrb{\dbdh-\dbds}
\leq \infnorm{\hat{\Hessian} \lrrb{\dbdh-\dbds}} \onenorm{\dbdh-\dbds}
\leq 4 \sqrt{\sparsity_{\dbds}} \regH \lrnorm{\dbdh-\dbds}_2,
\end{equation}
where we have used \cref{eq:Dantzig:graddiff,eq:Dantzig:dbdhdiff:1}. On the other hand, since $\dbdh-\dbds \in \Cone(\sparsity_{\dbds}, 1)$,
\begin{equation} \label{eq:Dantzig:dbdhdiff:3}
\lrrb{\dbdh-\dbds}^\top \hat{\Hessian} \lrrb{\dbdh-\dbds} \geq \frac{\mineigen_1 \mineigen_2}{2} \lrnorm{\dbdh-\dbds}_2^2.
\end{equation}
Combining \cref{eq:Dantzig:dbdhdiff:2,eq:Dantzig:dbdhdiff:3},
\begin{equation*}
\norm{\dbdh-\dbds}_2 \leq \frac{8 \sqrt{\sparsity_{\dbds}} \regH}{\mineigen_1 \mineigen_2}.
\end{equation*}
Again using \cref{eq:Dantzig:dbdhdiff:1},
\begin{equation*}
\norm{\dbdh-\dbds}_1 \leq \frac{16 \sparsity_{\dbds} \regH}{\mineigen_1 \mineigen_2}.
\end{equation*}
\end{proof}

\subsection{Proof of \cref{lem:BE:Dtrace:abbr}} \label{app:proof:BE:Dtrace:abbr}

\cref{lem:BE:Dtrace:abbr} is obtained as a corollary of \cref{lem:BE:Dtrace}.

\begin{lemma}[General Berry-Esseen inequality for a de-biased D-trace estimator] \label{lem:BE:Dtrace}
Let $\mineigen_1$ and $\mineigen_2$ be the minimum eigenvalues of $\Sxs$ and $\Sys$, respectively. Let $\sparsity_{\Ds} = \norm{\Ds}_0$ and $\sparsity_{\dbds} = \norm{\dbds}_0$. Suppose
\begin{gather*}
\sup_z \lrabs{\PP\lrcb{\asympSD^{-1} \sqrt{\effN} \lrcb{\mdbdsT_1 \vec\lrrb{\Sxh-\Sxs} + \mdbdsT_2 \vec\lrrb{\Syh-\Sys}} \leq z} - \Phi(z)} \leq \BEprob, \\
\PP\lrcb{\infnorm{\Sxh-\Sxs} > \Sxhrate \text{ or } \infnorm{\Syh-\Sys} > \Syhrate} \leq \LDprob
\end{gather*}
for some $\asympSD^2 > 0$, $\effN = \effN(T_1, T_2) > 0$, $\Sxhrate = \Sxhrate(d, T_1) > 0$, $\Syhrate = \Syhrate(d, T_2) > 0$, $\BEprob = \BEprob(d, T_1, T_2) \in (0, 1)$, $\LDprob = \LDprob(d, T_1, T_2) \in (0, 1)$ such that
\begin{equation} \label{cond:SC:RSC}
\Hessrate = \infnorm{\Sys} \Sxhrate + \infnorm{\Sxs} \Syhrate + \Sxhrate \Syhrate \leq \frac{\mineigen_1 \mineigen_2}{\max\cb{32 \sparsity_{\Ds}, 8 \sparsity_{\dbds}}}.
\end{equation}

If the regularization parameters satisfy
\begin{equation*}
\regD \geq 2 \lrrb{\onenorm{\Ds} \Hessrate + \Sxhrate + \Syhrate}, \quad \regH \geq \onenorm{\dbds} \Hessrate,
\end{equation*}
then
\begin{equation*}
\sup_z \lrabs{\PP\lrcb{\asympSD^{-1} \sqrt{\effN} \lrrb{\Dt_0-\Ds_0} \leq z} - \Phi(z)} \lesssim \BEprob + \LDprob + \LDrate
\end{equation*}
with
\begin{equation*}
\LDrate = \frac{\sqrt{\effN}}{\asympSD} \lrrb{\onenorm{\Ds} \onenorm{\dbds} \Sxhrate \Syhrate + \frac{\lrrb{\sparsity_{\Ds} + \sparsity_{\dbds}} \regD \regH}{\mineigen_1 \mineigen_2}} + \frac{\regH}{1-\regH}.
\end{equation*}
\end{lemma}

\begin{proof}
By the decomposition in \cref{eq:decomposition},
\begin{equation*}
\asympSD^{-1} \sqrt{\effN} \lrrb{\Dt_0-\Ds_0} = \frac{\asympSD^{-1} \sqrt{\effN} \lrrb{U + V}}{1+W},
\end{equation*}
where
\begin{gather*}
U = \mdbdsT_1 \vec\lrrb{\Sxh-\Sxs} + \mdbdsT_2 \vec\lrrb{\Syh-\Sys}, \quad V = V_1 + V_2 + V_3, \quad W = \lrsqb{\hat{\Hessian} \dbdh}_0-1,
\end{gather*}
with
\begin{gather*}
V_1 = -\dbdsT \frac 12 \lrcb{\lrrb{\Sxh-\Sxs} \otimes \lrrb{\Syh-\Sys} + \lrrb{\Syh-\Sys} \otimes \lrrb{\Sxh-\Sxs}} \vec(\Ds), \\
V_2 = -\lrsqb{\hat{\Hessian} \dbdh}_u^\top \lrrb{\Dh_u-\Ds_u}, \quad V_3 = \lrrb{\dbdh-\dbds}^\top \gradD\lrrb{\Ds; \Sxh, \Syh}.
\end{gather*}

Consider the event
\begin{equation*}
\infnorm{\Sxh-\Sxs} \leq \Sxhrate, \quad \infnorm{\Syh-\Sys} \leq \Syhrate.
\end{equation*}

On this event, all of the following statements hold:
\begin{enumerate}
\item By \cref{eq:Hessdiff},
\begin{equation*}
\infnorm{\hat{\Hessian}-\Hessian^*} \leq \infnorm{\Sys} \Sxhrate + \infnorm{\Sxs} \Syhrate + \Sxhrate \Syhrate = \Hessrate.
\end{equation*}

\item By \cref{eq:graddiff},
\begin{equation} \label{eq:infnorm:gradD}
\infnorm{\gradD\rb{\Ds; \Sxh, \Syh}} \leq \onenorm{\Ds} \Hessrate + \Sxhrate + \Syhrate = \frac{\regD}{2}.
\end{equation}

\item By the definition of $\dbds$,
\begin{equation} \label{eq:infnorm:gradH}
\infnorm{\hat\Hessian \dbds - \basis_0} = \infnorm{\hat{\Hessian} \dbds - \Hessian^* \dbds} \leq \onenorm{\dbds} \Hessrate = \regH.
\end{equation}

\item Under the condition of \cref{cond:SC:RSC}, by \cref{lem:RSC},
\begin{equation*}
\vecv^\top \hat{\Hessian} \vecv \geq \frac{\mineigen_1 \mineigen_2}{2} \norm{\vecv}_2^2,
\end{equation*}
for any $\vecv \in \Cone\rb{\sparsity_{\Ds}, 3} \cup \Cone\rb{\sparsity_{\dbds}, 1}$, where $\Cone\rb{\sparsity, \Conewidth}$ is the set defined in \cref{eq:cone}.
\end{enumerate}

These, in turn, imply the following statements:
\begin{enumerate}
\setcounter{enumi}{4}
\item By \cref{lem:Dh},
\begin{equation} \label{eq:onenorm:Ddiff}
\onenorm{\Dh-\Ds} \lesssim \frac{\sparsity_{\Ds} \regD}{\mineigen_1 \mineigen_2}.
\end{equation}

\item By \cref{lem:dbdh},
\begin{equation} \label{eq:onenorm:vdiff}
\onenorm{\dbdh-\dbds} \lesssim \frac{\sparsity_{\dbds} \regH}{\mineigen_1 \mineigen_2}.
\end{equation}
\end{enumerate}

Combining the implications of \cref{eq:infnorm:gradD,eq:infnorm:gradH,eq:onenorm:Ddiff,eq:onenorm:vdiff} via H\"{o}lder's inequality,
\begin{gather*}
\abs{V_1} \leq \onenorm{\Ds} \onenorm{\dbds} \Sxhrate \Syhrate, \\
\abs{V_2} \leq \infnorm{\lrsqb{\hat{\Hessian} \dbdh}_u} \onenorm{\Dh_u-\Ds_u} \leq \infnorm{\hat{\Hessian} \dbdh - \basis_0} \onenorm{\Dh - \Ds} \lesssim \frac{\sparsity_{\Ds} \regD \regH}{\mineigen_1 \mineigen_2}, \\
\abs{V_3} \leq \onenorm{\dbdh - \dbds} \infnorm{\gradD\rb{\Ds; \Sxh, \Syh}} \lesssim \frac{\sparsity_{\dbds} \regD \regH}{\mineigen_1 \mineigen_2}, \\
\abs{W} \leq \lrabs{\lrsqb{\hat{\Hessian} \dbdh}_0-1} \leq \infnorm{\hat{\Hessian} \dbdh - \basis_0} \leq \regH, \\
\end{gather*}
and hence,
\begin{equation*}
\abs{V} \lesssim \onenorm{\Ds} \onenorm{\dbds} \Sxhrate \Syhrate + \frac{\sparsity_{\Ds} \regD \regH}{\mineigen_1 \mineigen_2} + \frac{\sparsity_{\dbds} \regD \regH}{\mineigen_1 \mineigen_2}.
\end{equation*}

By \cref{lem:BE},
\begin{equation*}
\sup_z \lrabs{\PP\lrcb{\asympSD^{-1} \sqrt{\effN} \lrrb{\Dt_0-\Ds_0} \leq z} - \Phi(z)} \lesssim \BEprob + \LDprob + \LDrate
\end{equation*}
with
\begin{equation*}
\LDrate = \frac{\sqrt{\effN}}{\asympSD} \lrrb{\onenorm{\Ds} \onenorm{\dbds} \Sxhrate \Syhrate + \frac{\lrrb{\sparsity_{\Ds} + \sparsity_{\dbds}} \regD \regH}{\mineigen_1 \mineigen_2}} + \frac{\regH}{1-\regH}.
\end{equation*}
\end{proof}

\begin{proof}[Proof of \cref{lem:BE:Dtrace:abbr}]
\cref{lem:BE:Dtrace:abbr} is a special case of \cref{lem:BE:Dtrace} with the restriction $\Sxhrate, \Syhrate < 1$ and $\regH \leq 1/2$.

First, note that for any $a, b > 0$,
\begin{equation*}
\rb{a - b}^2 = a^2 + b^2 - 2ab \geq 0 \quad \iff \quad ab \leq \frac{a^2 + b^2}{2},
\end{equation*}
and if, in addition, $a, b < 1$, then
\begin{equation*}
ab \leq \frac{a^2 + b^2}{2} < \frac{a + b}{2}.
\end{equation*}

Utilizing the above inequality, we have
\begin{multline*}
\Hessrate
= \infnorm{\Sys} \Sxhrate + \infnorm{\Sxs} \Syhrate + \Sxhrate \Syhrate \\
\leq \lrrb{\infnorm{\Sys} + \frac 1 2} \Sxhrate + \lrrb{\infnorm{\Sxs} + \frac 1 2} \Syhrate
\lesssim_{\Sxs, \Sys} \Sxhrate + \Syhrate.
\end{multline*}
Therefore, the condition of \cref{cond:SC:RSC} simplifies to
\begin{equation*}
\Sxhrate + \Syhrate \leq \frac{\mineigen_1 \mineigen_2}{\max\lrcb{32 \sparsity_{\Ds}, 8 \sparsity_{\dbds}} \lrrb{\max\lrcb{\infnorm{\Sxs}, \infnorm{\Sys}} + 1/2}},
\end{equation*}
and it is possible to pick $\regD$ and $\regH$ such that
\begin{equation*}
\regD \asymp \Sxhrate + \Syhrate, \quad \regH \asymp \Sxhrate + \Syhrate.
\end{equation*}

We turn to the bound. We have
\begin{gather*}
\onenorm{\Ds} \onenorm{\dbds} \Sxhrate \Syhrate \lesssim \onenorm{\Ds} \onenorm{\dbds} \lrrb{\Sxhrate^2 + \Syhrate^2}, \\
\frac{\lrrb{\sparsity_{\Ds} + \sparsity_{\dbds}} \regD \regH}{\mineigen_1 \mineigen_2} \lesssim_{\Sxs, \Sys} \lrrb{\sparsity_{\Ds} + \sparsity_{\dbds}} \lrrb{\Sxhrate + \Syhrate}^2 \lesssim \lrrb{\sparsity_{\Ds} + \sparsity_{\dbds}} \lrrb{\Sxhrate^2 + \Syhrate^2}.
\end{gather*}
Moreover, since $\regH \in (0, 1/2)$ by the hypothesis of the lemma,
\begin{equation*}
\frac{\regH}{1-\regH} \leq 2 \regH \lesssim \Sxhrate + \Syhrate.
\end{equation*}
Thus,
\begin{equation*}
\LDrate \lesssim_{\Sxs, \Sys} \frac{\sqrt{\effN}}{\asympSD} \lrrb{\onenorm{\Ds} \onenorm{\dbds} + \sparsity_{\Ds} + \sparsity_{\dbds}} \lrrb{\Sxhrate^2 + \Syhrate^2} + \Sxhrate + \Syhrate.
\end{equation*}
\end{proof}

\subsection{Alternative estimating equations and the generalized D-trace loss} \label{app:loss:general}

\cref{eq:pop:Dtrace} is not the only linear estimating equation with $\Ds$ as the solution. Two other alternatives are
\begin{align*}
\lrrb{\Sxs \otimes \Sys} \vec\rb{\D} - \vec\lrrb{\Sys-\Sxs} &= 0, \\
\lrrb{\Sys \otimes \Sxs} \vec\rb{\D} - \vec\lrrb{\Sys-\Sxs} &= 0.
\end{align*}
In fact, any convex combination of \cref{eq:pop:1,eq:pop:0}
\begin{equation} \label{eq:pop:general}
\lrcb{\EEcoef \lrrb{\Sxs \otimes \Sys} + \lrrb{1-\EEcoef} \lrrb{\Sys \otimes \Sxs}} \vec\rb{\D} - \vec\lrrb{\Sys-\Sxs} = 0, \quad \EEcoef \in [0, 1],
\end{equation}
is a valid linear estimating equation for estimating $\Ds$. It is easy to see that \cref{eq:pop:general} is the fully general formulation that includes all of \cref{eq:pop:Dtrace,eq:pop:1,eq:pop:0} as special cases corresponding to the choices of $\EEcoef = 1/2$, $1$, and $0$, respectively. The latter two cases are of particular interest due to the ease of computing the inverse of corresponding Hessians.

Solving \cref{eq:pop:general} is equivalent to minimizing the population version of the following generalized D-trace loss:
\begin{equation*}
\lossDgen\rb{\D; \Sx, \Sy} = \frac 1 2 \vec\rb{\D}^\top \Hessian_\EEcoef \vec\rb{\D} - \vec\rb{\D}^\top \vec\rb{\Sy-\Sx},
\end{equation*}
where
\begin{equation*}
\Hessian_\EEcoef = \Hessian_\EEcoef \rb{\Sx, \Sy} = \EEcoef \lrrb{\Sx \otimes \Sy} + \lrrb{1-\EEcoef} \lrrb{\Sy \otimes \Sx}.
\end{equation*}

We derive the asymptotic linear approximation for the de-biased D-trace estimator based on the general form of the linear estimating equation \cref{eq:pop:general}.

First, we re-examine the deviation in the Hessian $\hat{\Hessian}_\EEcoef-\Hessian^*_\EEcoef$, where $\hat{\Hessian}_\EEcoef = \Hessian_\EEcoef\rb{\Sxh, \Syh}$ and $\Hessian^*_\EEcoef = \Hessian_\EEcoef\rb{\Sxs, \Sys}$. By \cref{eq:SxXSydiff,eq:SyXSxdiff},
\begin{align}
\hat{\Hessian}_\EEcoef-\Hessian^*_\EEcoef
&= \EEcoef \lrrb{\Sxh-\Sxs} \otimes \Sys + \rb{1-\EEcoef} \Sys \otimes \lrrb{\Sxh-\Sxs} \nonumber\\
& \quad + \rb{1-\EEcoef} \lrrb{\Syh-\Sys} \otimes \Sxs + \EEcoef \Sxs \otimes \lrrb{\Syh-\Sys} \nonumber\\
& \quad + \EEcoef \lrrb{\Sxh-\Sxs} \otimes \lrrb{\Syh-\Sys} + \rb{1-\EEcoef} \lrrb{\Syh-\Sys} \otimes \lrrb{\Sxh-\Sxs}. \label{eq:Hessdiff:general}
\end{align}

As for the deviation in the gradient,
\begin{align}
\gradDgen\rb{\Ds; \Sxh, \Syh}
&= \gradDgen\rb{\Ds; \Sxh, \Syh} - \gradDgen\rb{\Ds; \Sxs, \Sys} \nonumber\\
&= \lrrb{\hat{\Hessian}_\EEcoef - \Hessian^*_\EEcoef} \vec\rb{\Ds} - \vec\rb{\Syh-\Sys} + \vec\rb{\Sxh-\Sxs}. \label{eq:graddiff:general}
\end{align}
This now yields
\begin{align*}
\gradDgen\rb{\Ds; \Sxh, \Syh}
= \vec\Big\{ & \EEcoef \Sys \Sxsinv \lrrb{\Sxh-\Sxs} + \lrrb{1-\EEcoef} \lrrb{\Sxh-\Sxs} \Sxsinv \Sys \\
& - \lrrb{1-\EEcoef} \Sxs \Sysinv \lrrb{\Syh-\Sys} - \EEcoef \lrrb{\Syh-\Sys} \Sysinv \Sxs \\
& + \EEcoef \lrrb{\Syh-\Sys} \Ds \lrrb{\Sxh-\Sxs} + \lrrb{1-\EEcoef} \lrrb{\Sxh-\Sxs} \Ds \lrrb{\Syh-\Sys} \Big\}.
\end{align*}
Using the so-called ``vec trick'', this can be written as
\begin{align}
\gradDgen\rb{\Ds; \Sxh, \Syh}
&= \lrcb{\EEcoef \lrrb{\Identity_{p'} \otimes \Sys \Sxsinv} + \lrrb{1-\EEcoef} \lrrb{\Sys \Sxsinv \otimes \Identity_{p'}}} \vec\lrrb{\Sxh-\Sxs} \nonumber\\
& \quad - \lrcb{\lrrb{1-\EEcoef} \lrrb{\Identity_{p'} \otimes \Sxs \Sysinv} + \EEcoef \lrrb{\Sxs \Sysinv \otimes \Identity_{p'}}} \vec\lrrb{\Syh-\Sys} \nonumber\\
& \quad + \lrcb{\EEcoef \lrrb{\Sxh-\Sxs} \otimes \lrrb{\Syh-\Sys}  + \lrrb{1-\EEcoef} \lrrb{\Syh-\Sys} \otimes \lrrb{\Sxh-\Sxs}} \vec\lrrb{\Ds}. \label{eq:gradDgen}
\end{align}

The projection direction $\dbdh_\EEcoef$ is now picked to approximate $\dbds_\EEcoef$, where $\dbds_\EEcoef$ satisfies $\Hessian^*_\EEcoef \dbds_\EEcoef = \basis_0$. Thus, the sample version of the projected estimating equation becomes
\begin{equation} \label{eq:ee:general}
\dbdh_\EEcoef^\top \gradDgen\lrrb{\rb{\Dtgen{0}, \Dh_u}; \Sxh, \Syh} = \dbdh_\EEcoef^\top \gradDgen\rb{\Ds; \Sxh, \Syh} + \dbdh_\EEcoef^\top \hat{\Hessian}_\EEcoef \begin{bmatrix} \Dtgen{0}-\Ds_0 \\ \Dh_u-\Ds_u \end{bmatrix} = 0.
\end{equation}

Applying the analysis of \cref{subsec:infogeom} to \cref{eq:ee:general} yields
\begin{align*}
\Dtgen{0}-\Ds_0
&= \lrrb{\Dh_0-\lrsqb{\hat{\Hessian}_\EEcoef \dbdh_\EEcoef}_0^{-1} \dbdh_\EEcoef^\top \gradDgen\lrrb{\Ds; \Sxh, \Syh}}-\Ds_0 \\
&=
\begin{multlined}[t]
\biggcb{1+\lrrb{\lrsqb{\hat{\Hessian}_\EEcoef \dbdh_\EEcoef}_0-1}}^{-1} \bigg\{-\dbdsT_\EEcoef \gradDgen\lrrb{\Ds; \Sxh, \Syh} \\
- \lrsqb{\hat{\Hessian}_\EEcoef \dbdh_\EEcoef}_u^\top \lrrb{\Dh_u-\Ds_u} - \lrrb{\dbdh_\EEcoef-\dbds_\EEcoef}^\top \gradDgen\lrrb{\Ds; \Sxh, \Syh} \bigg\}.
\end{multlined}
\end{align*}
Plugging in \cref{eq:gradDgen} for $\gradDgen\lrrb{\Ds; \Sxh, \Syh}$ in the leading term, we have
\begin{multline*}
\Dtgen{0}-\Ds_0
= \biggcb{1+\lrrb{\lrsqb{\hat{\Hessian}_\EEcoef \dbdh_\EEcoef}_0-1}}^{-1} \bigg\{\mdbdsT_{\EEcoef, 1} \vec\lrrb{\Sxh-\Sxs} + \mdbdsT_{\EEcoef, 2} \vec\lrrb{\Syh-\Sys} \\
- \dbdsT_\EEcoef \lrcb{\EEcoef \lrrb{\Sxh-\Sxs} \otimes \lrrb{\Syh-\Sys}  + \lrrb{1-\EEcoef} \lrrb{\Syh-\Sys} \otimes \lrrb{\Sxh-\Sxs}} \vec\lrrb{\Ds} \\
- \lrsqb{\hat{\Hessian}_\EEcoef \dbdh_\EEcoef}_u^\top \lrrb{\Dh_u-\Ds_u} - \lrrb{\dbdh_\EEcoef-\dbds_\EEcoef}^\top \gradDgen\lrrb{\Ds; \Sxh, \Syh} \bigg\},
\end{multline*}
where
\begin{equation} \label{eq:mdbds:general}
\begin{gathered}
\mdbds_{\EEcoef, 1} = -\lrcb{\EEcoef \lrrb{\Identity_{p'} \otimes \Sxsinv \Sys} + \lrrb{1-\EEcoef} \lrrb{\Sxsinv \Sys \otimes \Identity_{p'}}} \dbds_\EEcoef, \\
\mdbds_{\EEcoef, 2} = \lrcb{\lrrb{1-\EEcoef} \lrrb{\Identity_{p'} \otimes \Sysinv \Sxs} + \EEcoef \lrrb{\Sysinv \Sxs \otimes \Identity_{p'}}} \dbds_\EEcoef.
\end{gathered}
\end{equation}

We note the cases of particular interest.
\begin{enumerate}
\item When $\EEcoef = 1$, we have $\Hessian^*_{\EEcoef = 1} = \Sxs \otimes \Sys$, $\dbds_{\EEcoef = 1} = \rb{\Sxsinv \otimes \Sysinv} \basis_0$, and hence,
\begin{gather*}
\mdbds_{\EEcoef = 1, 1} = -\lrrb{\Identity_{p'} \otimes \Sxsinv \Sys} \lrrb{\Sxsinv \otimes \Sysinv} \basis_0 = -\lrrb{\Sxsinv \otimes \Sxsinv} \basis_0, \\
\mdbds_{\EEcoef = 1, 2} = \lrrb{\Sysinv \Sxs \otimes \Identity_{p'}} \rb{\Sxsinv \otimes \Sysinv} \basis_0 = \lrrb{\Sysinv \otimes \Sysinv} \basis_0.
\end{gather*}
\item When $\EEcoef = 0$, we have $\Hessian^*_{\EEcoef = 0} = \Sys \otimes \Sxs$, $\dbds_{\EEcoef = 0} = \rb{\Sysinv \otimes \Sxsinv} \basis_0$, and hence,
\begin{gather*}
\mdbds_{\EEcoef = 0, 1} = -\lrrb{\Sxsinv \Sys \otimes \Identity_{p'}} \lrrb{\Sysinv \otimes \Sxsinv} \basis_0 = -\lrrb{\Sxsinv \otimes \Sxsinv} \basis_0, \\
\mdbds_{\EEcoef = 0, 2} = \lrrb{\Identity_{p'} \otimes \Sysinv \Sxs} \rb{\Sysinv \otimes \Sxsinv} \basis_0 = \lrrb{\Sysinv \otimes \Sysinv} \basis_0.
\end{gather*}
\end{enumerate}

We can also easily extend \cref{lem:BE:Dtrace:abbr} to the general de-biased D-trace estimator $\Dtgen{0}$.

\begin{lemma}[Berry-Esseen inequality for the general de-biased D-trace estimator] \label{lem:BE:Dtrace:abbr:general}
Let $\mineigen_1$ and $\mineigen_2$ be the minimum eigenvalues of $\Sxs$ and $\Sys$, respectively. Let $\sparsity_{\Ds} = \norm{\Ds}_0$ and $\sparsity_{\dbds_\EEcoef} = \norm{\dbds_\EEcoef}_0$. Suppose
\begin{gather*}
\sup_z \lrabs{\PP\lrcb{\asympSD^{-1} \sqrt{\effN} \lrcb{\mdbdsT_{\EEcoef, 1} \vec\lrrb{\Sxh-\Sxs} + \mdbdsT_{\EEcoef, 2} \vec\lrrb{\Syh-\Sys}} \leq z} - \Phi(z)} \leq \BEprob, \\
\PP\lrcb{\infnorm{\Sxh-\Sxs} > \Sxhrate \text{ or } \infnorm{\Syh-\Sys} > \Syhrate} \leq \LDprob
\end{gather*}
for some $\asympSD^2 > 0$, $\effN = \effN(T_1, T_2) > 0$, $\Sxhrate = \Sxhrate(d, T_1) \in (0, 1)$, $\Syhrate = \Syhrate(d, T_2) \in (0, 1)$, $\BEprob = \BEprob(d, T_1, T_2) \in (0, 1)$, $\LDprob = \LDprob(d, T_1, T_2) \in (0, 1)$ such that
\begin{equation*}
\Sxhrate + \Syhrate \leq \frac{\mineigen_1 \mineigen_2}{\max\lrcb{32 \sparsity_{\Ds}, 8 \sparsity_{\dbds_\EEcoef}} \lrrb{\max\lrcb{\infnorm{\Sxs}, \infnorm{\Sys}} + 1/2}}.
\end{equation*}

If the regularization parameters satisfy
\begin{equation*}
\regD \geq 2 \lrrb{\onenorm{\Ds} \Hessrate + \Sxhrate + \Syhrate}, \quad \onenorm{\dbds_\EEcoef} \Hessrate \leq \regH \leq 1/2,
\end{equation*}
where
\begin{equation*}
\Hessrate = \infnorm{\Sys} \Sxhrate + \infnorm{\Sxs} \Syhrate + \Sxhrate \Syhrate,
\end{equation*}
then
\begin{equation*}
\sup_z \lrabs{\PP\lrcb{\asympSD^{-1} \sqrt{\effN} \lrrb{\Dtgen{0}-\Ds_0} \leq z} - \Phi(z)} \leq \BEprob + \LDprob + \LDrate + 2 \regH
\end{equation*}
with
\begin{equation*}
\LDrate \lesssim_{\Sxs, \Sys} \frac{\sqrt{\effN}}{\asympSD} \lrrb{\onenorm{\Ds} \onenorm{\dbds_\EEcoef} + \sparsity_{\Ds} + \sparsity_{\dbds_\EEcoef}} \lrrb{\Sxhrate^2 + \Syhrate^2}.
\end{equation*}
\end{lemma}

The proof is nearly identical to that of \cref{lem:BE:Dtrace:abbr}, and is hence omitted. In particular, $\norm{\hat{\Hessian}_\EEcoef-\Hessian^*_\EEcoef}_\infty$ obeys the same bound as $\norm{\hat{\Hessian}-\Hessian^*}_\infty$, and hence all the other bounds, which are downstream of the deviation in the Hessian, remain identical to the $\EEcoef = 1/2$ case. It is also not difficult to see that results that are analogous to \cref{thm:BE:projWelch,thm:BE:SDD} established for the $\EEcoef = 1/2$ case hold in general with $\mdbds_{\EEcoef, 1}$ and $\mdbds_{\EEcoef, 2}$ replacing $\mdbds_1$ and $\mdbds_2$, respectively.

Interestingly, although the choice of $\EEcoef$ has almost no effect on the theory, it can have significant implications for both the computational cost and practical performance of our procedure. For $\EEcoef = 1$ or $0$, $\dbds_\EEcoef$ can be estimated by first estimating $\Sxsinv$ and $\Sysinv$ and then taking their Kronecker product. By contrast, for $\EEcoef \in (0, 1)$, no such shortcut exists. A subtler point is that $\sparsity_{\dbds_\EEcoef}$ or $\norm{\mdbds_{\EEcoef, 1}}_1$ and $\norm{\mdbds_{\EEcoef, 2}}_1$ can be smaller for $\EEcoef = 1$ or $0$ compared to other values of $\EEcoef$.

\section{Proof of \cref{thm:BE:projWelch}} \label{app:proof:BE:projWelch}
In this section, we prove \cref{thm:BE:projWelch}, which is a crucial stepping stone to our main result. We first provide a more detailed explanation of the intuition underlying the result, and then present the proof. To avoid clutter, we present the exposition in the context of a generic single process $X_\cdot$ before returning to our two sample setting for the formal proof.

Let $X_\cdot$ be a process with a finite dependence adjusted Orlicz-type norm for some $\alpha \geq 0$:
\begin{equation*}
\norm{X_\cdot}_{\daOpsi_\daOalpha} = \max_j \norm{X_{\cdot, j}}_{\daOpsi_\daOalpha} < \infty.
\end{equation*}
We have seen that this effectively imposes a decay condition on long lag functional dependence measures, which implies that $X_t$ may be accurately approximated by its $m$-dependent projection $X_{m, t} = \EE\sqb{X_t \mid \epsilon_t, \epsilon_{t-1}, \ldots \epsilon_{t-m}}$ for some $m \geq 0$. This is useful, because we know exactly how many observations we need to remove from each blocks to make the block sums independent for $m$-dependent processes.

Similarly to \citet{zhang2025spectral}, we can think of $\Sh-\Ss$ as
\begin{multline} \label{eq:projWelch:decompose}
\Sh-\Ss = \lrrb{\hat{\matS}_m-\EE \hat{\matS}_m} \\
+ \lrcb{\lrrb{\Sh-\EE \Sh} - \lrrb{\Sh_m-\EE \Sh_m}} + \lrcb{\lrrb{\Sh_m-\EE \Sh_m}-\lrrb{\hat{\matS}_m-\EE \hat{\matS}_m}} \\
+ \lrrb{\EE \Sh-\Ss},
\end{multline}
where for $\freq \notin \cb{0, \pi}$,
\begin{align}
\Sh_m &= \frac{1}{2 \pi (T/B)} \sum_{n = 1}^{T/B} \frac{1}{B} \sum_{s = (n-1) B+1}^{nB} \sum_{t = (n-1) B+1}^{nB} \Rotation_{-\freq (s-t)} \otimes X_{m, s} X_{m, t}^\top, \label{eq:Sighm} \\
\hat{\matS}_m
&= \frac{1}{2 \pi (T/B)} \sum_{n = 1}^{T/B} \frac{1}{B} \sum_{s = (n-1) B+1}^{nB-m} \sum_{t = (n-1) B+1}^{nB-m} \Rotation_{-\freq (s-t)} \otimes X_{m, s} X_{m, t}^\top, \label{eq:Shatm}
\end{align}
and $\Rotation_\freq$ is the rotation matrix
\begin{equation*}
\Rotation_\freq = \begin{bmatrix} \cos (\freq) & -\sin (\freq) \\ \sin (\freq) & \phantom{-}\cos (\freq) \end{bmatrix};
\end{equation*}
and for $\freq \in \cb{0, \pi}$,
\begin{align}
\Sh_m &= \frac{1}{2 \pi (T/B)} \sum_{n = 1}^{T/B} \frac{1}{B} \sum_{s = (n-1) B+1}^{nB} \sum_{t = (n-1) B+1}^{nB} e^{-\imath \freq (s-t)} X_{m, s} X_{m, t}^\top, \label{eq:Sighm:real} \\
\hat{\matS}_m
&= \frac{1}{2 \pi (T/B)} \sum_{n = 1}^{T/B} \frac{1}{B} \sum_{s = (n-1) B+1}^{nB-m} \sum_{t = (n-1) B+1}^{nB-m} e^{-\imath \freq (s-t)} X_{m, s} X_{m, t}^\top. \label{eq:Shatm:real}
\end{align}

Provided that the long lag functional dependencies in $X_\cdot$ decay sufficiently fast---which would imply that $m$ can be kept small relative to $B$---we can expect $\Sh \approx \Sh_m \approx \hat{\matS}_m$:
\begin{itemize}
\item $\Sh \approx \Sh_m$, because $X_t \approx X_{t, m}$,
\item $\Sh_m \approx \hat{\matS}_m$, because $\hat{\matS}_m$ differ from $\Sh_m$ only by the last $m$ observations in each block.
\end{itemize}
At the same time, $\hat{\matS}_m$ is clearly an average of $T/B$ i.i.d.~summands, and hence, any projection $\vecv^\top \vec\rb{\hat{\matS}_m-\EE \hat{\matS}_m}$ is an average of $T/B$ i.i.d.~univariate random variables with zero mean:
\begin{equation*}
\vecw^\top \vec\lrrb{\hat{\matS}_m-\EE \hat{\matS}_m} = \frac{1}{T/B} \sum_{n = 1}^{T/B} \vecw^\top \lrrb{Y_{m, n}-\EE Y_{m, n}},
\end{equation*}
where for $\freq \notin \cb{0, \pi}$,
\begin{equation} \label{eq:Ym}
Y_{m, n} = \frac{1}{2 \pi B} \sum_{s = (n-1) B+1}^{nB-m} \sum_{t = (n-1) B+1}^{nB-m} \vec\lrrb{\Rotation_{-\freq (s-t)} \otimes X_{m, s} X_{m, t}^\top};
\end{equation}
and for $\freq \in \cb{0, \pi}$,
\begin{equation} \label{eq:Ym:real}
Y_{m, n} = \frac{1}{2 \pi B} \sum_{s = (n-1) B+1}^{nB-m} \sum_{t = (n-1) B+1}^{nB-m} \vec\lrrb{e^{-\imath \freq (s-t)} X_{m, s} X_{m, t}^\top}.
\end{equation}
Thus, we can hope to approximate the distribution of $\vecv^\top \vec\rb{\hat{\matS}_m-\EE \hat{\matS}_m}$ with a Normal, and by extension, that of $\vecv^\top \vec\rb{\Sh-\Ss}$.

This is the idea behind the proof strategy of \cref{thm:asympnormal} for a generic single process $X_\cdot$. Returning to our two sample setting, applying the decomposition in \cref{eq:projWelch:decompose} to each of $X_{1, \cdot}$ and $X_{2, \cdot}$, we see that the leading term is
\begin{equation*}
\mdbdsT_1 \vec\lrrb{\hat{\matS}_{1, m_1}-\EE \hat{\matS}_{1, m_1}} + \mdbdsT_2 \vec\lrrb{\hat{\matS}_{2, m_2}-\EE \hat{\matS}_{2, m_2}},
\end{equation*}
where $\hat{\matS}_{1, m_1}$ and $\hat{\matS}_{2, m_2}$ are as in \cref{eq:Shatm} or \cref{eq:Shatm:real} but for $X_{1, \cdot}$ and $m_1 \geq 0$ and $X_{2, \cdot}$ and $m_2 \geq 0$, respectively. We first apply the standard Berry-Esseen inequality to the leading term and then extended it to the original sum. This is the gist of the proof.

\begin{proof}[Proof of \cref{thm:BE:projWelch}]
Put $\relN_1 = \effN/\effN_1$ and $\relN_2 = \effN/\effN_2$. Let $\asympSD^2_{1, m_1}$ and $\asympSD^2_{2, m_2}$ be the variances of the random variables $\mdbdsT_1 \rb{Y_{1, m_1, 1}-\EE Y_{1, m_1, 1}}$ and $\mdbdsT_2 \rb{Y_{2, m_2, 1}-\EE Y_{2, m_2, 1}}$, respectively, where $Y_{1, m_1, 1}$ and $Y_{2, m_2, 1}$ are as in \cref{eq:Ym} or \cref{eq:Ym:real}. Let $\asympSD^2 = \relN_1 \asympSD^2_1 + \relN_2 \asympSD^2_2$ and $\asympSD^2_{m_1, m_2} = \relN_1 \asympSD^2_{1, m_1} + \relN_2 \asympSD^2_{2, m_2}$. Then, $\rb{\asympSD^2_1/\effN_1} + \rb{\asympSD^2_2/\effN_2} = \asympSD^2/\effN$ and
\begin{equation*}
\Var\lrcb{\mdbdsT_1 \vec\lrrb{\hat{\matS}_{1, m_1}-\EE \hat{\matS}_{1, m_1}} + \mdbdsT_2 \vec\lrrb{\hat{\matS}_{2, m_2}-\EE \hat{\matS}_{2, m_2}}} = \frac{\asympSD^2_{1, m_1}}{N_1} + \frac{\asympSD^2_{2, m_2}}{N_2} = \frac{\asympSD^2_{m_1, m_2}}{\effN}.
\end{equation*}

In light of the decomposition in \cref{eq:projWelch:decompose},
\begin{equation*}
\frac{\sqrt{\effN}}{\asympSD} \lrcb{\mdbdsT_1 \vec\lrrb{\Sxh-\Sxs} + \mdbdsT_2 \vec\lrrb{\Syh-\Sys}} = \frac{U + V}{1 + W},
\end{equation*}
where
\begin{gather*}
U = \frac{\sqrt{\effN}}{\asympSD_{m_1, m_2}} \lrcb{\mdbdsT_1 \vec\lrrb{\hat{\matS}_{1, m_1}-\EE \hat{\matS}_{1, m_1}} + \mdbdsT_2 \vec\lrrb{\hat{\matS}_{2, m_2}-\EE \hat{\matS}_{2, m_2}}},\\
V = \frac{\sqrt{\effN}}{\asympSD_{m_1, m_2}} \lrsqb{\mdbdsT_1 \vec\lrcb{\lrrb{\Sxh-\Sxs} - \lrrb{\hat{\matS}_{1, m_1}-\EE \hat{\matS}_{1, m_1}}} + \mdbdsT_2 \vec\lrcb{\lrrb{\Syh-\Sys} - \lrrb{\hat{\matS}_{2, m_2}-\EE \hat{\matS}_{2, m_2}}}}, \\
W = \frac{\asympSD_{m_1, m_2}}{\asympSD} - 1.
\end{gather*}
Clearly, the term $U$ admits a Berry-Esseen bound. This bound can then be extended to $\rb{U + V} / \rb{1 + W}$ with additional error terms that reflect the contributions of the terms $V$ and $W$.

We can further decompose $V$ and $W$ according to the error source. In the case of $V$,
\begin{equation*}
V = V_1 + V_2,
\end{equation*}
where
\begin{equation*}
V_1 = V_{11} + V_{12}, \quad V_2 = V_{21} + V_{22}
\end{equation*}
with
\begin{gather*}
V_{11} = \frac{\sqrt{\effN}}{\asympSD_{m_1, m_2}} \mdbdsT_1 \vec\lrcb{\lrrb{\Sxh-\EE \Sxh}-\lrrb{\hat{\matS}_{1, m_1}-\EE \hat{\matS}_{1, m_1}}}, \quad V_{12} = \frac{\sqrt{\effN}}{\asympSD_{m_1, m_2}} \mdbdsT_1 \vec\lrrb{\EE \Sxh-\Sxs}, \\
V_{21} = \frac{\sqrt{\effN}}{\asympSD_{m_1, m_2}} \mdbdsT_2 \vec\lrcb{\lrrb{\Syh-\EE \Syh}-\lrrb{\hat{\matS}_{2, m_2}-\EE \hat{\matS}_{2, m_2}}}, \quad V_{22} = \frac{\sqrt{\effN}}{\asympSD_{m_1, m_2}} \mdbdsT_2 \vec\lrrb{\EE \Syh-\Sys}.
\end{gather*}
Each term in this decomposition can be interpreted as follows:
\begin{itemize}
\item The terms $V_{11}$ and $V_{21}$ reflect the stochastic approximation error from replacing $\Sxh-\EE \Sxh$ and $\Syh-\EE \Syh$ with $\hat{\matS}_{1, m_1}-\EE \hat{\matS}_{1, m_1}$ and $\hat{\matS}_{2, m_2}-\EE \hat{\matS}_{2, m_2}$, respectively.
\item The terms $V_{12}$ and $V_{22}$ reflect the deterministic truncation biases in $\Sxh$ and $\Syh$.
\end{itemize}
As for $W$,
\begin{equation*}
\asympSD_{m_1, m_2}-\asympSD
= \frac{\asympSD^2_{m_1, m_2}-\asympSD^2}{\asympSD_{m_1, m_2}+\asympSD}
= \frac{\relN_1 \lrrb{\asympSD^2_{1, m_1}-\asympSD^2_1} + \relN_2 \lrrb{\asympSD^2_{2, m_2}-\asympSD^2_2}}{\asympSD_{m_1, m_2}+\asympSD},
\end{equation*}
and hence,
\begin{equation} \label{eq:BE:W}
W = \frac{\asympSD_{m_1, m_2}}{\asympSD} - 1 = \frac{\asympSD_{m_1, m_2} - \asympSD}{\asympSD} = \frac{\relN_1 \lrrb{\asympSD^2_{1, m_1}-\asympSD^2_1} + \relN_2 \lrrb{\asympSD^2_{2, m_2}-\asympSD^2_2}}{\asympSD \asympSD_{m_1, m_2} + \asympSD^2}.
\end{equation}
Put
\begin{equation*}
W_1 = \asympSD^2_{1, m_1}-\asympSD^2_1, \quad W_2 = \asympSD^2_{2, m_2}-\asympSD^2_2.
\end{equation*}
The terms $W_1$ and $W_2$ are the deterministic approximation biases resulting from replacing $\Sxh$ and $\Syh$ with $\hat{\matS}_{1, m_1}$ and $\hat{\matS}_{2, m_2}$ in the asymptotic variance.

We shall see that:
\begin{enumerate}
\item By the standard Berry-Esseen inequality \citep{petrov1975sums},
\begin{equation} \label{eq:BE:U}
\sup_z \lrabs{\PP\lrcb{U \leq z}-\Phi(z)} \lesssim \frac{1}{\sqrt{\effN}} \lrrb{\frac{\asympSD^2_{1, m_1}}{\relN_1} + \frac{\asympSD^2_{2, m_2}}{\relN_2}}^{-1/2} \max\lrcb{\frac{\skewness_{1, m_1}}{\asympSD^2_{1, m_1}}, \frac{\skewness_{2, m_2}}{\asympSD^2_{2, m_2}}}
\end{equation}
with
\begin{equation*}
\skewness_{1, m_1} = \EE\lrsqb{\lrabs{\mdbdsT_1 \lrrb{Y_{1, m_1, 1}-\EE Y_{1, m_1, 1}}}^3}, \quad \skewness_{2, m_2} = \EE\lrsqb{\lrabs{\mdbdsT_2 \lrrb{Y_{2, m_2, 1}-\EE Y_{2, m_2, 1}}}^3}.
\end{equation*}

\item By Lemmas 12 and 13 of \citet{zhang2025spectral},
\begin{equation} \label{eq:BE:V1}
\begin{gathered}
\PP\lrcb{\abs{V_{11}} > \delta_{V, 11}} \leq \varepsilon_{V, 11}, \quad \delta_{V, 11} \lesssim_{\damrho_1, \daOalpha_1} \frac{\norm{\mdbds_1}_1 \norm{X_{1, \cdot}}_{\daOpsi_{\daOalpha_1}}^2}{\asympSD_{m_1, m_2}} \sqrt{\frac{\log B_1}{B_1}} \rb{\log d}^{1 + 2 \daOalpha_1}, \quad \varepsilon_{V, 11} \lesssim d^{-1}, \\
\PP\lrcb{\abs{V_{21}} > \delta_{V, 21}} \leq \varepsilon_{V, 21}, \quad \delta_{V, 21} \lesssim_{\damrho_2, \daOalpha_2}  \frac{\norm{\mdbds_2}_1 \norm{X_{2, \cdot}}_{\daOpsi_{\daOalpha_2}}^2}{\asympSD_{m_1, m_2}} \sqrt{\frac{\log B_2}{B_2}} \rb{\log d}^{1 + 2 \daOalpha_2}, \quad \varepsilon_{V, 21} \lesssim d^{-1}.
\end{gathered}
\end{equation}

\item By Proposition 2 of \citet{zhang2025spectral},
\begin{equation} \label{eq:BE:V2}
\abs{V_{12}} \lesssim_{\damrho_1} \frac{\norm{\mdbds_1}_1 \dammax_{1, 2}^2}{\asympSD_{m_1, m_2}} \sqrt{\frac{T_1}{B_1^3}}, \quad
\abs{V_{22}} \lesssim_{\damrho_2} \frac{\norm{\mdbds_2}_1 \dammax_{2, 2}^2}{\asympSD_{m_1, m_2}} \sqrt{\frac{T_2}{B_2^3}}.
\end{equation}

\item By Lemma 11 of \citet{zhang2025spectral} and \cref{prop:asympcov} or \ref{prop:asympcov:real},
\begin{equation} \label{eq:BE:W12}
\abs{W_1} \lesssim_{\damrho_1} \norm{\mdbds_1}_1 \dammax_{1, 4}^4 \sqrt{\frac{\log B_1}{B_1}}, \quad
\abs{W_2} \lesssim_{\damrho_2} \norm{\mdbds_2}_1 \dammax_{2, 4}^4 \sqrt{\frac{\log B_2}{B_2}}.
\end{equation}

\item Combining the implications of \cref{eq:BE:U,eq:BE:V1,eq:BE:V2,eq:BE:W12} via \cref{lem:BE} and simplifying yields the conclusion.
\end{enumerate}

\paragraph{1.~Proof of \cref{eq:BE:U}}

Let
\begin{equation*}
Z_n = \begin{cases} \effN_1^{-1} \mdbdsT_1 \rb{Y_{1, m_1, n}-\EE Y_{1, m_1, 1}} & \text{ for } n = 1, \ldots, \effN_1, \\ \effN_2^{-1} \mdbdsT_2 \rb{Y_{2, m_2, n-\effN_1}-\EE Y_{2, m_2, 1}} & \text{ for } n = \effN_1+1, \ldots, \effN_1+\effN_2. \end{cases}
\end{equation*}
Then, $Z_n$'s are independent but not identically distributed, and
\begin{equation*}
\sum_{n = 1}^{\effN_1+\effN_2} Z_n = \mdbdsT_1 \vec\rb{\hat{\matS}_{1, m_1}-\EE \hat{\matS}_{1, m_1}} + \mdbdsT_2 \vec\rb{\hat{\matS}_{2, m_2}-\EE \hat{\matS}_{2, m_2}}.
\end{equation*}
By the standard Berry-Esseen inequality for the normalized sum of non-identically distributed random variables,
\begin{equation*}
\sup_z \lrabs{\PP\lrcb{U \leq z}-\Phi(z)} \leq \varepsilon_U,
\end{equation*}
where
\begin{equation*}
\varepsilon_U = \frac{c_1}{\sqrt{\effN}} \lrrb{\frac{\asympSD^2_{1, m_1}}{\relN_1} + \frac{\asympSD^2_{2, m_2}}{\relN_2}}^{-1/2} \max\lrcb{\frac{\skewness_{1, m_1}}{\asympSD^2_{1, m_1}}, \frac{\skewness_{2, m_2}}{\asympSD^2_{2, m_2}}}
\end{equation*}
with
\begin{equation*}
\skewness_{1, m_1} = \EE\lrsqb{\lrabs{\mdbdsT_1 \lrrb{Y_{1, m_1, 1}-\EE Y_{1, m_1, 1}}}^3}, \quad \skewness_{2, m_2} = \EE\lrsqb{\lrabs{\mdbdsT_2 \lrrb{Y_{2, m_2, 1}-\EE Y_{2, m_2, 1}}}^3},
\end{equation*}
and an absolute constant $c_1 > 0$.

\paragraph{2.~Proof of \cref{eq:BE:V1}}

We prove the $V_{11}$ part of \cref{eq:BE:V1} only; the proof for the $V_{21}$ part proceeds similarly.

Let $\delta > 0$ be arbitrary. By H\"{o}lder's inequality and Boole's inequality,
\begin{multline} \label{eq:BE:V1:1}
\PP\lrcb{\abs{V_{11}} > \delta}
\leq \PP\lrcb{\max_{i, j} \lrabs{\lrrb{\Sh_{1, i, j}-\EE \Sh_{1, i, j}}-\lrrb{\hat{\matS}_{1, m_1, i, j}-\EE \hat{\matS}_{1, m_1, i, j}}} > \frac{\asympSD_{m_1, m_2} \delta}{\norm{\mdbds_1}_1 \sqrt{\effN}}} \\
\leq \sum_{i, j} \PP\lrcb{\lrabs{\lrrb{\Sh_{1, i, j}-\EE \Sh_{1, i, j}}-\lrrb{\hat{\matS}_{1, m_1, i, j}-\EE \hat{\matS}_{1, m_1, i, j}}} > \frac{\asympSD_{m_1, m_2} \delta}{\norm{\mdbds_1}_1 \sqrt{\effN}}}.
\end{multline}

Furthermore, for each $i, j$,
\begin{multline*}
\PP\lrcb{\lrabs{\lrrb{\Sh_{1, i, j}-\EE \Sh_{1, i, j}}-\lrrb{\hat{\matS}_{1, m_1, i, j}-\EE \hat{\matS}_{1, m_1, i, j}}} > \frac{\asympSD_{m_1, m_2} \delta}{\norm{\mdbds_1}_1 \sqrt{\effN}}} \\
\leq \PP\lrcb{\lrabs{\lrrb{\Sh_{1, i, j}-\EE \Sh_{1, i, j}}-\lrrb{\Sh_{1, m_1, i, j}-\EE \Sh_{1, m_1, i, j}}} > \frac{\asympSD_{m_1, m_2} \delta}{2 \norm{\mdbds_1}_1 \sqrt{\effN}}} \\
+ \PP\lrcb{\lrabs{\lrrb{\Sh_{1, m_1, i, j}-\EE \Sh_{1, m_1, i, j}}-\lrrb{\hat{\matS}_{1, m_1, i, j}-\EE \hat{\matS}_{1, m_1, i, j}}} > \frac{\asympSD_{m_1, m_2} \delta}{2 \norm{\mdbds_1}_1 \sqrt{\effN}}},
\end{multline*}
where $\Sh_{1, m_1}$ and $\Sh_{2, m_2}$ are as in \cref{eq:Sighm} or \cref{eq:Sighm:real}.

Lemma 12 of \citet{zhang2025spectral} implies that for each $i, j$,
\begin{multline*}
\PP\lrcb{\lrabs{\lrrb{\Sh_{1, i, j}-\EE \Sh_{1, i, j}}-\lrrb{\Sh_{1, m_1, i, j}-\EE \Sh_{1, m_1, i, j}}} > \frac{\asympSD_{m_1, m_2} \delta}{2 \norm{\mdbds_1}_1 \sqrt{\effN}}} \\
\lesssim \exp\lrcb{-c_{\daOalpha_1, 1} \lrrb{\frac{\asympSD_{m_1, m_2} \delta}{\norm{\mdbds_1}_1 \norm{X_{1, \cdot}}_{\daOpsi_{\daOalpha_1}}^2 \damrho_1^{m_1} \sqrt{\relN_1}}}^{\frac{1}{1 + 2 \daOalpha_1}}},
\end{multline*}
where $\damrho_1 \in (0, 1)$ is a constant satisfying the condition of \cref{eq:dam} for $X_{1, \cdot}$ and $c_{\daOalpha_1, 1} > 0$ is a constant that depends on $\daOalpha_1$.

Lemma 13 of \citet{zhang2025spectral} implies that for each $i, j$,
\begin{multline*}
\PP\lrcb{\lrabs{\lrrb{\Sh_{1, m_1, i, j}-\EE \Sh_{1, m_1, i, j}}-\lrrb{\hat{\matS}_{1, m_1, i, j}-\EE \hat{\matS}_{1, m_1, i, j}}} > \frac{\asympSD_{m_1, m_2} \delta}{2 \norm{\mdbds_1}_1 \sqrt{\effN}}} \\
\lesssim \exp\lrcb{-c_{\daOalpha_1, 2} \lrrb{\frac{\asympSD_{m_1, m_2} \sqrt{B_1/m_1} \delta}{\norm{\mdbds_1}_1 \norm{X_{1, \cdot}}_{\daOpsi_{\daOalpha_1}}^2 \sqrt{\relN_1}}}^{\frac{1}{1 + 2 \daOalpha_1}}},
\end{multline*}
where $c_{\daOalpha_1, 2} > 0$ is a constant that depends on $\daOalpha_1$.

Setting $m_1 \asymp \log B_1 / \log (1/\damrho_1)$, we obtain
\begin{multline} \label{eq:BE:V1:2}
\PP\lrcb{\lrabs{\lrrb{\Sh_{1, i, j}-\EE \Sh_{1, i, j}}-\lrrb{\hat{\matS}_{1, m_1, i, j}-\EE \hat{\matS}_{1, m_1, i, j}}} > \frac{\asympSD_{m_1, m_2} \delta}{\norm{\mdbds_1}_1 \sqrt{\effN}}} \\
\lesssim \exp\lrcb{-c_{\daOalpha_1, 1} \lrrb{\frac{\asympSD_{m_1, m_2} B_1 \delta}{\norm{\mdbds_1}_1 \norm{X_{1, \cdot}}_{\daOpsi_{\daOalpha_1}}^2 \sqrt{\relN_1}}}^{\frac{1}{1 + 2 \daOalpha_1}}} \\
+ \exp\lrcb{-c_{\daOalpha_1, 2} \lrrb{\frac{\asympSD_{m_1, m_2} \sqrt{\log (1/\damrho_1)} \sqrt{B_1/\log B_1} \delta}{\norm{\mdbds_1}_1 \norm{X_{1, \cdot}}_{\daOpsi_{\daOalpha_1}}^2 \sqrt{\relN_1}}}^{\frac{1}{1 + 2 \daOalpha_1}}}.
\end{multline}

Since $B_1 \gtrsim \sqrt{B_1 / \log B_1}$, the second term dominates. Combining \cref{eq:BE:V1:1} with \cref{eq:BE:V1:2},
\begin{equation*}
\PP\lrcb{\abs{V_{11}} > \delta} \lesssim d \exp\lrcb{-c_{\daOalpha_1, 3} \lrrb{\frac{\asympSD_{m_1, m_2} \sqrt{\log (1/\damrho_1)} \sqrt{B_1/\log B_1} \delta}{\norm{\mdbds_1}_1 \norm{X_{1, \cdot}}_{\daOpsi_{\daOalpha_1}}^2 \sqrt{\relN_1}}}^{\frac{1}{1 + 2 \daOalpha_1}}}.
\end{equation*}

Setting $\delta = \delta_{V, 11}$, where
\begin{equation*}
\delta_{V, 11} = \frac{c_{\daOalpha_1, 4} \norm{\mdbds_1}_1 \norm{X_{1, \cdot}}_{\daOpsi_{\daOalpha_1}}^2 \rb{\log d}^{1 + 2 \daOalpha_1}}{\asympSD_{m_1, m_2} \sqrt{\log (1/\damrho_1)}} \sqrt{\frac{\relN_1 \log B_1}{B_1}}
\end{equation*}
for some constant $c_{\daOalpha_1, 4} > 0$ that depends on $\daOalpha_1$, we have $\varepsilon_{V, 11} \lesssim d^{-1}$.

\paragraph{3.~Proof of \cref{eq:BE:V2}}

By Proposition 2 of \citet{zhang2025spectral},
\begin{equation*}
\abs{V_{12}}
\leq \frac{\sqrt{\effN} \norm{\mdbds_1}_1 \infnorm{\EE \Sxh-\Sxs}}{\asympSD_{m_1, m_2}}
\lesssim_{\damrho_1} \frac{\norm{\mdbds_1}_1 \dammax_{1, 2}^2}{\asympSD_{m_1, m_2}} \sqrt{\frac{\relN_1 T_1}{B_1^3}},
\end{equation*}
where we have used $N = \relN_1 \effN_1$ and $\effN_1 = T_1 / B_1$. The bound for $V_{22}$ is proved similarly.

\paragraph{4.~Proof of \cref{eq:BE:W}}

We prove the $W_1$ part of \cref{eq:BE:W} only; the proof for the $W_2$ part proceeds similarly.

Lemma 11 of \citet{zhang2025spectral} computes the limiting covariance of each pair of components of $\hat{\matS}_m$, which is shown to satisfy the requirements of \cref{prop:asympcov} or \ref{prop:asympcov:real}. Therefore, the two results imply: in the case of $\freq \notin \cb{0, \pi}$,
\begin{multline*}
\abs{W_1}
= \abs{\asympSD^2_{1, m_1} - \asympSD^2_1}
= \lrabs{\mdbdsT_1 \lrcb{\Var\lrrb{Y_{1, m_1, 1}-\EE Y_{1, m_1, 1}}-\Permutation^\top \VarSxh} \mdbds_1} \\
\leq \norm{\mdbds_1}_1^2 \infnorm{\Var\lrrb{Y_{1, m_1, 1}-\EE Y_{1, m_1, 1}}-\Permutation^\top \VarSxh}
\lesssim_{\damrho_1} \norm{\mdbds_1}_1^2 \dammax_{1, 4}^4 \lrrb{\damrho_1^m + \sqrt{\frac{m_1}{B_1}}},
\end{multline*}
where $\Permutation$ is the matrix define above \cref{prop:asympcov}; in the case of $\freq \in \cb{0, \pi}$,
\begin{multline*}
\abs{W_1}
= \abs{\asympSD^2_{1, m_1} - \asympSD^2_1}
= \lrabs{\mdbdsT_1 \lrcb{\Var\lrrb{Y_{1, m_1, 1}-\EE Y_{1, m_1, 1}}-\VarSxh} \mdbds_1} \\
\leq \norm{\mdbds_1}_1^2 \infnorm{\Var\lrrb{Y_{1, m_1, 1}-\EE Y_{1, m_1, 1}}-\VarSxh}
\lesssim_{\damrho_1} \norm{\mdbds_1}_1^2 \dammax_{1, 4}^4 \lrrb{\damrho_1^m + \sqrt{\frac{m_1}{B_1}}}.
\end{multline*}
Setting $m_1 \asymp \log B_1 / \log (1/\damrho_1)$ as in Step 2 above,
\begin{equation} \label{eq:BE:lesssim_damrho}
\lrabs{\asympSD^2_{1, m_1}-\asympSD^2_1} \leq c_{\damrho_1} \norm{\mdbds_1}_1^2 \dammax_{1, 4}^4 \sqrt{\frac{\log B_1}{B_1}}
\end{equation}
for some constant $c_{\damrho_1} > 0$ that depends on $\damrho_1$. The bound for $W_2$ is proved similarly; this contains on a constant $c_{\damrho_2} > 0$ that depends on $\damrho_2$.

\paragraph{5.~Conclusion}

First, note that since
\begin{equation*}
\frac{\asympSD^2_{1, m_1}}{\asympSD^2_1} = 1 + \frac{\asympSD^2_{1, m_1} - \asympSD_1^2}{\asympSD^2_1}, \quad
\frac{\asympSD^2_{2, m_2}}{\asympSD^2_2} = 1 + \frac{\asympSD^2_{2, m_2} - \asympSD_2^2}{\asympSD^2_2}, 
\end{equation*}
we have
\begin{equation*}
1 - \frac{\abs{\asympSD^2_{1, m_1} - \asympSD_1^2}}{\asympSD^2_1} \leq \frac{\asympSD^2_{1, m_1}}{\asympSD^2_1} \leq 1 + \frac{\abs{\asympSD^2_{1, m_1} - \asympSD_1^2}}{\asympSD^2_1}, \quad
1 - \frac{\abs{\asympSD^2_{2, m_2} - \asympSD_2^2}}{\asympSD^2_2} \leq \frac{\asympSD^2_{2, m_2}}{\asympSD^2_2} \leq 1 + \frac{\abs{\asympSD^2_{2, m_2} - \asympSD_1^2}}{\asympSD^2_2}.
\end{equation*}
Under the condition of the theorem, this implies that
\begin{equation*}
\asympSD^2_{1, m_1} \geq \frac 1 4 \asympSD^2_1, \quad
\asympSD^2_{2, m_2} \geq \frac 1 4\asympSD^2_2, \quad
\asympSD^2_{m_1, m_2} \geq \frac 1 4 \asympSD^2.
\end{equation*}
In other words, we may replace $\asympSD_{m_1, m_2}$, $\asympSD_{1, m_1}$, and $\asympSD_{2, m_2}$ with $\asympSD$, $\asympSD_1$, and $\asympSD_2$, respectively, plus some multiplicative constants.

By \cref{eq:BE:W},
\begin{equation*}
\abs{W}
\leq \frac{2 c_{\damrho_1} \relN_1 \norm{\mdbds_1}_1^2 \dammax_{1, 4}^4}{3 \asympSD^2} \sqrt{\frac{\log B_1}{B_1}} + \frac{2 c_{\damrho_2} \relN_2 \norm{\mdbds_2}_1^2 \dammax_{2, 4}^4}{3 \asympSD^2} \sqrt{\frac{\log B_2}{B_2}}
\leq \frac 1 2,
\end{equation*}
and hence, $\abs{W} / \rb{1-\abs{W}} \leq 2 \abs{W}$.

By \cref{lem:BE}, \cref{eq:BE:U,eq:BE:V1,eq:BE:V2,eq:BE:W12} together imply
\begin{multline*}
\BEprob \lesssim_{\mdbds_1, \mdbds_2, \asympSD, \skewness, \damrho_1, \damrho_2, \daOalpha_1, \daOalpha_2} \frac{1}{\sqrt{\effN}} + d^{-1}
+ \norm{X_{1, \cdot}}_{\daOpsi_{\daOalpha_1}}^2 \rb{\log d}^{1 + 2 \daOalpha_1} \sqrt{\frac{\log B_1}{B_1}} + \dammax_{1, 2}^2 \sqrt{\frac{T_1}{B_1^3}} + \dammax_{1, 4}^4 \sqrt{\frac{\log B_1}{B_1}} \\
+ \norm{X_{2, \cdot}}_{\daOpsi_{\daOalpha_2}}^2 \rb{\log d}^{1 + 2 \daOalpha_2} \sqrt{\frac{\log B_2}{B_2}} + \dammax_{2, 2}^2 \sqrt{\frac{T_2}{B_2^3}} + \dammax_{2, 4}^4 \sqrt{\frac{\log B_2}{B_2}}.
\end{multline*}

For $p \geq 2$, $d \geq 3$, and we have
\begin{equation*}
\rb{\log d}^{1 + 2 \daOalpha_1} \sqrt{\frac{\log B_1}{B_1}} \geq \sqrt{\frac{\log B_1}{B_1}}, \quad
\rb{\log d}^{1 + 2 \daOalpha_2} \sqrt{\frac{\log B_2}{B_2}} \geq \sqrt{\frac{\log B_2}{B_2}}.
\end{equation*}
Also, by \cref{eq:daO},
\begin{equation*}
4^{\daOalpha_1} \norm{X_{1, \cdot}}_{\daOpsi_{\daOalpha_1}} \geq \dammax_{1, 4}, \quad 4^{\daOalpha_2} \norm{X_{2, \cdot}}_{\daOpsi_{\daOalpha_2}} \geq \dammax_{2, 4}.
\end{equation*}
Thus,
\begin{multline*}
\BEprob \lesssim_{\mdbds_1, \mdbds_2, \asympSD, \skewness, \damrho_1, \damrho_2, \daOalpha_1, \daOalpha_2} \frac{1}{\sqrt{\effN}} + d^{-1}
+ \norm{X_{1, \cdot}}_{\daOpsi_{\daOalpha_1}}^4 \rb{\log d}^{1 + 2 \daOalpha_1} \sqrt{\frac{\log B_1}{B_1}} + \dammax_{1, 2}^2 \sqrt{\frac{T_1}{B_1^3}} \\
+ \norm{X_{2, \cdot}}_{\daOpsi_{\daOalpha_2}}^4 \rb{\log d}^{1 + 2 \daOalpha_2} \sqrt{\frac{\log B_2}{B_2}} + \dammax_{2, 2}^2 \sqrt{\frac{T_2}{B_2^3}}.
\end{multline*}
\end{proof}

\begin{remark}
An astute reader may have noticed that the error bound also depends on $\skewness_{1, m_1}$ and $\skewness_{2, m_2}$, the third central moments of the $m_1$- and $m_2$-dependent approximations. By extending the analysis of \citet{zhang2025spectral}, both $\skewness_{1, m_1}$ and $\skewness_{2, m_2}$ may be bounded by moments of the underlying processes, which are finite by \cref{assm}. However, since the term in which $\skewness_{1, m_1}$ and $\skewness_{2, m_2}$ appear is not the dominant error term, we do not pursue this topic further.
\end{remark}

\section{Proof of \cref{thm:BE:SDD} and \cref{cor:BE:SDD:abbr}} \label{app:proof:BE:SDD}
In this section, we prove \cref{thm:BE:SDD} and \cref{cor:BE:SDD:abbr}. Most of the hard work has already been accomplished in \cref{lem:BE:Dtrace:abbr} and \cref{thm:BE:projWelch}. All that remains to be done is to combine the results and simplify the bounds.

\begin{proof}[Proof of \cref{thm:BE:SDD}]
By \cref{lem:BE:Dtrace:abbr},
\begin{equation*}
\sup_z \lrabs{\PP\lrcb{\asympSD^{-1} \sqrt{\effN} \lrrb{\Dt_0-\Ds_0} \leq z} - \Phi(z)} \leq \BEprob + \LDprob + \LDrate
\end{equation*}
with
\begin{equation*}
\LDrate \lesssim_{\Sxs, \Sys, \asympSD} \sqrt{\effN} \sparsity_{\text{tot}} \lrrb{\Sxhrate^2 + \Syhrate^2} + \Sxhrate + \Syhrate.
\end{equation*}

By \cref{eq:Sxhrate,eq:Syhrate},
\begin{multline*}
\LDrate \lesssim_{\Sxs, \Sys, \asympSD, \damrho_1, \damrho_2, \daOalpha_1, \daOalpha_2}
\sqrt{\effN} \sparsity_{\text{tot}} \lrrb{\norm{X_{1, \cdot}}_{\daOpsi_{\daOalpha_1}}^4 \rb{\log d}^{2 + 4 \daOalpha_1} \frac{B_1}{T_1} + \frac{\dammax_{1, 2}^4}{B_1^2} + \norm{X_{2, \cdot}}_{\daOpsi_{\daOalpha_2}}^4 \rb{\log d}^{2 + 4 \daOalpha_2} \frac{B_2}{T_2} + \frac{\dammax_{2, 2}^4}{B_2^2}} \\
+ \norm{X_{1, \cdot}}_{\daOpsi_{\daOalpha_1}}^2 \rb{\log d}^{1 + 2 \daOalpha_1} \sqrt{\frac{B_1}{T_1}} + \frac{\dammax_{1, 2}^2}{B_1} + \norm{X_{2, \cdot}}_{\daOpsi_{\daOalpha_2}}^2 \rb{\log d}^{1 + 2 \daOalpha_2} \sqrt{\frac{B_2}{T_2}} + \frac{\dammax_{2, 2}^2}{B_2},
\end{multline*}
where we recall $(a+b)^2 \leq 2a^2 + 2b^2$.

Put $\relN_1 = \effN/\effN_1$, and $\relN_2 = \effN/\effN_2$. Then,
\begin{gather*}
\frac{\sqrt{\effN} B_1}{T_1} = \sqrt{\frac{\effN}{\effN_1^2}} = \sqrt{\frac{\relN_1 B_1}{T_1}}, \quad \frac{\sqrt{\effN}}{B_1^2} = \sqrt{\frac{\relN_1 \effN_1}{B_1^4}} = \sqrt{\frac{\relN_1 T_1}{B_1^5}}, \\
\frac{\sqrt{\effN} B_2}{T_2} = \sqrt{\frac{\effN}{\effN_2^2}} = \sqrt{\frac{\relN_2 B_2}{T_2}}, \quad \frac{\sqrt{\effN}}{B_2^2} = \sqrt{\frac{\relN_2 \effN_2}{B_2^4}} = \sqrt{\frac{\relN_2 T_2}{B_2^5}},
\end{gather*}
and we have
\begin{multline*}
\LDrate \lesssim_{\Sxs, \Sys, \asympSD, \damrho_1, \damrho_2, \daOalpha_1, \daOalpha_2} \\
\norm{X_{1, \cdot}}_{\daOpsi_{\daOalpha_1}}^4 \sparsity_{\text{tot}} \rb{\log d}^{2 + 4 \daOalpha_1} \sqrt{\frac{B_1}{T_1}} + \dammax_{1, 2}^4 \sparsity_{\text{tot}} \sqrt{\frac{T_1}{B_1^5}} 
+ \norm{X_{2, \cdot}}_{\daOpsi_{\daOalpha_2}}^4 \sparsity_{\text{tot}} \rb{\log d}^{2 + 4 \daOalpha_2} \sqrt{\frac{B_2}{T_2}} + \dammax_{2, 2}^4 \sparsity_{\text{tot}} \sqrt{\frac{T_2}{B_2^5}} \\
+ \norm{X_{1, \cdot}}_{\daOpsi_{\daOalpha_1}}^2 \rb{\log d}^{1 + 2 \daOalpha_1} \sqrt{\frac{B_1}{T_1}} + \frac{\dammax_{1, 2}^2}{B_1}
+ \norm{X_{2, \cdot}}_{\daOpsi_{\daOalpha_2}}^2 \rb{\log d}^{1 + 2 \daOalpha_2} \sqrt{\frac{B_2}{T_2}} + \frac{\dammax_{2, 2}^2}{B_2}.
\end{multline*}
Since $d \geq 3$ for $p \geq 2$,
\begin{equation*}
\sqrt{\frac{B_1}{T_1}} \rb{\log d}^{1 + 2 \daOalpha_1} \leq \sparsity_{\text{tot}} \sqrt{\frac{B_1}{T_1}} \rb{\log d}^{2 + 4 \daOalpha_1}, \quad
\sqrt{\frac{B_2}{T_2}} \rb{\log d}^{1 + 2 \daOalpha_2}  \leq \sparsity_{\text{tot}} \sqrt{\frac{B_2}{T_2}} \rb{\log d}^{2 + 4 \daOalpha_2},
\end{equation*}
and we have
\begin{multline*}
\LDrate \lesssim_{\Sxs, \Sys, \asympSD, \damrho_1, \damrho_2, \daOalpha_1, \daOalpha_2}
\max\lrcb{\norm{X_{1, \cdot}}_{\daOpsi_{\daOalpha_1}}^2, \norm{X_{1, \cdot}}_{\daOpsi_{\daOalpha_1}}^4} \sparsity_{\text{tot}} \rb{\log d}^{2 + 4 \daOalpha_1} \sqrt{\frac{B_1}{T_1}} + \dammax_{1, 2}^4 \sparsity_{\text{tot}} \sqrt{\frac{T_1}{B_1^5}} + \frac{\dammax_{1, 2}^2}{B_1} \\
+ \max\lrcb{\norm{X_{2, \cdot}}_{\daOpsi_{\daOalpha_2}}^2, \norm{X_{2, \cdot}}_{\daOpsi_{\daOalpha_2}}^4} \sparsity_{\text{tot}} \rb{\log d}^{2 + 4 \daOalpha_2} \sqrt{\frac{B_2}{T_2}} + \dammax_{2, 2}^4 \sparsity_{\text{tot}} \sqrt{\frac{T_2}{B_2^5}} + \frac{\dammax_{2, 2}^2}{B_2}.
\end{multline*}

On the other hand, \cref{thm:BE:projWelch} implies that
\begin{multline*}
\BEprob \lesssim_{\mdbds_1, \mdbds_2, \asympSD, \damrho_1, \damrho_2, \daOalpha_1, \daOalpha_2} \frac{1}{\sqrt{\effN}} + d^{-1}
+ \norm{X_{1, \cdot}}_{\daOpsi_{\daOalpha_1}}^4 \rb{\log d}^{1 + 2 \daOalpha_1} \sqrt{\frac{\log B_1}{B_1}} + \dammax_{1, 2}^2 \sqrt{\frac{T_1}{B_1^3}} \\
+ \norm{X_{2, \cdot}}_{\daOpsi_{\daOalpha_2}}^4 \rb{\log d}^{1 + 2 \daOalpha_2} \sqrt{\frac{\log B_2}{B_2}} + \dammax_{2, 2}^2 \sqrt{\frac{T_2}{B_2^3}}.
\end{multline*}

Again, $d \geq 3$ for $p \geq 2$. Also, $B_1, B_2 \geq \sparsity_{\text{tot}}$. Under these conditions,
\begin{gather*}
\frac{1}{\sqrt{\effN}} = \frac{1}{\sqrt{\relN_1 \effN_1}} \lesssim \frac{\sparsity_{\text{tot}} \rb{\log d}^{2 + 4 \daOalpha_1}}{\sqrt{\effN_1}}, \quad \frac{1}{\sqrt{\effN}} = \frac{1}{\sqrt{\relN_2 \effN_2}} \lesssim \frac{\sparsity_{\text{tot}} \rb{\log d}^{2 + 4 \daOalpha_2}}{\sqrt{\effN_2}}, \\
\sparsity_{\text{tot}} \sqrt{\frac{T_1}{B_1^5}} \leq \sqrt{\frac{T_1}{B_1^3}}, \quad \sparsity_{\text{tot}} \sqrt{\frac{T_2}{B_2^5}} \leq \sqrt{\frac{T_2}{B_2^3}}.
\end{gather*}
Also, clearly,
\begin{equation*}
\frac{1}{B_1} \leq \sqrt{\frac{T_1}{B_1^3}}, \quad \frac{1}{B_2} \leq \sqrt{\frac{T_2}{B_2^3}}.
\end{equation*}
Thus,
\begin{multline*}
\sup_z \lrabs{\PP\lrcb{\asympSD^{-1} \sqrt{\effN} \lrrb{\Dt_0-\Ds_0} \leq z} - \Phi(z)} \lesssim_{\Sxs, \Sys, \mdbds_1, \mdbds_2, \asympSD, \damrho_1, \damrho_2, \daOalpha_1, \daOalpha_2} d^{-1} \\
+ \max\lrcb{\norm{X_{1, \cdot}}_{\daOpsi_{\daOalpha_1}}^2, \norm{X_{1, \cdot}}_{\daOpsi_{\daOalpha_1}}^4} \sparsity_{\text{tot}} \rb{\log d}^{2 + 4 \daOalpha_1} \sqrt{\frac{B_1}{T_1}} + \norm{X_{1, \cdot}}_{\daOpsi_{\daOalpha_1}}^4 \rb{\log d}^{1 + 2 \daOalpha_1} \sqrt{\frac{\log B_1}{B_1}} + \max\lrcb{\dammax_{1, 2}^2, \dammax_{1, 2}^4} \sqrt{\frac{T_1}{B_1^3}} \\
+ \max\lrcb{\norm{X_{2, \cdot}}_{\daOpsi_{\daOalpha_2}}^2, \norm{X_{2, \cdot}}_{\daOpsi_{\daOalpha_2}}^4} \sparsity_{\text{tot}} \rb{\log d}^{2 + 4 \daOalpha_2} \sqrt{\frac{B_2}{T_2}} + \norm{X_{2, \cdot}}_{\daOpsi_{\daOalpha_2}}^4 \rb{\log d}^{1 + 2 \daOalpha_2} \sqrt{\frac{\log B_2}{B_2}} + \max\lrcb{\dammax_{2, 2}^2, \dammax_{2, 2}^4} \sqrt{\frac{T_2}{B_2^3}}.
\end{multline*}
\end{proof}

\begin{proof}[Proof of \cref{cor:BE:SDD:abbr}]
Under the condition on the window sizes $B_1$ and $B_2$,
\begin{gather*}
\rb{\log d}^{2 + 4 \daOalpha_1} \sqrt{\frac{B_1}{T_1}} \asymp \frac{\rb{\log d}^{\frac{3}{2} + 3 \daOalpha_1}}{T_1^{\frac 1 4}}
\quad
\sqrt{\frac{\log B_1}{B_1}} \asymp \frac{\rb{\log T_1}^{\frac 1 2} \rb{\log d}^{\frac{1}{2} + \daOalpha_1}}{T_1^{\frac 1 4}}, \quad
\sqrt{\frac{T_1}{B_1^3}} \asymp \frac{\rb{\log d}^{\frac{3}{2} + 3 \daOalpha_1}}{T_1^{\frac 1 4}}, \\
\rb{\log d}^{2 + 4 \daOalpha_2} \sqrt{\frac{B_2}{T_2}} \asymp \frac{\rb{\log d}^{\frac{3}{2} + 3 \daOalpha_2}}{T_2^{\frac 1 4}}
\quad
\sqrt{\frac{\log B_2}{B_2}} \asymp \frac{\rb{\log T_2}^{\frac 1 2} \rb{\log d}^{\frac{1}{2} + \daOalpha_2}}{T_2^{\frac 1 4}}, \quad
\sqrt{\frac{T_2}{B_2^3}} \asymp \frac{\rb{\log d}^{\frac{3}{2} + 3 \daOalpha_2}}{T_2^{\frac 1 4}}.
\end{gather*}

Plugging these into \cref{eq:BE:SDD} yields
\begin{multline*}
\sup_z \lrabs{\PP\lrcb{\asympSD^{-1} \sqrt{\effN} \lrrb{\Dt_0-\Ds_0} \leq z} - \Phi(z)} \lesssim_{\Sxs, \Sys, \mdbds_1, \mdbds_2, \asympSD, \damrho_1, \damrho_2, \daOalpha_1, \daOalpha_2} d^{-1} \\
+ \frac{\max\lrcb{\norm{X_{1, \cdot}}_{\daOpsi_{\daOalpha_1}}^2, \norm{X_{1, \cdot}}_{\daOpsi_{\daOalpha_1}}^4} \sparsity_{\text{tot}} \rb{\log d}^{\frac{3}{2} + 3 \daOalpha_1}}{T_1^{\frac{1}{4}}}
+ \frac{\max\lrcb{\norm{X_{2, \cdot}}_{\daOpsi_{\daOalpha_2}}^2, \norm{X_{2, \cdot}}_{\daOpsi_{\daOalpha_2}}^4} \sparsity_{\text{tot}} \rb{\log d}^{\frac{3}{2} + 3 \daOalpha_2}}{T_2^{\frac{1}{4}}},
\end{multline*}
which tend to $0$ under the conditions of the corollary.
\end{proof}

\section{Asymptotic variance formulae} \label{app:asympcov}
In this section, we prove \cref{prop:asympcov,prop:asympcov:real}. We first explain where these results fit in our framework, and then present the proofs.

At a conceptual level, the proof of \cref{thm:BE:SDD} boils down to establishing that
\begin{equation*}
\sqrt{\effN} \lrrb{\Dt_0-\Ds_0} \stackrel{d}{\approx} \frac{\effN_1 + \effN_2}{\effN_1} \mdbdsT_1 \Permutation^\top \begin{bmatrix} \vec \Re \matZ_1 \\ \vec \Im \matZ_1 \end{bmatrix} + \frac{\effN_1 + \effN_2}{\effN_2} \mdbdsT_2 \Permutation^\top \begin{bmatrix} \vec \Re \matZ_2 \\ \vec \Im \matZ_2 \end{bmatrix}
\end{equation*}
for some complex Normal random matrices $\matZ_1, \matZ_2 \in \complex^{p \times p}$ in the case of $\freq \neq 0$ and $\pi$, or
\begin{equation*}
\sqrt{\effN} \lrrb{\Dt_0-\Ds_0} \stackrel{d}{\approx} \frac{\effN_1 + \effN_2}{\effN_1} \mdbdsT_1 \vec \matZ_1 + \frac{\effN_1 + \effN_2}{\effN_2} \mdbdsT_2 \vec \matZ_2
\end{equation*}
for some real Normal random matrices $\matZ_1, \matZ_2 \in \reals^{p \times p}$ in the case of $\freq = 0$ or $\pi$. Here, $\matZ_1$ and $\matZ_2$ are independent with the covariances and pseudo-covariances that match the limiting covariances and pseudo-covariances of Welch's estimators $\SDh_1$ and $\SDh_2$, respectively. The fixed vectors $\mdbds_1$ and $\mdbds_2$ are as defined in \cref{eq:mdbds}, and the fixed matrix $\Permutation$ is as defined prior to \cref{prop:asympcov,prop:asympcov:real}.

\citet{zhang2025spectral} gave the limiting covariance and pseudo-covariance for each pair of components in a single Welch's estimator, but did not specify the limiting joint covariance and pseudo-covariance structures implied by the pairwise result. However, since
\begin{equation*}
\Var\lrrb{\vecw^\top \Permutation^\top \begin{bmatrix} \vec \Re \matZ \\ \vec \Im \matZ \end{bmatrix}} = \lrrb{\Permutation \vecw}^\top \Var\lrrb{\begin{bmatrix} \vec \Re \matZ \\ \vec \Im \matZ \end{bmatrix}} \Permutation \vecw,
\end{equation*}
it is helpful to have a matrix characterization of the limiting covariance matrix. Although the idea may appear routine, inferring the full joint covariance structure from the pairwise result shall require nontrivial utilizations of the properties of the Kronecker product.

Below are some facts about the Kronecker product that will be used in our proofs. \cref{fact:kron} below clarifies how each component of a Kronecker product relates to components of the original matrices. \cref{fact:K} below specifies the effect of pre-multiplying a Kronecker product by the commutation matrix $\Commutation$. Recall that the commutation matrix $\Commutation = \sum_{i = 1}^{p} \sum_{j = 1}^{p} \Basis_{ij} \otimes \Basis_{ji}$, where $\Basis_{ij}$ is the $(i, j)$-th standard basis matrix in $\reals^{p \times p}$.

\begin{fact} \label{fact:kron}
Let $\matA, \matB \in \reals^{p \times p}$. Given $i, j, k, l \in [p]$,
\begin{equation*}
a_{i, k} b_{j, l} = \lrrb{\matB \otimes \matA}_{q, r},
\end{equation*}
where
\begin{equation*}
q = (j-1)p + i, \quad r = (l-1)p + k.
\end{equation*}
Conversely, given $q, r \in [p^2]$,
\begin{equation*}
\lrrb{\matB \otimes \matA}_{q, r} = a_{i, k} b_{j, l},
\end{equation*}
where
\begin{gather*}
    j = \lceil q/p\rceil, \quad i = (q-1) \% p+1 = q-(j-1) p, \\
    l = \lceil r/p\rceil, \quad k = (r-1) \% p+1 = r-(l-1) p.
\end{gather*}
Here, $\%$ denotes the remainder after dividing by an integer.
\end{fact}

\cref{fact:kron} is immediate from the definition of Kronecker product, and hence, its proof is omitted.

\begin{fact} \label{fact:K}
Let $\matA, \matB \in \reals^{p \times p}$. Given $i, j, k, l \in [p]$,
\begin{equation*}
a_{j, k} b_{i, l} = \lrcb{\Commutation \lrrb{\matB \otimes \matA}}_{q, r},
\end{equation*}
where
\begin{equation*}
q = (j-1)p + i, \quad r = (l-1)p + k.
\end{equation*}
Conversely, given $q, r \in [p^2]$,
\begin{equation*}
\lrcb{\Commutation \lrrb{\matB \otimes \matA}}_{q, r} = a_{j, k} b_{i, l},
\end{equation*}
where
\begin{gather*}
    j = \lceil q/p\rceil, \quad i = (q-1) \% p+1 = q-(j-1) p, \\
    l = \lceil r/p\rceil, \quad k = (r-1) \% p+1 = r-(l-1) p.
\end{gather*}
\end{fact}

\begin{proof}
We prove the second statement only; the first statement follows by solving for $q$ and $r$ in the second statement. We can write
\begin{equation*}
\lrcb{\Commutation \lrrb{\matB \otimes \matA}}_{q, r} = \basis_q^\top \Commutation \lrrb{\matB \otimes \matA} \basis_r,
\end{equation*}
where $\basis_q$ and $\basis_r$ respectively denote the $q$-th and the $r$-th standard basis vectors in $\reals^{p^2}$.

Recall the definition of $\Commutation$:
\begin{equation*}
\Commutation = \sum_{i'=1}^{p} \sum_{j'=1}^{p} \Basis_{i'j'} \otimes \Basis_{j'i'}.
\end{equation*}
Note that for each $i', j' \in [p]$, $\Basis_{i'j'} \otimes \Basis_{j'i'}$ has a single nonzero component 1 in the $((i'-1)p+j', (j'-1)p+i')$-th position. Thus,
\begin{equation*}
\Commutation \basis_q = \sum_{i'=1}^{p} \sum_{j'=1}^{p} \lrrb{\Basis_{i'j'} \otimes \Basis_{j'i'}} \basis_q = \lrrb{\Basis_{ji} \otimes \Basis_{ij}} \basis_q,
\end{equation*}
where
\begin{equation*}
j = \lceil q/p \rceil, \quad i = (q-1) \% p+1 = q-(j-1) p.
\end{equation*}
Multiplying by $\basis_q$ picks out the $q$-th column, and the sum representation makes clear that the only nonzero contribution occurs when $i' = j$ and $j' = i$. Moreover, the $q$-th column of $\Basis_{ji} \otimes \Basis_{ij}$ is itself a standard basis vector $\basis_{q'}$, where $q' = (i-1)p+j$. Thus, by \cref{fact:kron},
\begin{equation*}
\lrcb{\Commutation \lrrb{\matB \otimes \matA}}_{q, r} = \basis_{q'}^\top \lrrb{\matB \otimes \matA} \basis_r = \lrrb{\matB \otimes \matA}_{q', r} = a_{j, k} b_{i, l}.
\end{equation*}
\end{proof}

We are ready to prove \cref{prop:asympcov,prop:asympcov:real}.

\begin{proof}[Proof of \cref{prop:asympcov}]
It suffices to show that $\VarSh = \Var\rb{\rb{\rb{\vec \Re \matZ}^\top, \rb{\vec \Im \matZ}^\top}^\top}$ has the required form.

For convenience, put
\begin{equation*}
\rSDs = \Re \SDs, \quad \iSDs = \Im \SDs.
\end{equation*}
We have
\begin{gather*}
f^*_{i, k} f^*_{l, j} = \lrrb{a^*_{i, k} + \imath b^*_{i, k}} \lrrb{a^*_{l, j} + \imath b^*_{l, j}} = \lrrb{a^*_{i, k} a^*_{l, j} - b^*_{i, k} b^*_{l, j}} + \imath \lrrb{a^*_{i, k} b^*_{l, j} + a^*_{l, j} b^*_{i, k}}, \\
f^*_{i, l} f^*_{k, j} = \lrrb{a^*_{i, l} + \imath b^*_{i, l}} \lrrb{a^*_{k, j} + \imath b^*_{k, j}} = \lrrb{a^*_{i, l} a^*_{k, j} - b^*_{i, l} b^*_{k, j}} + \imath \lrrb{a^*_{i, l} b^*_{k, j} + a^*_{k, j} b^*_{i, l}}.
\end{gather*}
By the property of covariance and pseudo-covariance of complex random variables,
\begin{align*}
2 \Cov\lrrb{\Re Z_{i, j}, \Re Z_{k, l}}
&= \Re\lrrb{f^*_{i, k} f^*_{l, j} + f^*_{i, l} f^*_{k, j}} \\
&= a^*_{i, k} a^*_{l, j} - b^*_{i, k} b^*_{l, j} + a^*_{i, l} a^*_{k, j} - b^*_{i, l} b^*_{k, j} \\
&= a^*_{i, k} a^*_{j, l} + b^*_{i, k} b^*_{j, l} + a^*_{j, k} a^*_{i, l} + b^*_{j, k} b^*_{i, l} \tag{$\SD$ is Hermitian} \\
&= \lrcb{\rSDs \otimes \rSDs + \iSDs \otimes \iSDs + \Commutation \lrrb{\rSDs \otimes \rSDs} + \Commutation \lrrb{\iSDs \otimes \iSDs}}_{q, r} \tag{\cref{fact:kron,fact:K}} \\
&= \lrcb{\lrrb{\Identity_{p^2} + \Commutation} \lrrb{\rSDs \otimes \rSDs+\iSDs \otimes \iSDs}}_{q, r}, \\
2 \Cov\lrrb{\Re Z_{i, j}, \Im Z_{k, l}}
&= \Im\lrrb{f^*_{i, l} f^*_{k, j} - f^*_{i, k} f^*_{l, j}} \\
&= a^*_{k, j} b^*_{i, l} + a^*_{i, l} b^*_{k, j} - a^*_{i, k} b^*_{l, j} - a^*_{l, j} b^*_{i, k} \\
&= a^*_{j, k} b^*_{i, l} - b^*_{j, k} a^*_{i, l} + a^*_{i, k} b^*_{j, l} - b^*_{i, k} a^*_{j, l} \tag{$\SD$ is Hermitian} \\
&= \lrcb{\Commutation \lrrb{\iSDs \otimes \rSDs} - \Commutation \lrrb{\rSDs \otimes \iSDs} + \iSDs \otimes \rSDs - \rSDs \otimes \iSDs}_{q, r} \tag{\cref{fact:kron,fact:K}} \\
&= \lrcb{\lrrb{\Identity_{p^2} + \Commutation} \lrrb{\iSDs \otimes \rSDs - \rSDs \otimes \iSDs}}_{q, r}, \\
2 \Cov\lrrb{\Im Z_{i, j}, \Re Z_{k, l}}
&= \Im\lrrb{f^*_{i, l} f^*_{k, j} + f^*_{i, k} f^*_{l, j}} \\
&= a^*_{k, j} b^*_{i, l} + a^*_{i, l} b^*_{k, j} + a^*_{i, k} b^*_{l, j} + a^*_{l, j} b^*_{i, k} \\
&= a^*_{j, k} b^*_{i, l} - b^*_{j, k} a^*_{i, l} - a^*_{i, k} b^*_{j, l} + b^*_{i, k} a^*_{j, l} \tag{$\SD$ is Hermitian} \\
&= \lrcb{\Commutation \lrrb{\iSDs \otimes \rSDs} - \Commutation \lrrb{\rSDs \otimes \iSDs} - \iSDs \otimes \rSDs + \rSDs \otimes \iSDs}_{q, r} \tag{\cref{fact:kron,fact:K}} \\
&= \lrcb{\lrrb{\Identity_{p^2} - \Commutation} \lrrb{\rSDs \otimes \iSDs - \iSDs \otimes \rSDs}}_{q, r}, \\
2 \Cov\lrrb{\Im Z_{i, j}, \Im Z_{k, l}}
&= \Re\lrrb{f^*_{i, k} f^*_{l, j} - f^*_{i, l} f^*_{k, j}} \\
&= a^*_{i, k} a^*_{l, j} - b^*_{i, k} b^*_{l, j} - a^*_{i, l} a^*_{k, j} + b^*_{i, l} b^*_{k, j} \\
&= a^*_{i, k} a^*_{j, l} + b^*_{i, k} b^*_{j, l} - a^*_{j, k} a^*_{i, l} - b^*_{j, k} b^*_{i, l} \tag{$\SD$ is Hermitian} \\
&= \lrcb{\rSDs \otimes \rSDs + \iSDs \otimes \iSDs - \Commutation \lrrb{\rSDs \otimes \rSDs} - \Commutation \lrrb{\iSDs \otimes \iSDs}}_{q, r} \tag{\cref{fact:kron,fact:K}} \\
&= \lrcb{\lrrb{\Identity_{p^2} - \Commutation} \lrrb{\rSDs \otimes \rSDs + \iSDs \otimes \iSDs}}_{q, r},
\end{align*}
where
\begin{equation*}
q = (j-1)p+i, \quad r = (l-1)p+k.
\end{equation*}
Thus,
\begin{align*}
\VarSh
&= \Var\lrrb{\begin{bmatrix} \vec \Re \matZ \\ \vec \Im \matZ \end{bmatrix}} \\
&= \frac 1 2
\begin{bmatrix}
\lrrb{\Identity_{p^2} + \Commutation} \lrrb{\rSDs \otimes \rSDs + \iSDs \otimes \iSDs} &
\lrrb{\Identity_{p^2} + \Commutation} \lrrb{\iSDs \otimes \rSDs - \rSDs \otimes \iSDs} \\
\lrrb{\Identity_{p^2} - \Commutation} \lrrb{\rSDs \otimes \iSDs - \iSDs \otimes \rSDs} &
\lrrb{\Identity_{p^2} - \Commutation} \lrrb{\rSDs \otimes \rSDs + \iSDs \otimes \iSDs}
\end{bmatrix}.
\end{align*}
\end{proof}

\begin{proof}[Proof of \cref{prop:asympcov:real}]
\begin{align*}
\Cov\lrrb{Z_{i, j}, Z_{k, l}}
&= f^*_{i, k} f^*_{l, j} + f^*_{i, l} f^*_{k, j} \\
&= f^*_{i, k} f^*_{j, l} + f^*_{j, k} f^*_{i, l} \tag{$\SDs$ is symmetric} \\
&= \lrcb{\SDs \otimes \SDs + \Commutation \lrrb{\SDs \otimes \SDs}}_{q, r} \tag{\cref{fact:kron,fact:K}} \\
&= \lrcb{\lrrb{\Identity_{p^2} + \Commutation} \lrrb{\SDs \otimes \SDs}}_{q, r},
\end{align*}
where $$q = (j-1)p+i, \quad r = (l-1)p+k.$$ Thus, $$\VarSh = \lrrb{\Identity_{p^2} + \Commutation} \lrrb{\SDs \otimes \SDs}.$$
\end{proof}

\section{Plug-in estimate of asymptotic variance is consistent} \label{app:plugin_var_est}
\newcommand*{\eeasympvar}{G}

In \cref{lem:plugin_var_est} we establish that the plug-in estimates of the asymptotic variance in \cref{thm:BE:SDD} are consistent. We establish this consistency for each of the three estimating equations, $\gradD, \gradDL, \gradDR,$ i.e. the symmetric ($\EEcoef = 1/2$), $\EEcoef = 0$, and $\EEcoef = 1$ estimating equations. Explicit forms for these variances can be derived similar to \cref{thm:BE:projWelch} by applying \cref{prop:asympcov} or \cref{prop:asympcov:real} with $\mdbds_1, \mdbds_2$ given by setting $\EEcoef = 1/2, 0, 1$ in \cref{eq:mdbds:general}. 

\begin{lemma}[Consistency of plug-in estimate of asymptotic variance]
    \label{lem:plugin_var_est}
    Suppose \cref{assm} holds and we are in the setting of \cref{lem:Dh,lem:dbdh}, then
    \begin{align*}
        \asympSDh^2 & \overset{p}{\rightarrow} \asympSD^2 \\
        \asympSDh_{\EEcoef = 0}^2 & \overset{p}{\rightarrow} \asympSD_{\EEcoef = 0}^2 \\
        \asympSDh_{\EEcoef = 1}^2 & \overset{p}{\rightarrow} \asympSD_{\EEcoef = 1}^2 \\
    \end{align*}
    where $\asympSD^2, \asympSD_{\EEcoef = 0}^2, \asympSD_{\EEcoef = 1}^2$ are the asymptotic variances in \cref{thm:BE:projWelch} corresponding to the symmetric ($\EEcoef = 1/2$), $\EEcoef = 0$, and $\EEcoef = 1$ estimating equations and $\asympSDh^2$, $\asympSDh_{\EEcoef = 0}^2$, $\asympSDh_{\EEcoef = 1}^2$ are the plug-in estimates where any unknowns are replaced by their estimates.
\end{lemma}

\begin{proof}
We will only prove this result for the case where $\freq \notin \cb{0,\pi}$ and for $\asympSDh^2$ as the results for $\freq \in \cb{0,\pi}$ and $\asympSDh_{\EEcoef = 0}^2$, $\asympSDh_{\EEcoef = 1}^2$ follow the same structure. Recall that for $\freq \notin \cb{0,\pi}$, the realified spectral density is a $2p \times 2p$ matrix and that the form of $\mdbds_l$ is given in \cref{eq:mdbds}. Using this, the expanded form of $\asympSD^2$ is
\begin{equation*}
    \asympSD^2 = 
    \begin{aligned}[t]
        & \frac{\effN}{4 \effN_1} \dbdsT \lrrb{\Identity_{2p} \otimes \Sxsinv \Sys + \Sxsinv \Sys \otimes \Identity_{2p}}^{\top} \Permutation^{\top} \VarSh_1 \Permutation \lrrb{\Identity_{2p} \otimes \Sxsinv \Sys + \Sxsinv \Sys \otimes \Identity_{2p}} \dbds \\
        &+ \frac{\effN}{4 \effN_2} \dbdsT  \lrrb{\Identity_{2p} \otimes \Sysinv \Sxs + \Sysinv \Sxs \otimes \Identity_{2p}}^{\top} \Permutation^{\top} \VarSh_2 \Permutation  \lrrb{\Identity_{2p} \otimes \Sysinv \Sxs + \Sysinv \Sxs \otimes \Identity_{2p}} \dbds \, .
    \end{aligned}
\end{equation*}
We write $\asympSD^2$ in an alternative form which uses only $\Sxs$, $\Sys$, and $\Ds$.
\begin{equation*}
    \asympSD^2 = 
    \begin{aligned}[t]
        & \frac{\effN}{4 \effN_1} \dbdsT \lrrb{\Sys \Ds \otimes \Identity_{2p} + \Identity_{2p} \otimes \Sys \Ds + 2 \Identity_{2p} \otimes \Identity_{2p}} \Permutation^{\top} \VarSh_1 \Permutation \lrrb{\Sys \Ds \otimes \Identity_{2p} + \Identity_{2p} \otimes \Sys \Ds + 2 \Identity_{2p} \otimes \Identity_{2p}}^{\top} \dbds \\
        &+ \frac{\effN}{4 \effN_2} \dbdsT  \lrrb{\Sxs \Ds \otimes \Identity_{2p} + \Identity_{2p} \otimes \Sxs \Ds - 2 \Identity_{2p} \otimes \Identity_{2p}} \Permutation^{\top} \VarSh_2 \Permutation  \lrrb{\Sxs \Ds \otimes \Identity_{2p} + \Identity_{2p} \otimes \Sxs \Ds - 2 \Identity_{2p} \otimes \Identity_{2p}}^{\top} \dbds \, .
    \end{aligned}
\end{equation*}
Define
\begin{align*}
    \matMs_1 & = \Sys \Ds \otimes \Identity_{2p} + \Identity_{2p} \otimes \Sys \Ds + 2 \Identity_{2p} \otimes \Identity_{2p}, &  \matMs_2 & = \Sxs \Ds \otimes \Identity_{2p} + \Identity_{2p} \otimes \Sxs \Ds - 2 \Identity_{2p} \otimes \Identity_{2p} \\
    \matVs_1 & =  \Permutation^{\top} \VarSh_1 \Permutation, & \matVs_2 & =  \Permutation^{\top} \VarSh_2 \Permutation \, ,
\end{align*}
then, denoting $\matGs_l = \matMs_l \matVs_l \matMsT_l$ and $\matGh_l = \matMh_l \matVh_l \matMhT_l$, we write
\begin{align*}
    \asympSD^2 & = \frac{\effN}{4 \effN_1} \dbdsT \matMs_1 \matVs_1 \matMsT_1 \dbds + \frac{\effN}{4 \effN_2} \dbdsT  \matMs_2 \matVs_2 \matMsT_2 \dbds \\
    & = \frac{\effN}{4 \effN_1} \dbdsT \matGs_1 \dbds + \frac{\effN}{4 \effN_2} \dbdsT  \matGs_2 \dbds \, .
\end{align*}
The plug-in estimate of $\asympSD^2$ is written as
\begin{align*}
    \asympSDh^2 & = \frac{\effN}{4 \effN_1} \dbdhT \matMh_1 \matVh_1 \matMhT_1 \dbdh + \frac{\effN}{4 \effN_2} \dbdhT  \matMh_2 \matVh_2 \matMhT_2 \dbdh \\
    & = \frac{\effN}{4 \effN_1} \dbdhT \matGh_1 \dbdh + \frac{\effN}{4 \effN_2} \dbdhT  \matGh_2 \dbdh \, ,
\end{align*}
where $\dbdhT$ is given in \cref{eqn:debias_direction}, and $\matMh_l$, $\matVh_l$ are formed by replacing population versions with their estimates. For example, $\Sys$ is replaced by $\Syh$ and $\Ds$ is replaced by $\Dh$. We proceed by showing that each of these pieces are consistent estimates of their population counterparts.

\paragraph{Convergence of $\matMh_l$ to $\matMs_l$.} We establish convergence for $l = 1$ as $l = 2$ follows with the same logic.
\begin{align*}
    \onenorm{\matMh_1 - \matMs_1} & = \onenorm{\Syh \Dh \otimes \Identity_{2p} + \Identity_{2p} \otimes \Syh \Dh - \Sys \Ds \otimes \Identity_{2p} - \Identity_{2p} \otimes \Sys \Ds} \\
    & = 2p \onenorm{\Syh \Dh - \Sys \Ds}  + 2p \onenorm{\Syh \Dh - \Sys \Ds}  \\
    & = 4p \onenorm{\rb{\Syh - \Sys} \Dh + \Sys \rb{\Dh - \Ds}} \\
    & = 4p \onenorm{\rb{\Syh - \Sys} \rb{\Dh - \Ds} + \rb{\Syh - \Sys}\Ds + \Sys \rb{\Dh - \Ds}} \\
    & \leq 8p^2 \infnorm{\Syh - \Sys} \onenorm{\Dh - \Ds} + 8p^2 \infnorm{\Syh - \Sys} \onenorm{\Ds} + 8p^2 \infnorm{\Sys} \onenorm{\Dh - \Ds} \, ,
\end{align*}
where we used that if $\matA, \matB \in \reals^{2p \times 2p}$, then $\onenorm{\matA \matB} \leq 2p\infnorm{\matA}\onenorm{\matB}$. Then, by \cref{lem:Dh},
\begin{align*}
    \Mxhrate = \onenorm{\matMh_1 - \matMs_1} & \lesssim 8p^2 \lrrb{\frac{ \rb{\Syhrate + \infnorm{\Sys} }\sparsity_{\Ds} \regD}{\mineigen_1 \mineigen_2} + \Syhrate \onenorm{\Ds}} \\
    \Myhrate = \onenorm{\matMh_2 - \matMs_2} & \lesssim 8p^2 \lrrb{\frac{ \rb{\Sxhrate + \infnorm{\Sxs} }\sparsity_{\Ds} \regD}{\mineigen_1 \mineigen_2} + \Sxhrate \onenorm{\Ds}}  \, ,
\end{align*}
where $\regD$ is defined in \cref{lem:Dh}.

\paragraph{Convergence of $\matVh_l$ to $\matVs_l$.} Next we will establish the rate of $\infnorm{\matVh_l - \matVs_l}$. Recall that $\SDs_l = \Re \SDs_l + \imath \Im \SDs_l = \rSDs_l + \imath \iSDs_l $ and in \cref{prop:asympcov},
\begin{equation*}
    \matVs_l = \Permutation^{\top} \VarSh_l \Permutation \, ,
\end{equation*}
where $\VarSh_l$ is given by \cref{eq:VarSh}
\begin{align*}
    \VarSh_l = \frac 1 2 \begin{bmatrix}
    \lrrb{\Identity_{p^2} + \Commutation} \lrrb{ \rSDs_l \otimes \rSDs_l + \iSDs_l \otimes \iSDs_l} &
    \lrrb{\Identity_{p^2} + \Commutation} \lrrb{ \iSDs_l \otimes \rSDs_l -  \rSDs_l \otimes \iSDs_l} \\
    \lrrb{\Identity_{p^2} - \Commutation} \lrrb{ \rSDs_l \otimes \iSDs_l - \iSDs_l \otimes \rSDs_l} &
    \lrrb{\Identity_{p^2} - \Commutation} \lrrb{\rSDs_l \otimes \rSDs_l + \iSDs_l \otimes \iSDs_l}
    \end{bmatrix}.
\end{align*}
We can remove the left and right multiplication of $\Permutation$ as each row and column only has one $1$. In other words, these matrices simply select one element from the matrix so their multiplication does not change the max norm. Similarly the rows and columns of $\Identity_{p^2} + \Commutation$ have either one entry with a value of $2$ or two entries each with a value of $1$ so the max norm with these matrices is at most the max norm without them multiplied by 2. Thus 
\begin{align*}
    & \infnorm{\matVh_l - \matVs_l} \\
    & \quad = \infnorm{\begin{bmatrix}
    \rb{\rSDh_l \otimes \rSDh_l + \iSDh_l \otimes \iSDh_l} - \rb{\rSDs_l \otimes \rSDs_l + \iSDs_l \otimes \iSDs_l} &
    \rb{\iSDh_l \otimes \rSDh_l - \rSDh_l \otimes \iSDh_l} - \rb{\iSDs_l \otimes \rSDs_l - \rSDs_l \otimes \iSDs_l} \\
    \rb{\rSDh_l \otimes \iSDh_l - \iSDh_l \otimes \rSDh_l} - \rb{\rSDs_l \otimes \iSDs_l - \iSDs_l \otimes \rSDs_l} &
    \rb{\rSDh_l \otimes \rSDh_l + \iSDh_l \otimes \iSDh_l} - \rb{\rSDs_l \otimes \rSDs_l + \iSDs_l \otimes \iSDs_l}
    \end{bmatrix}} \, ,
\end{align*}
which is equivalent to the maximum of the $(1,1)$ and $(2,1)$ blocks. We will consider each of these blocks separately. Consider just the (1,1) entry,
\begin{equation*}
    \infnorm{ \rb{\rSDh_l \otimes \rSDh_l + \iSDh_l \otimes \iSDh_l} - \rb{\rSDs_l \otimes \rSDs_l + \iSDs_l \otimes \iSDs_l} } \leq \infnorm{\rSDh_l \otimes \rSDh_l - \rSDs_l \otimes \rSDs_l} + \infnorm{\iSDh_l \otimes \iSDh_l - \iSDs_l \otimes \iSDs_l} \,.
\end{equation*}
Studying the first entry on the right hand side,
\begin{align*}
    \infnorm{\rSDh_l \otimes \rSDh_l - \rSDs_l \otimes \rSDs_l} & = \infnorm{ \rb{\rSDh_l - \rSDs_l} \otimes \rSDh_l + \rSDs_l \otimes \rb{\rSDh_l - \rSDs_l}} \\
    & = \infnorm{ \rb{\rSDh_l - \rSDs_l} \otimes \rb{\rSDh_l - \rSDs_l} + \rb{\rSDh_l - \rSDs_l} \otimes \rSDs_l + \rSDs_l \otimes \rb{\rSDh_l - \rSDs_l} } \\
    & \leq \infnorm{ \rb{\rSDh_l - \rSDs_l} \otimes \rb{\rSDh_l - \rSDs_l}} + 2 \infnorm{\rb{\rSDh_l - \rSDs_l}} \infnorm{\rSDs_l} \\
    & \leq  \infnorm{ \rb{\rSDh_l - \rSDs_l}}^2 + 2 \infnorm{\rb{\rSDh_l - \rSDs_l}} \infnorm{\rSDs_l} \, ,
\end{align*} 
where we have used the fact that $\infnorm{\matX \otimes \matY} \leq \infnorm{\matX} \infnorm{\matY}$ since the Kronecker product multiplies every entry in $\matX$ with every entry in $\matY$. Using that $\infnorm{\rSDh_l - \rSDs_l}$, $\infnorm{\iSDh_l - \iSDs_l}$ $\leq \infnorm{\Sh_l - \Ss_l}$ and also that $\infnorm{\rSDs_l}$, $\infnorm{\iSDs_l} <  \infnorm{\Ss_l}$ and repeating the same process for $\infnorm{\iSDh_l \otimes \iSDh_l - \iSDs_l \otimes \iSDs_l}$,
\begin{align*}
    & \infnorm{\rSDh_l \otimes \rSDh_l - \rSDs_l \otimes \rSDs_l}  \leq \LDlrate^2 + 2 \LDlrate \infnorm{\Ss_l} \\
    & \infnorm{\iSDh_l \otimes \iSDh_l - \iSDs_l \otimes \iSDs_l}  \leq \LDlrate^2 + 2 \LDlrate \infnorm{\Ss_l} \\
    & \infnorm{ \rb{\rSDh_l \otimes \rSDh_l + \iSDh_l \otimes \iSDh_l} - \rb{\rSDs_l \otimes \rSDs_l + \iSDs_l \otimes \iSDs_l} } \leq 2 \LDlrate^2 + 4 \LDlrate \infnorm{\Ss_l} \, .
\end{align*}
Next we consider the $(2,1)$ entry,
\begin{equation*}
    \infnorm{\rb{\rSDh_l \otimes \iSDh_l - \iSDh_l \otimes \rSDh_l} - \rb{\rSDs_l \otimes \iSDs_l - \iSDs_l \otimes \rSDs_l}} \leq \infnorm{\iSDs_l \otimes \rSDs_l - \iSDh_l \otimes \rSDh_l} + \infnorm{ \rSDh_l \otimes \iSDh_l - \rSDs_l \otimes \iSDs_l}  \, .
\end{equation*}
For the first entry on the right hand side,
\begin{align*}
    \infnorm{\iSDs_l \otimes \rSDs_l - \iSDh_l \otimes \rSDh_l} & =  \infnorm{\iSDh_l \otimes \rSDh_l - \iSDs_l \otimes \rSDs_l}  \\
    & = \infnorm{ \rb{\iSDh_l - \iSDs_l} \otimes \rSDh_l + \iSDs_l \otimes \rSDh_l - \iSDs_l \otimes \rSDs_l} \\
    & = \infnorm{\rb{\iSDh_l - \iSDs_l} \otimes \rSDh_l + \iSDs_l \otimes \rb{\rSDh_l - \rSDs_l} } \\
    & = \infnorm{\rb{\iSDh_l - \iSDs_l} \otimes \rb{\rSDh_l - \rSDs_l} + \rb{\iSDh_l - \iSDs_l} \otimes \rSDs_l + \iSDs_l \otimes \rb{\rSDh_l - \rSDs_l})} \\
    & \leq \infnorm{\iSDh_l - \iSDs_l} \infnorm{\rSDh_l - \rSDs_l} + \infnorm{\iSDh_l - \iSDs_l} \infnorm{\rSDs_l} + \infnorm{\rSDh_l - \rSDs_l} \infnorm{\iSDs_l} \, .
\end{align*}
Repeating the same process for $ \infnorm{\rSDh_l \otimes \iSDh_l - \rSDs_l \otimes \iSDs_l}$,
\begin{align*}
    & \infnorm{\iSDh_l \otimes \rSDh_l - \iSDs_l \otimes \rSDs_l} \leq \LDlrate^2 + 2 \LDlrate \infnorm{\Ss_l} \\
    & \infnorm{\rSDh_l \otimes \iSDh_l - \rSDs_l \otimes \iSDs_l} \leq \LDlrate^2 + 2 \LDlrate \infnorm{\Ss_l} \\
    & \infnorm{ \rb{\iSDh_l \otimes \rSDh_l - \rSDh_l \otimes \iSDh_l} - \rb{\iSDs_l \otimes \rSDs_l - \rSDs_l \otimes \iSDs_l} } \leq 2 \LDlrate^2 + 4 \LDlrate \infnorm{\Ss_l}  \, .
\end{align*}
and we conclude
\begin{align*}
    \Vxhrate & = \infnorm{\matVh_1 - \matVs_1} \leq 2 \Sxhrate^2 + 4 \Sxhrate \infnorm{\Sxs} \\
    \Vyhrate & = \infnorm{\matVh_2 - \matVs_2} \leq 2 \Syhrate^2 + 4 \Syhrate \infnorm{\Sys} \, .
\end{align*}

\paragraph{Convergence of $\Identity_{4p^2} \otimes \matMh_l \matVh_l$ to $\Identity_{4p^2} \otimes \matMs_l \matVs_l$.}
\begin{align*}
    \infnorm{\Identity_{4p^2} \otimes \matMh_l \matVh_l - \Identity_{4p^2} \otimes \matMs_l \matVs_l} & = \infnorm{\matMh_l \matVh_l -  \matMs_l \matVs_l} \\
    & = \infnorm{\rb{\matVhT_l \otimes \Identity_{4p^2}} \vec\rb{\matMh_l} - \rb{\matVsT_l \otimes \Identity_{4p^2}} \vec\rb{\matMs_l}} \\
    & = \infnorm{\rb{\matVhT_l \otimes \Identity_{4p^2} - \matVsT_l \otimes \Identity_{4p^2}} \vec\rb{\matMh_l} - \rb{\matVsT_l \otimes \Identity_{4p^2}} \vec\rb{\matMs_l - \matMh_l}} \\
    & \leq \infnorm{\rb{\matVhT_l \otimes \Identity_{4p^2} - \matVsT_l \otimes \Identity_{4p^2}} \vec\rb{\matMh_l - \matMs_l}} \\ 
    & \quad + \infnorm{\rb{\matVhT_l \otimes \Identity_{4p^2} - \matVsT_l \otimes \Identity_{4p^2}} \vec\rb{\matMs_l}} \\
    & \quad - \infnorm{\rb{\matVsT_l \otimes \Identity_{4p^2}} \vec\rb{\matMs_l - \matMh_l}} \\
    & \leq \infnorm{ \matVh_l - \matVs_l} \onenorm{\matMh_l - \matMs_l} + \infnorm{ \matVh_l - \matVs_l} \onenorm{\matMs_l} \\
    & \quad + \infnorm{\matVs_l} \onenorm{\matMh_l - \matMs_l} \, .
\end{align*}
Then,
\begin{align*}
    \infnorm{\Identity_{4p^2} \otimes \matMh_1 \matVh_1 - \Identity_{4p^2} \otimes \matMs_1 \matVs_1} & \lesssim  \Vxhrate \Mxhrate + \Vxhrate \onenorm{\matMs_1} + \infnorm{\matVs_1} \Mxhrate \\
    \infnorm{\Identity_{4p^2} \otimes \matMh_2 \matVh_2 - \Identity_{4p^2} \otimes \matMs_2 \matVs_2} & \lesssim  \Vyhrate \Myhrate + \Vyhrate \onenorm{\matMs_2} + \infnorm{\matVs_2} \Myhrate \, .
\end{align*}
\paragraph{Convergence of $\matGh_l$ to $\matGs_l$.}
\begin{align*}
\infnorm{\matGh_l - \matGs_l} & = \infnorm{\matMh_l \matVh_l \matMhT_l - \matMs_l \matVs_l \matMsT_l} \\
    & = \infnorm{\rb{\Identity_{4p^2} \otimes \matMh_l \matVh_l} \vec\rb{\matMhT_l} - \rb{\Identity_{4p^2} \otimes \matMs_l \matVs_l} \vec\rb{\matMsT_l}} \\
    & \leq \infnorm{\rb{\Identity_{4p^2} \otimes \matMh_l \matVh_l - \Identity_{4p^2} \otimes \matMs_l \matVs_l} \vec\rb{\matMhT_l}} + \infnorm{\rb{\Identity_{4p^2} \otimes \matMs_l \matVs_l} \vec\rb{\matMhT_l - \matMsT_l}} \\
    & \leq \infnorm{\matMh_l \matVh_l - \matMs_l \matVs_l} \onenorm{\vec\rb{\matMhT_l - \matMsT_l}} + \infnorm{\matMh_l \matVh_l - \matMs_l \matVs_l} \onenorm{\vec\rb{\matMsT_l}} \\
    & \quad + \infnorm{\matMs_l \matVs_l} \onenorm{\vec\rb{\matMhT_l - \matMsT_l}} \, .
\end{align*}
Using the same simplification for $\infnorm{\matGh_2 - \matGs_2}$ and based on the rates above,
\begin{align*}
    \Gxhrate & = \infnorm{\matGh_1 - \matGs_1} \lesssim  \lrrb{\Vxhrate \Mxhrate + \Vxhrate \onenorm{\matMs_1} + \infnorm{\matVs_1} \Mxhrate} \lrrb{\Mxhrate + \onenorm{\matMs_1}} + \infnorm{\matMs_1 \matVs_1} \Mxhrate \\
    \Gyhrate & = \infnorm{\matGh_2 - \matGs_2} \lesssim  \lrrb{\Vyhrate \Myhrate + \Vyhrate \onenorm{\matMs_2} + \infnorm{\matVs_2} \Myhrate} \lrrb{\Myhrate + \onenorm{\matMs_2}} + \infnorm{\matMs_2 \matVs_2} \Myhrate \, .
\end{align*}

\paragraph{Convergence of plug-in variance estimator.} Now we are ready to show that our plug-in estimate of the variance is consistent. To show that $\asympSDh^2$ is consistent we will show that 
\begin{equation*}
    \lrabs{\asympSDh^2 - \asympSD^2} = o_p(1) \,.
\end{equation*}
We proceed by studying $\abs{\dbdhT \matGh_l \dbdh - \dbdsT \matGs_l \dbds}$,
\begin{align*}
    & \lrabs{\dbdhT \matGh_l \dbdh - \dbdsT \matGs_l \dbds} \\ 
    & \quad = \lrabs{\dbdhT \matGh_l \dbdh - \dbdsT \matGh_l \dbdh + \dbdsT \matGh_l \dbdh - \dbdsT \matGs_l \dbds} \\
    & \quad = \lrabs{ \rb{\dbdhT - \dbdsT} \matGh_l \dbdh + \dbdsT \matGh_l \dbdh - \dbdsT \matGs_l \dbds} \\
    & \quad = \lrabs{ \rb{\dbdhT - \dbdsT} \matGh_l \rb{\dbdh - \dbds} + \rb{\dbdhT - \dbdsT} \matGh_l \dbds + \dbdsT \matGh_l \dbdh - \dbdsT \matGs_l \dbds} \\
    & \quad \leq \lrabs{\rb{\dbdhT - \dbdsT} \matGh_l \rb{\dbdh - \dbds}} + \lrabs{\rb{\dbdhT - \dbdsT} \matGh_l \dbds} + \lrabs{\dbdsT \matGh_l \rb{\dbdh - \dbds}} + \lrabs{\dbdsT \rb{\matGh_l - \matGs_l} \dbds} \\
    & \quad \leq \onenorm{\dbdhT - \dbdsT} \infnorm{\matGh_l \rb{\dbdh - \dbds}} + 2 \onenorm{\dbdhT - \dbdsT}  \infnorm{\matGh_l \dbds} + \onenorm{\dbds} \infnorm{\rb{\matGh_l - \matGs_l} \dbds} \\
    & \quad \leq \onenorm{\dbdh - \dbds}^2 \infnorm{\matGh_l}  + 2 \onenorm{\dbdhT - \dbdsT}  \infnorm{ \rb{\matGh_l - \matGs_l} \dbds + \matGs_l \dbds} + \onenorm{\dbds}^2 \infnorm{\matGh_l - \matGs_l} \\
    & \quad \leq \onenorm{\dbdh - \dbds}^2 \infnorm{\matGh_l - \matGs_l} +  \onenorm{\dbdh - \dbds}^2 \infnorm{\matGs_l} + 2 \onenorm{\dbdh - \dbds} \infnorm{\matGh_l - \matGs_l} \onenorm{\dbds} \\
    & \qquad + 2 \onenorm{\dbdh - \dbds} \infnorm{\matGs_l \dbds} + \onenorm{\dbds}^2 \infnorm{\matGh_l - \matGs_l} \, .
\end{align*}
Comparing the rate for $\onenorm{\dbdh - \dbds}$ from \cref{lem:dbdh} with $\Gxhrate, \Gyhrate$ it can be seen that $\Gxhrate, \Gyhrate$ are slower and thus dominate the error rate. Closer examination of $\Gxhrate, \Gyhrate$ and their components reveal that $\Mxhrate, \Myhrate$, specifically $8p^2 \Sxhrate$ and  $8p^2 \Syhrate$, are the dominant terms. Recall from \cref{thm:BE:projWelch} that the effective sample size is $N = \effN_1 + \effN_2$ where $\effN_1 = T_1/B_1$, $\effN_2 = T_2/B_2$. Then, if $\effN_1 \asymp \effN_2$, and under \cref{assm}, by \cref{thm:B_rates}, we see that $\Sxhrate, \Syhrate$ have rates given by \cref{eq:Sxhrate,eq:Syhrate} with high probability. Plugging in these rates, and simplifying we conclude that 
\begin{align*}
    \lrabs{\dbdhT \matGh_1 \dbdh - \dbdsT \matGs_1 \dbds} & = O_p\lrrb{p^2 \sqrt{\frac{B_2}{T_2}} } \\
    \lrabs{\dbdhT \matGh_2 \dbdh - \dbdsT \matGs_2 \dbds} & = O_p\lrrb{p^2 \sqrt{\frac{B_1}{T_1}} } \, ,
\end{align*}
and
\begin{align*}
    \lrabs{\asympSDh^2 - \asympSD^2} & \leq \frac{\effN}{4\effN_1} \lrabs{\dbdhT \matGh_1 \dbdh - \dbdsT \matGs_1 \dbds} + \frac{\effN}{4\effN_2} \lrabs{\dbdhT \matGh_2 \dbdh - \dbdsT \matGs_2 \dbds} \\
    & = O_p\lrrb{p^2 \lrrb{\sqrt{\frac{B_2}{T_2}} + \sqrt{\frac{B_1}{T_1}}}} \\
    & = o_p(1) \, .
\end{align*}
 \end{proof}

\section{Useful lemmas} \label{app:useful_lemmas}
The following lemma from \citet{barber2018rocket} is useful for combining large deviation bounds with a Berry-Esseen bound. We can think of it as a finite-sample version of Slutsky's theorem.

\begin{lemma}[Lemma D.3 in \citet{barber2018rocket}] \label{lem:BE}
Suppose
\begin{equation*}
\sup_z \lrabs{\PP\lrcb{U \leq z} - \Phi(z)} \leq \varepsilon_U, \quad \PP\lrcb{|V| \leq \delta_V,\ |W| \leq \delta_W} \geq 1-\varepsilon_{VW}
\end{equation*}
for some $\delta_V, \delta_W, \varepsilon_U, \varepsilon_{VW} \in [0, 1)$. Then,
 \begin{equation*}
 \sup_z \lrabs{\PP\lrcb{\frac{U+V}{1+W} \leq z} - \Phi(z)} \leq \delta_V + \frac{\delta_W}{1-\delta_W} + \varepsilon_U + \varepsilon_{VW}.
 \end{equation*}
\end{lemma}

\section{Method details}\label{app:method_details}
In this section we discuss implementation details of our method. 

\subsection{Estimation}

We solve the D-trace optimization problem in \cref{eqn:sdd_estimator} using the \texttt{L1\_dts} solver from the \texttt{Difdtl} R-package. Because this package is no longer on CRAN, but is available on GitHub at SusanYuan/Difdtl (GPL ($\geq$ 2)), we incorporated the \texttt{L1\_dts} function directly into our \texttt{sdd} R-package.

To select the tuning parameter $\regD$ in \cref{eqn:sdd_estimator}, we choose the $\regD$ that minimizes eBIC \citep{foygel2010extended} which is computed as
\begin{equation}
\label{eqn:sdd_ebic}
    \operatorname{eBIC}(\Dh)_{\gamma} = \min(T_1,T_2)\infnorm{\frac{1}{2} \lrrb{\Sxh \Dh \Syh + \Syh \Dh \Sxh  - \Syh + \Sxh}} + \log(\min(T_1,T_2))\abs{E} + 4 \gamma \abs{E} \log(p) \, ,
\end{equation} where $|E|$ is the number of unique edges in $\hat{\Delta}$. Since $\hat{\Delta}$ represents the difference in expanded spectral densities $|E|$ is the number of non-zero entries in the upper triangular portions, including diagonals, of the submatrices $\hat{\Delta}_{1:p,1:p}$ and $\hat{\Delta}_{1:p,(p+1):2p}$ when $\freq \notin \cb{0,\pi}$. If $\freq \in \cb{0,\pi}$, only the submatrix $\hat{\Delta}_{1:p,1:p}$  is considered. For all applications and simulations in this work, we use $\gamma = 0.5$. We generate the range of penalty parameters $\cb{\regD}$ using the following procedure. We define $\cb{\regD }$ to be 20 values on a log-linear scale from $0.001*\regDmax$ to  $\regDmax$ where  $\regDmax$ represents the minimum value where all entries of $\Dh$ are 0. For the D-trace loss problem, $\regDmax = 2\max\rb{\abs{\Sxh - \Syh}}$. 

\subsection{Inference}

In the case of the standard D-trace loss estimating equation $\gradD$, we appeal to the asymptotic equivalence of the LASSO and Dantzig selectors from \citet{bickel2009simultaneous} and estimate the sparse de-biasing directions, $\dbdh$, using GLASSO tuned with eBIC with $\gamma = 0.5$ \citep{friedman2008sparse,foygel2010extended}. GLASSO was used due to the computational efficiency as for even small $p$, $\hat{\Hessian}$ is is too large for the CLIME procedure on a personal laptop. For example, for $p = 15$, $\hat{\Hessian} \in \reals^{900 \times 900}$ when $\freq \notin \cb{0,\pi}$. For the $\gradDL$ and $\gradDR$ estimating equations, $\Ssinv_l$ was estimated using CLIME \citep{cai2011constrained} with a fixed penalty of
\begin{equation*}
    \regH = \sqrt{\log(p)^2 T^{-1/4}} \sqrt{\frac{1}{(T\log(T))^{1/2}}} \, .
\end{equation*} 
The $\log(p)^2 T^{-1/4}$ in the first term comes from the convergence rate in \cref{thm:B_rates} with $B = \lceil T^{1/2} \rceil$ and assuming $\alpha_l = 1/2$. In this setting, the appropriate choice of $\regH$ according to \cref{lem:dbdh} is $O(\log(p)^2 T^{-1/4})$. For large $p$ and small $T$, it is possible this penalty is larger than $1$ which results in a CLIME solution that is all zeros. To remedy this, we use the square root of $\log(p)^2 T^{-1/4}$ and a factor of $\sqrt{1/(T\log(T))^{1/2}}$ to decrease the penalty parameter.

Next, we discuss the computational complexity of the asymptotic variance in greater detail. For our example, consider the case when $\freq \notin \cb{0,\pi}$ and we are interested in computing the asymptotic variance for the standard D-trace loss estimating equation $\gradD$. We will use the notation in the proof of \cref{lem:plugin_var_est} to represent various component matrices. From \cref{lem:plugin_var_est}, this can be written in the form
\begin{equation*}
    \asympSDh^2 = \frac{N}{4 N_1} \dbdhT \matMh_1 \matVh_1 \matMhT_1 \dbdh + \frac{N}{4 N_2} \dbdhT  \matMh_2 \matVh_2 \matMhT_2 \dbdh
\end{equation*}
where 
\begin{align*}
    \matMh_1 & = \Syh \Dh \otimes \Identity_{2p} + \Identity_{2p} \otimes \Syh \Dh + 2 \Identity_{2p} \otimes \Identity_{2p}, &  \matMh_2 & = \Sxh \Dh \otimes \Identity_{2p} + \Identity_{2p} \otimes \Sxh \Dh - 2 \Identity_{2p} \otimes \Identity_{2p} \\
    \matVh_1 & =  \Permutation^{\top} \hat{\VarSh}_1 \Permutation, & \matVs_2 & =  \Permutation^{\top} \hat{\VarSh}_2 \Permutation \, ,
\end{align*}
and
\begin{equation*}
    \hat{\VarSh}_l = \frac 1 2
\begin{bmatrix}
\lrrb{\Identity_{p^2} + \Commutation} \lrrb{\rSDh_l \otimes \rSDh_l + \iSDh_l \otimes \iSDh_l} &
\lrrb{\Identity_{p^2} + \Commutation} \lrrb{\iSDh_l \otimes \rSDh_l - \rSDh_l \otimes \iSDh_l} \\
\lrrb{\Identity_{p^2} - \Commutation} \lrrb{\rSDh_l \otimes \iSDh_l - \iSDh \otimes \rSDh_l} &
\lrrb{\Identity_{p^2} - \Commutation} \lrrb{\rSDh_l \otimes \rSDh_l + \iSDh_l \otimes \iSDh_l}
\end{bmatrix}.
\end{equation*}

The difficulty with estimating $\asympSDh^2$ lies in the fact that $\matMh_l, \matVh_l \in \reals^{4p^2 \times 4p^2}$ which are too large to form except when $p$ is small. However, since $\dbdh$ will be sparse in general, we will only need to compute specific entries of $\matGh_1 + \matGh_2 = \matMh_1 \matVh_1 \matMhT_1 + \matMh_2 \matVh_2 \matMhT_2$. Examining each component matrix in this sum shows that they are entirely determined by $\Sh_l$, $\Dh$, $\rSDh_l$, $\iSDh_l$ which are at most of dimension $2p \times 2p$. Thus, an $(i,j)$ entry can be computed using only these matrices as inputs. As an example, consider computation of  $\asympSDh^2$, and suppose $\dbdh^{\top} = \begin{bmatrix} c_1 & 0 & c_2 & 0 & \dots \end{bmatrix}$. That is, $\dbdh^{\top}$ only has two non-zero entries in the first and third positions. Then estimating $\asympSDh^2$ requires only four computations:
\begin{align*}
\asympSDh^2 & = \begin{bmatrix} c_1 & 0 & c_2 & 0 & \dots \end{bmatrix} \rb{\matGh_1 + \matGh_2} \begin{bmatrix} c_1 \\ 0  \\c_2 \\ 0 \\ \vdots \end{bmatrix} \\
& = c_1^2 \rb{\matGh_1 + \matGh_2}_{1,1} + c_1 c_2 \rb{\matGh_1 + \matGh_2}_{3,1} + c_1 c_2 \rb{\matGh_1 + \matGh_2}_{1,3} + c_2^2 \rb{\matGh_1 + \matGh_2}_{3,3} \, .
\end{align*}
In general, if $\dbdh$ has $\sparsity_{\dbdh}$ non-zero entries, then $\sparsity_{\dbdh}^2$ computations are needed. However, computing the entries of  $\rb{\matGh_1 + \matGh_2}_{i,j}$ is also non-trivial. Specifically, $\rb{\matGh_1 + \matGh_2}_{i,j}$ can be computed as
\begin{equation*}
\rb{\matGh_1 + \matGh_2}_{i,j} = \rb{\matMh_1}_{i,\cdot} \matVh_1 \rb{\matMh_1^{\top}}_{\cdot,j} + \rb{\matMh_2}_{i,\cdot} \matVh_2 \rb{\matMh_2^{\top}}_{\cdot,j} \, .
\end{equation*}
From the form of $\matMh_l$, we see that each row or column of $\matMh_l$ will have at most $4p - 1$ non-zero entries. It is worth noting that in $\gradDL$ and $\gradDR$ this matrix will have at most $2p$ non-zero entries. Thus, forming one half of $\rb{\matGh_1 + \matGh_2}_{i,j}$ requires $(4p-1)^2$ computations. Combining these observations, we see that, in general, $O\lrrb{16p^2 \sparsity_{\dbdh}^2}$ computations are needed to compute $\asympSDh^2$. Consider the case with $p = 100$ and $\sparsity_{\dbdh} = 100$ (since $\dbdh \in \reals^{4p^2}$ this is a very sparse $\dbdh$). Then, computing $\asympSDh^2$ would require $1.6 \times 10^{9}$ computations.  This is only to compute the variance of one de-biased parameter. If confidence intervals on many parameters are of interest, we will need to repeat this process for each parameter. 

In practice, we use these concepts to estimate the asymptotic variances using only the sparse projection estimators, e.g. $\dbdh$, $\dbdh_{0}$, $\dbdh_{1}$ and the matrices $\Sh_l$, $\Dh$, $\rSDh_l$ and $\iSDh_l$ as inputs. Computation of the asymptotic variances in this manner requires on-the-fly matrix multiplication and careful accounting of matrix indices. We implement these operations in \texttt{C++} by leveraging the \texttt{Rcpp} package \citep{Rcpp} in \texttt{R} \citep{Rcore}. To maximize speed, the Dantzig selectors as well as rows and columns of $\matMh_l$ matrices are stored as sparse vectors using \texttt{RcppEigen} \citep{RcppEigen}. 

As discussed in \cref{sec:simulations}, we can additionally reduce the computational complexity of the variance estimation by reducing the number of non-zero components in each of $\dbdh$ and $\matMh_l$. This is detailed in \cref{app:sim_results}.

\section{Simulation results}\label{app:sim_results}
In \cref{fig:sim_inference_full} we show the results for the full variance estimation without any restrictions to improve computational complexity. Compared to the computationally efficient variance estimation in \cref{fig:sim_inference_efficient}, the results in \cref{fig:sim_inference_full} are very similar although the Type I error is inflated for the efficient variance estimation even for large sample sizes.

The computational efficient variance estimation procedure drastically improves the computational complexity of variance estimation by reducing the number of non-zero components in each of $\dbdh$ and $\matMh$ in \cref{lem:plugin_var_est}. To reduce the number of non-zero components of $\dbdh$ we consider two restrictions. The first uses only the largest values explaining 99\% of the total sum of squares of $\dbdh$. The total sum of squares is used since the variance is computed as, e.g., $\dbdh^T S \dbdh$ for some matrix $S$. This reduces the number of non-zero values substantially. However, when $p$ is large, this may not be a sufficient restriction. To bound the worst-case computational complexity, we further cap $\dbdh$ to use only at most $500$ non-zero components for any $p$. Lastly, we reduce the number of non-zero components in the rows and columns of the scaling matrix $M$ from $2p$ to $\lfloor \sqrt{2p}\rfloor$.
\begin{figure}[t!]
    \centering
    \includegraphics[width=\textwidth]{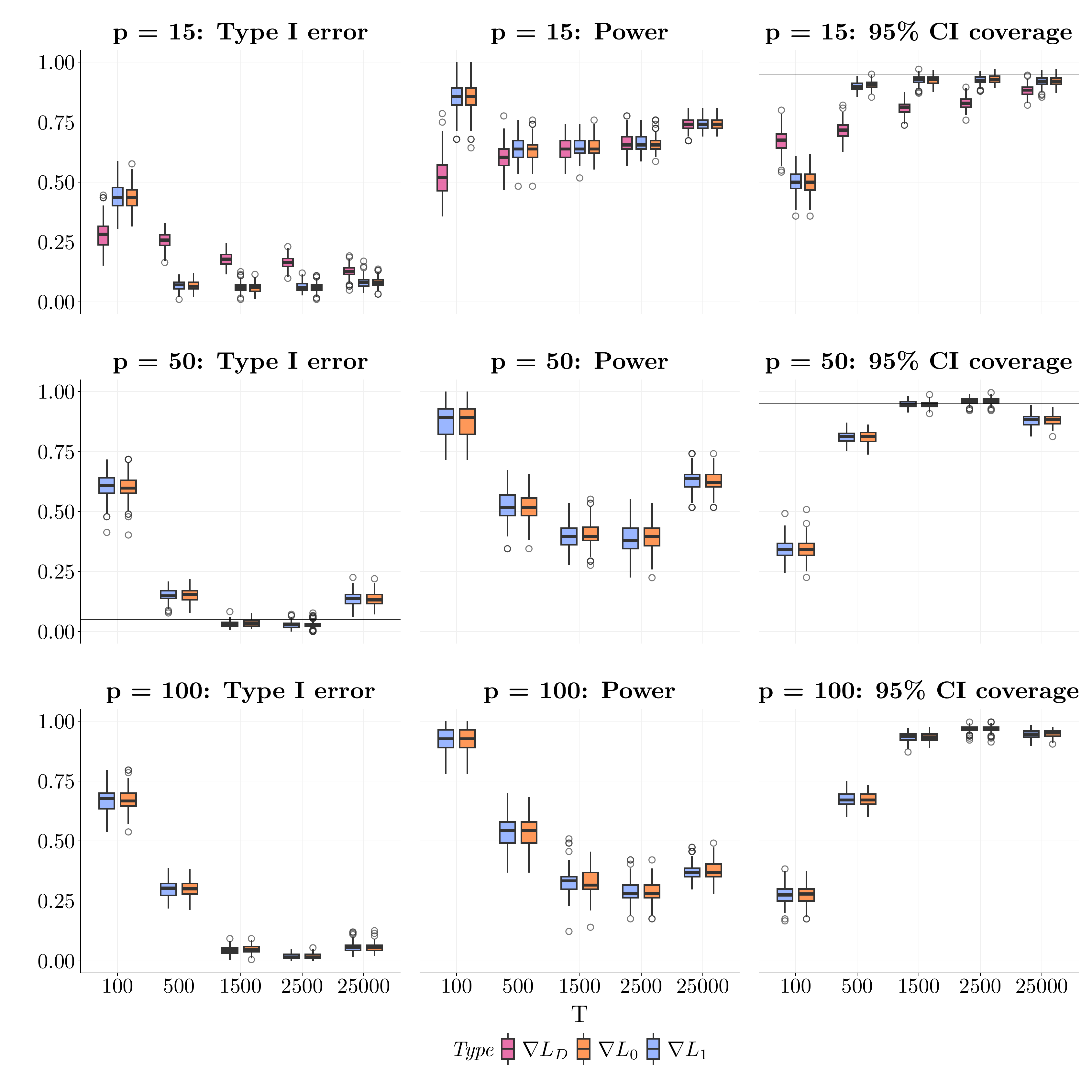} 
    \caption{Type I error, power, and coverage for different combinations of $p$ and $T$ using computationally efficient variance estimation.}
    \label{fig:sim_inference_efficient}
\end{figure}

\section{EEG analysis details}\label{app:eeg_details}
Below we describe additional details of the FGL method implementation. FGL has two penalty parameters, $\lambda_1, \lambda_2$ which control the sparsity of the individual inverses and the sparsity of their difference respectively. As sparse differences are of more interest in our application, we tuned FGL for a wider range of $\lambda_2$ values than $\lambda_1$. Specifically, we tuned FGL using AIC for two $\lambda_1$ values and 10 $\lambda_2$ values for a total of 20 combinations, the same number of penalty values considered for other methods. The values for $\lambda_1$ were $(0.01\lambda_{1,\max}, 0.1\lambda_{1,\max})$ where $\lambda_{1,\max} = \left\| \hat{\Sigma}_{1,O} + \hat{\Sigma}_{2,O} \right\|_{\infty}/2$ and $\hat{\Sigma}_{i,O}$ is the off-diagonals of $\hat{\Sigma}_i$. The $\lambda_2$ sequence was generated using 10 values on a log-linear scale from $0.0001* \lambda_{2,\max}$ to $\lambda_{2,\max}$ where $\lambda_{2,\max} = \max\left( \left\| \hat{\Sigma}_{1,O} \right\|_{\infty} -  0.01\lambda_{1,\max}, \left\| \hat{\Sigma}_{2,O} \right\|_{\infty} -  0.01\lambda_{1,\max} \right)$. 

\end{document}